\newif\ifrevtex
\revtextrue
\ifrevtex
  \documentclass[reprint,twocolumn,superscriptaddress,secnumarabic,amssymb,aps,prx]{revtex4-2}
\else
  \documentclass[10pt]{article}
\fi
\ifrevtex
  \usepackage[pass]{geometry}
\else
  \usepackage[margin=1in]{geometry}
\fi
\usepackage[T1]{fontenc}   %
\usepackage{bold-extra}    %
\DeclareMathAlphabet{\mathtt}{T1}{lmtt}{m}{n}

\newcommand{\panelfont}{\bfseries\fontsize{9}{10.8}\selectfont}

\usepackage{amsmath,amssymb,amsthm}
\usepackage{empheq}
\usepackage{graphicx}
\usepackage{caption}
\ifrevtex
\makeatletter
\DeclareCaptionJustification{flushjustified}{\leftskip\z@skip\rightskip\z@skip\parfillskip\z@ plus 1fil\relax}
\makeatother
\else
\fi
\usepackage[table]{xcolor}
\usepackage{tikz}
\usetikzlibrary{positioning,arrows.meta}
\usepackage{booktabs}
\usepackage{threeparttable}
\usepackage{enumitem}
\setlist[itemize,enumerate]{itemsep=0.15em,parsep=0pt}
\ifrevtex
\usepackage{float}
\makeatletter
\let\newfloat\newfloat@ltx
\let\float@end\float@end@ltx
\makeatother
\fi
\usepackage{algorithm}
\usepackage{algpseudocode}
\algrenewcommand\algorithmicindent{1em}
\makeatletter
\providecommand*{\theHALG@line}{}
\renewcommand*{\theHALG@line}{\thealgorithm.\arabic{ALG@line}}
\newcommand*{\toclevel@algorithm}{0}
\makeatother
\floatstyle{ruled}
\restylefloat{algorithm}
\ifrevtex
\makeatletter
\renewcommand*{\fnum@algorithm}{\fname@algorithm~\thealgorithm}
\makeatother
\fi
\usepackage{microtype}
\usepackage{tocloft}
\usepackage{titletoc}   %
\newcommand{\maintocdepth}[1]{\setcounter{tocdepth}{#1}}
\definecolor{deepslate}{HTML}{355070}
\definecolor{linkink}{HTML}{2E5C9E}
\definecolor{sageink}{HTML}{4E7D68}
\definecolor{rustink}{HTML}{9A4F2A}
\usepackage[colorlinks,linkcolor=linkink,citecolor=linkink,urlcolor=linkink,filecolor=black,bookmarksdepth=2]{hyperref}
\hypersetup{
  pdftitle={Discovering state-of-the-art quantum circuits with large language models},
  pdfauthor={Andi Gu, Alek Kemeny, Susanne F. Yelin, Mikhail D. Lukin}}
\usepackage{footnotehyper}
\makesavenoteenv{figure}
\ifrevtex\else
\usepackage[backend=biber,style=numeric-comp,sorting=none,natbib=true,maxbibnames=99]{biblatex}
\AtEveryBibitem{\iffieldundef{doi}{}{\clearfield{eprint}\clearfield{eprinttype}\clearfield{eprintclass}}}
\fi
\usepackage{braket}
\usepackage{listings}
\usepackage[capitalize,nameinlink,noabbrev]{cleveref}
\crefname{algorithm}{Algorithm}{Algorithms}
\Crefname{algorithm}{Algorithm}{Algorithms}
\crefname{appendix}{Appendix}{Appendices}
\Crefname{appendix}{Appendix}{Appendices}
\crefname{extendeddatatable}{Extended Data Table}{Extended Data Tables}
\Crefname{extendeddatatable}{Extended Data Table}{Extended Data Tables}
\crefname{extendeddatafigure}{Extended Data Figure}{Extended Data Figures}
\Crefname{extendeddatafigure}{Extended Data Figure}{Extended Data Figures}
\lstdefinestyle{solver}{
  language=Python,
  basicstyle=\footnotesize\ttfamily,
  commentstyle=\rmfamily\itshape\color{black!55},
  keywordstyle=\color{deepslate}\bfseries,
  stringstyle=\color{rustink},
  emph={ccx,cx,cz,hmr,free,alloc_qubit,alloc_qubits,free_vec,
        declare_qubit_register,checkpoint,mirror,emit},
  emphstyle=\color{sageink},
  numbers=left,
  numberstyle=\tiny\color{black!50},
  numbersep=6pt,
  xleftmargin=1.5em,
  columns=fullflexible,
  keepspaces=true,
  showstringspaces=false
}
\lstdefinestyle{solver-loop}{
  style=solver,
  basicstyle=\ttfamily\scriptsize,
  numbers=none,
  xleftmargin=0pt,
  xrightmargin=0pt,
  aboveskip=0pt,
  belowskip=0pt,
  frame=none,
  escapeinside={(*@}{@*)}
}

\definecolor{deltabg}{HTML}{EEF7F0}
\definecolor{deltafg}{HTML}{19693E}

\usepackage{refcount}
\makeatletter
\newcommand{\eqnumpin}[2]{%
  \@ifundefined{r@#1}{}{%
    \edef\kmx@pinned{#2}%
    \edef\kmx@printed{\getrefnumber{#1}}%
    \ifx\kmx@pinned\kmx@printed\else
      \PackageError{kickmix}{Stale figure: equation \detokenize{#1} now
        prints as (\kmx@printed), but (#2) is baked into
        figs/fig-rs-locator-recursion.pdf}{Regenerate that figure and
        rebuild figs/fig-rs-shared-update.tex; see the note by \string\eqnumpin
        in the preamble.}%
    \fi}}
\makeatother

\DeclareRobustCommand{\kickmix}{\textsc{kickmix}}

\definecolor{claudecoral}{HTML}{D97757}
\definecolor{ink}{HTML}{25313D}
\definecolor{panelblue}{HTML}{466791}
\definecolor{panelbluebg}{HTML}{F5F8FC}
\definecolor{panelgreen}{HTML}{417E63}
\definecolor{panelgreenbg}{HTML}{F5FBF7}
\definecolor{histred}{HTML}{FF0000}

\newsavebox{\loopboxA}
\newsavebox{\loopboxB}

\newtheorem{theorem}{Theorem}[section]
\newtheorem{lemma}[theorem]{Lemma}
\newtheorem{corollary}[theorem]{Corollary}
\newtheorem{proposition}[theorem]{Proposition}
\theoremstyle{remark}

\makeatletter
\renewcommand\paragraph{%
  \@startsection{paragraph}{4}{\z@}%
  {.8125ex plus .25ex minus .05ex}{-1em}%
  {\normalfont\normalsize\bfseries}}
\makeatother

\newcommand{\Tr}{\operatorname{Tr}}

\ifrevtex
\newenvironment{widefigure}[1][]{\begin{figure*}[t]}{\end{figure*}}
\newenvironment{widetable}{\begin{table*}[t]}{\end{table*}}
\else
\newenvironment{widefigure}[1][H]{\begin{figure}[#1]}{\end{figure}}
\newenvironment{widetable}{\begin{table}[H]}{\end{table}}
\newenvironment{widetext}{}{}
\newcommand{\onecolumngrid}{}
\renewenvironment{abstract}{%
  \par\vspace{0.8em}%
  \noindent\hspace*{0.11\textwidth}%
  \begin{minipage}{0.78\textwidth}\small
}{%
  \end{minipage}%
  \par\vspace{1.2em}%
}
\fi

\begin{document}

\ifrevtex
\title{Discovering state-of-the-art quantum circuits with large language models}
\author{Andi Gu}
\email{andigu@g.harvard.edu}
\affiliation{Harvard Quantum Initiative for Science and Engineering, Cambridge, MA 02138, USA}
\affiliation{Department of Physics, Harvard University, Cambridge, MA 02138, USA}
\affiliation{Anthropic STEM Fellows Program}
\author{Alek H. Kemeny}
\affiliation{Anthropic}
\author{Susanne F. Yelin}
\affiliation{Harvard Quantum Initiative for Science and Engineering, Cambridge, MA 02138, USA}
\affiliation{Department of Physics, Harvard University, Cambridge, MA 02138, USA}
\author{Mikhail D. Lukin}
\affiliation{Harvard Quantum Initiative for Science and Engineering, Cambridge, MA 02138, USA}
\affiliation{Department of Physics, Harvard University, Cambridge, MA 02138, USA}
\affiliation{QuEra Computing, 1380 Soldiers Field Road, Boston, MA 02135, USA}
\else
\begingroup
\renewcommand{\thefootnote}{\fnsymbol{footnote}}
\begin{center}
\vspace*{-3.2em}
{\Large\bfseries Discovering state-of-the-art quantum circuits with large language models\par}
\vspace{1em}
{\normalsize Andi Gu$^{1,2,3,}$\footnotemark[1], Alek Kemeny$^{4}$, Susanne F. Yelin$^{1,2}$, Mikhail D. Lukin$^{1,2,5}$\par}
\vspace{0.2em}
{\small\itshape
$^1$Harvard Quantum Initiative for Science and Engineering, Cambridge, MA 02138, USA\\
$^2$Department of Physics, Harvard University, Cambridge, MA 02138, USA\\
$^3$Anthropic STEM Fellows Program\\
$^4$Anthropic\\
$^5$QuEra Computing Inc., 1380 Soldiers Field Road, Boston, MA 02135, USA\par}
\end{center}
\footnotetext[1]{\href{mailto:andigu@g.harvard.edu}{andigu@g.harvard.edu}}
\endgroup
\fi

\begin{abstract}
Significant efforts are currently being directed toward the implementation of large-scale quantum information processors, but realizing useful quantum advantage also requires the development of efficient %
quantum algorithms. 
While for many practical problems the dominant resource costs are associated with
a small set of computational primitives,
 optimizing them in the quantum setting presents unique challenges, including a
\emph{reversibility} constraint absent from ordinary classical computation.
Here we introduce an evolutionary framework in which large language models (LLMs) propose circuit designs and an independent evaluator verifies and scores them.
We apply this framework to tasks spanning numerical computation, chemistry, optimization, quantum simulation, and number theory.
In seven of the eight tasks, the discovered circuits significantly reduce \emph{both} Toffoli and qubit counts relative to the state of the art.
In the eighth, the search discovers a new reversible architecture for prime-field modular inversion.
By handling reversibility differently, the architecture makes divsteps---a classical algorithm previously impractical for reversible integer arithmetic---viable
and substantially improves the Toffoli--qubit trade-off for modular inversion.
These results open the door for evaluator-guided LLM co-design of
new algorithmic structures for quantum computation.

\end{abstract}
\ifrevtex\maketitle\fi

Computational primitives based on classical function evaluation play a central
role in a wide range of quantum algorithms. Such primitives include decoding routines for
optimization~\cite{jordan2024dqi}, numerical functions for quantum
simulation~\cite{babbush2016exponentially}, arithmetic
routines for number theory~%
\cite{roetteler2017resource}, and quantum read-only memory
(QROM), which coherently loads classically specified data into quantum
registers~\cite{babbush2018encoding}. To operate on quantum data, these
functions must be compiled into circuits that can be executed coherently on quantum hardware. This imposes
two constraints absent from ordinary classical programs. First, because the
input may be in superposition, the circuit cannot inspect an intermediate
value and select an ordinary classical control-flow path. Data-dependent
behavior is instead expressed through controlled operations within a
fixed gate sequence, so every branch contributes to the circuit's cost even if
it applies only to some inputs. Second, the computation must be reversible:
intermediate information cannot simply be overwritten or discarded. Temporary
values must eventually be uncomputed, returning their qubits to their initial
state and preventing residual information from remaining entangled with the
output. Until they are uncomputed, these values occupy additional qubits. At the same time, when bounded error is permitted, machinery needed only for rare or atypical
inputs can sometimes be omitted, reducing resource requirements while keeping
the failure probability controlled.
Efficient coherent compilation is of major practical importance, and recent reductions in resource requirements have come
from coordinated changes to arithmetic algorithms, data representations, and
their reversible
implementation~\cite{chevignard2025reducing,gidney2025factor}. Such redesigns
are difficult for general quantum circuits because directly
testing correctness can require exponentially expensive classical simulation or fully operational large-scale quantum processors. Automated circuit
optimizers therefore typically preserve correctness by restricting the search
to predefined equivalence-preserving rewrites~\cite{nam2018automated},
limiting them to variations of a fixed design.

\looseness=-1
Here we introduce a framework that pairs large language model (LLM) proposals with an independent evaluator that verifies and scores each candidate.
We focus on a class of quantum circuits that compute a classical function $f$: %
for an arbitrary input
$\ket{\psi}=\sum_{x,y}\alpha_{x,y}\ket{x}\ket{y}$, the circuit
must act as $\ket{\psi}
  \longmapsto
  \sum_{x,y}\alpha_{x,y}
  \ket{x}\ket{y\oplus f(x)}$, where $\ket{x}$ and $\ket{y}$ are computational-basis
states of the input and output registers, respectively. By linearity, this coherent action is
correct if the circuit performs the required map on every computational-basis
input: $\ket{x}\ket{y} \longmapsto \ket{x}\ket{y\oplus f(x)}$.
This basis-state reduction makes candidate circuits amenable to programmatic
evaluation, since every
allowed operation maps a computational-basis state to another
computational-basis state, possibly changing its phase, but never creating a
superposition~\cite{babbush2026kickmix}. The action of a candidate on any
chosen input can therefore be simulated by tracking a bit string and its phase,
rather than $2^n$ amplitudes. Resource accounting is also simple: we use the
Toffoli count $T$ and the peak qubit count $Q$, and rank circuits by their product, the $TQ$-cost.
Together, these properties let us
build an \emph{efficient, programmatic} evaluator that checks the correctness and resource cost of any candidate circuit design.

\begin{widefigure}
\centering
\includegraphics[width=\linewidth]{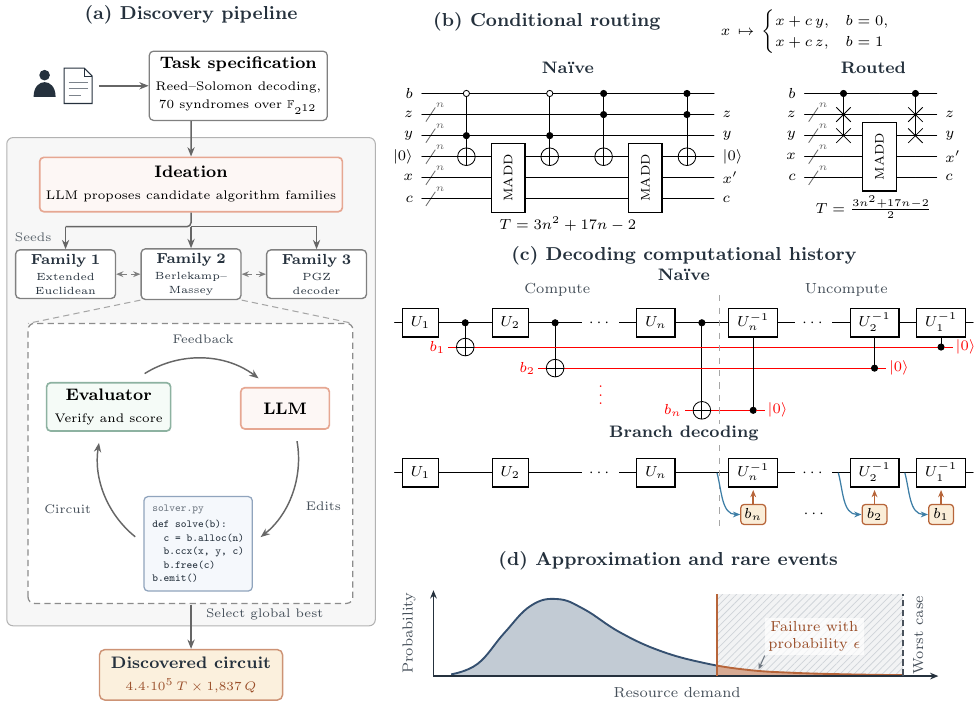}
\caption{Overview of the circuit-discovery framework. (a) A hand-written task specification defines the circuit design task (for example, Reed--Solomon decoding) and is given to an LLM. The LLM proposes architecturally different approaches (\emph{algorithm families}; here extended Euclidean, Berlekamp--Massey, and Peterson--Gorenstein--Zierler (PGZ) decoders), which are then developed and optimized in parallel. Strong candidates are periodically shared between families to allow for cross-pollination of useful ideas.
(b) Conditional routing reduces the cost of coherent control flow:
for instance, two controlled multiply-adds (MADDs) become
inexpensive controlled swaps and one shared arithmetic block. (c) Branch
decisions that a na\"ive reversible computation
stores during the forward pass (red) can instead be decoded from the current state
during uncomputation. (d) When an application permits bounded error, allowing rare failures on atypical, costly inputs reduces resource requirements.}
\label{fig:overview}
\end{widefigure}

Each search begins from a hand-written task specification and a blank
program; the specification states the required circuit behavior but provides
no implementation guidance. The LLM proposes a candidate and iteratively revises
it using evaluator feedback (\cref{fig:overview}(a)). Previous work
represents candidates as sequences of gates or predefined circuit
blocks~\cite{liang2023qgas,macarone2026circuit}, which can obscure
the algorithmic structure of the circuit. Our framework instead
lets the LLM propose and revise candidates through a high-level Python interface,
where reusable functions and ordinary programming constructs make that structure
explicit. This allows a single conceptual change to reorganize operations
throughout the circuit, while the evaluator tests its consequences for
correctness and resource cost. However, purely iterative refinement of a
single candidate can remain tied to the architecture from
which it began, resulting in convergence to local (but not necessarily global) optima. We therefore start with an ideation stage in which the LLM is prompted to propose several distinct \emph{algorithm families}, each
representing a different approach to the problem. Circuits within each family
are optimized independently in parallel, allowing the search to explore different strategies rather than
only variations of one design. Well-performing candidates are periodically shared
between families so that an idea discovered in one approach can inform
another.

Across eight diverse compilation tasks, the discovered circuits repeatedly
exhibit three common principles of algorithm--circuit co-design. These
principles address the features of
coherent compilation described above: expressing data-dependent control flow, making intermediate computations reversible
without retaining their entire history, and using a failure budget to remove circuit resources devoted to atypical cases (see \cref{fig:overview}).

\section*{Design principles}
\label{sec:principles-overview}

One recurring principle in the circuits discovered by the search is to replace
conditional arithmetic with \emph{conditional routing}. A na\"ive coherent
translation of a casewise classical algorithm includes a separately controlled
arithmetic block for every branch. The discovered circuits
instead use a representation in which the branches differ primarily in which
values occupy particular computational roles. Relatively cheap controlled swaps then
route the appropriate values into those roles, one unconditional arithmetic
path performs the expensive work, and further routing places the results in
their required registers. In this sense, where a classical program routes
execution to the selected operation, the coherent circuit routes the selected
data through a shared operation. The idea resembles bucket-brigade quantum random-access memory (QRAM)~\cite{giovannetti2008quantum}.

For example, consider four $n$-qubit fixed-point registers $c$, $x$, $y$, and
$z$, together with a control qubit $b$. Suppose that we wish to perform
\begin{equation}
  x\longmapsto
  \begin{cases}
    x+cy & \text{for } b=0,\\
    x+cz & \text{for } b=1,
  \end{cases}
  \label{eq:conditional-multiply-add}
\end{equation}
while returning $c$, $y$, and $z$ unchanged. A na\"ive implementation uses two
controlled multiply-adds. Conditional routing instead swaps $y$ and $z$
according to $b$, applies one uncontrolled multiply-add, and undoes the swap.
For the standard schoolbook construction~\cite{haner2018optimizing}, this
reduces the Toffoli count from
$T_{\mathrm{naive}}(n)=3n^2+17n-2$ to
$T_{\mathrm{routed}}(n)=(3n^2+17n-2)/2$, as shown in
\cref{fig:overview}(b).

More generally, conditional routing requires \emph{canonicalizing}
branch-dependent states into a representation on which the same expensive
computation applies. Schematically, when the branch can be expressed as $U_b = R_b^{-1} K R_b$, the circuit need only control the inexpensive routing transformation $R_b$;
the expensive kernel $K$ appears once and need not be controlled. Finding such
a decomposition typically takes more than recognizing identical arithmetic in two
branches: it exploits an equivalence
in the required output so that apparently different updates become instances
of the same computation. The difficulty lies less in the
routing circuit than in determining which
branch-dependent differences matter and which can be absorbed into a shared
representation.

Reversibility creates a second challenge: reversing a conditional update
requires knowing which branch was taken during the forward computation. A
na\"ive implementation records each decision in a history stack whose
size grows with the number of updates. But the branch decisions need not always
be stored explicitly: different branches may leave distinguishable patterns in
registers that remain available during uncomputation. As illustrated in \cref{fig:overview}(c), the circuit can recover
the branch bit from such a pattern, use it to select the appropriate inverse
update, and then clear the recovered bit, so the register contents serve as an
inexpensive record of the branch decisions. For example, suppose we have a register $g$ that begins at $g_0=0$. In round $j$, conditioned on some branch bit $b_j$, we add $2^j$, so that
\begin{equation}
  g_{j+1}=g_j+b_j2^j.
  \label{eq:binary-history-example}
\end{equation}
For arbitrary values of
$g_j$, there is no way to tell from $g_{j+1}$ which branch was taken. Globally, however, there is an important invariant: $g_j<2^j$, and therefore bit $j$ of $g_{j+1}$
must equal $b_j$. During uncomputation, the circuit can simply read this bit to recover $b_j$. The computation's
evolving state thus supplies the branch decisions in reverse order without a
separate history stack. In general, finding such a decoder requires the forward algorithm, its data
representation, and its uncomputation procedure to be designed together.
Viewed locally, an update may be ambiguous, since either branch can appear
to admit a valid predecessor. The ambiguity can disappear once the circuit is
restricted to states reachable from the prescribed initialization, where a
relation across the registers can distinguish histories that no single
register can.

The third principle exploits the small but finite failure probability permitted
by most quantum algorithms. Turning this failure budget into resource
savings, however, is not automatic. One must first identify the rare
inputs that determine worst-case resource requirements. Unlike a
classical program, a quantum circuit must contain the gates and ancillas
required by every supported path. A rare input can therefore determine
the circuit's Toffoli count or peak qubit count. Yet identifying such inputs is
not enough: the circuitry supporting them may not be separable from the
computation on typical inputs. The co-design task is to reorganize the
computation so that the additional resources required by rare inputs are
localized---for example, to extra iterations, stack entries, or arithmetic
lanes. Those components can then be omitted while confining failure to the
low-probability inputs that require them
(\cref{fig:overview}(d)). In what follows we show how these design principles result in significantly improved algorithmic circuits across a wide variety of quantum tasks.

\section*{Reversible syndrome decoding for DQI}
\label{sec:dqi-principles}

Decoded quantum interferometry (DQI) is a quantum algorithm that turns certain
optimization problems into decoding problems for classical error-correcting
codes, using quantum interference to favor better solutions. A central
application is optimal polynomial intersection (OPI), a form of polynomial
fitting over finite fields for which DQI offers a superpolynomial speedup over
known classical algorithms. DQI achieves this by using a Reed--Solomon decoder
as a subroutine within the quantum computation~\cite{jordan2024dqi}. The decoder identifies the positions and values of errors from a collection
of algebraic checks called \emph{syndromes}. DQI requires this decoding to run
coherently on a superposition of inputs, making the reversible decoder a
dominant resource cost. \citet{khattar2025optimized-dqi} compiled DQI circuits
for OPI over binary extension fields using reversible Reed--Solomon decoders
based on the extended Euclidean algorithm.

We study their primary Reed--Solomon instance, in which 70 syndrome values
uniquely determine an error pattern containing at most 35 errors. Our search
framework discovered a circuit that replaces their synchronized Euclidean
solve with the Berlekamp--Massey (BM) recurrence. BM substantially reduces the
arithmetic cost, but its updates overwrite intermediate values that a
na\"ive reversible implementation would retain in a large history
stack. The discovered circuit preserves BM's arithmetic advantage while
avoiding most of this storage (\cref{fig:rs-shared-update}).

\begin{widefigure}
\centering
\includegraphics[width=\linewidth]{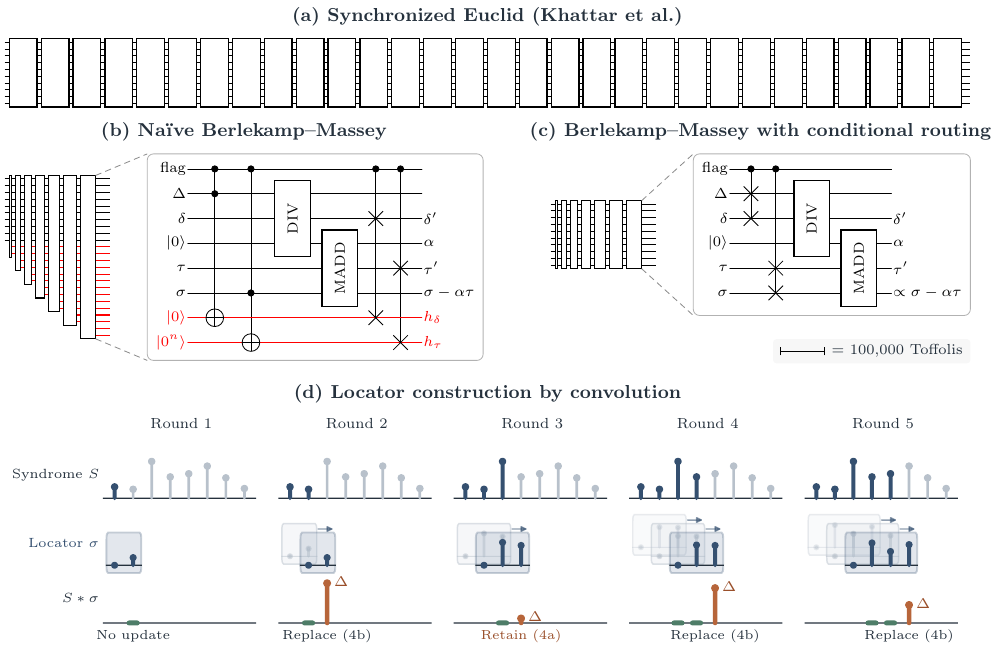}
\caption{Reversible Reed--Solomon decoding. (a) The synchronized extended-Euclidean construction of
\citet{khattar2025optimized-dqi} requires approximately $1.9$ million
Toffolis for the locator solve. (b) Berlekamp--Massey reduces this cost to approximately
$140{,}000$ Toffolis but na\"ively accumulates intermediate information (red). DIV forms $\alpha=\Delta/\delta$, MADD applies the multiply-add, and the flag marks the replace case.
(c) Conditional routing preserves this arithmetic advantage while avoiding
most history storage. Block widths are proportional to Toffoli count.
(d) Locator construction as convolution with the syndrome sequence. The
discrepancy $\Delta$ is the convolution at the next syndrome; a nonzero value
requires updating the locator.}
\label{fig:rs-shared-update}
\end{widefigure}
\IfFileExists{figs/fig-rs-locator-recursion-refs.tex}{\eqnumpin{eq:rs-bm-retain-earlier}{4a}
\eqnumpin{eq:rs-bm-retain-current}{4b}
}{}

Reed--Solomon decoding amounts to determining which positions in the
data contain errors. The $L$ error positions are encoded in a locator
polynomial $\sigma(z)=\sum_{k=0}^{L}\sigma_kz^k$, whose roots identify those
positions. To recover the locator coefficients, BM solves for the lowest-degree
nonzero polynomial whose convolution (see \cref{fig:rs-shared-update}(d)) with the syndrome sequence vanishes:
\begin{equation}
  (S*\sigma)_r=\sum_{k=0}^{L}\sigma_kS_{r-k}=0,
  \qquad r=L+1,\ldots,70.
\end{equation}
BM constructs the
locator incrementally. At each round, the convolution evaluated at the next
syndrome gives the discrepancy $\Delta$: if it vanishes, no correction is
needed; otherwise, BM adjusts the locator coefficients and, when necessary,
increases its degree.

To make a correction without invalidating previous rounds of updates, BM retains an earlier polynomial $\tau$ as a correction template, together with its nonzero discrepancy $\delta$. When $\Delta\ne0$, BM scales $\tau$ by
$\alpha=\Delta/\delta$ and applies the correction $\sigma \gets \sigma - \alpha \tau$. BM keeps the template until the current locator has accumulated enough information from subsequent syndromes to take over its role:
\begin{subequations}\label{eq:rs-bm-casewise-overview}
\begin{empheq}[left={(\sigma,\tau,\delta)\gets\empheqlbrace}]{alignat=2}
  &(\sigma-\alpha\tau,\tau,\delta)
    &&\quad\text{(retain)},
    \label{eq:rs-bm-retain-earlier}\\
  &(\sigma-\alpha\tau,\sigma,\Delta)
    &&\quad\text{(replace)}.
    \label{eq:rs-bm-retain-current}
\end{empheq}
\end{subequations}
The replace update in \cref{eq:rs-bm-retain-current} creates the main reversibility problem: it overwrites the previous template $\tau$ and discrepancy $\delta$. A na\"ive reversible implementation stores these displaced values so that the update can later be undone. Over repeated rounds, this accumulates a history of whole polynomials (\cref{fig:rs-shared-update}(b)).

The discovered circuit avoids this polynomial history by exploiting a freedom in the decoding task: the locator is needed only up to a nonzero multiplicative factor, because rescaling leaves its roots unchanged. The
circuit (\cref{fig:rs-shared-update}(c)) therefore implements a rescaled version of both branches of
\cref{eq:rs-bm-casewise-overview}, using simple conditional routing and a single
shared update:
\begin{subequations}\label{eq:rs-bm-rescaled-overview}
\begin{empheq}[left={(\sigma,\tau,\delta)\gets\empheqlbrace}]{alignat=2}
  &(\sigma-\alpha\tau,\tau,\delta)
    &&\quad\text{(retain)},
    \label{eq:rs-bm-rescaled-retain}\\
  &(\tau-\alpha^{-1}\sigma,\sigma,\Delta)
    &&\quad\text{(replace)}.
    \label{eq:rs-bm-rescaled-replace}
\end{empheq}
\end{subequations}
The result
$\tau-\alpha^{-1}\sigma=-\alpha^{-1}(\sigma-\alpha\tau)$
differs from the locator in the unrescaled update only by the nonzero scale factor
$-\alpha^{-1}$ and so has the same roots. Both branches can then use one shared
arithmetic path without storing the displaced template $\tau$.

\looseness=-1
This resolves the history problem for the polynomial registers, but the
discrepancies present a second reversibility problem. To uncompute a correction,
the circuit must recover whichever discrepancy is no longer present after
the update. The replace case of \cref{eq:rs-bm-rescaled-replace} overwrites the stored
value $\delta_{\mathrm{old}}$ with $\delta_{\mathrm{new}}=\Delta$, but the circuit can \emph{decode}
$\delta_{\mathrm{old}}$ from the constant coefficients of the updated locator registers. BM shifts the template by one power of $z$ every round, so $\tau$ has zero constant coefficient entering the round. Writing
$\sigma_0$ and $\tau_0$ for the constant coefficients after the update gives $\delta_{\mathrm{old}}=\sigma_0\tau_0^{-1}\delta_{\mathrm{new}}$.
This case thus adds no entry to the history
stack. In fact, the circuit never stores $\delta$ at all: the same relation fixes $\delta$ in terms of the locator's constant coefficient (Supplementary Information [SI]~\ref*{sec:case-dqi-rs}). The retain update of \cref{eq:rs-bm-rescaled-retain} is the exception: it keeps the previous discrepancy $\delta$, while the current discrepancy $\Delta$, which is needed to reverse the correction, remains only in a temporary register. Each retain update therefore stores $\Delta$ on the history stack. In the worst case, 35 stack entries (665 qubits) would be needed. Under the DQI input distribution, a second stack entry is required with probability $1.0\times10^{-4}$, so the discovered circuit allocates only one stack entry (19 qubits instead of 665) and permits failure on these rare inputs.

Together, these transformations make the arithmetically efficient BM algorithm suitable for reversible~\mbox{computation}.

\section*{Reversible modular inversion}
\label{sec:modinv-principles}

Modular inversion is the task of finding $x$ such that $wx\equiv1\pmod p$,
given a nonzero input $w$ and a prime $p$. It is a central arithmetic subroutine
in a range of quantum algorithms~\cite{vandam2002gausssums,roetteler2017resource}.
A standard way to compute the inverse is through the extended Euclidean algorithm.
Euclidean algorithms repeatedly simplify a pair of integers while preserving
their greatest common divisor. By tracking how this divisor is expressed in
terms of the original inputs, they also recover coefficients $\alpha,\beta$ satisfying
$\alpha p+\beta w=1$, so $\beta$ is the inverse of $w$ modulo $p$.

Divsteps is a variant of the Euclidean algorithm that replaces division with
conditional additions, subtractions, and halving~\cite{bernstein2019divsteps}.
Starting from $f=p$ and $g=w$, the goal is to reduce the pair to $f=\pm1$ and
$g=0$ while preserving their greatest common divisor. The basic strategy is to
hold the odd integer $f$ fixed while reducing $g$: if $g$ is even, it can be
halved directly; if $g$ is odd, adding or subtracting $f$ first makes it even.
These operations preserve the greatest common divisor because $f$ is odd. Holding $f$
fixed indefinitely, however, would leave only one integer being reduced, so
the algorithm must periodically swap their roles. A small signed counter
$z$ schedules these swaps without ever comparing $f$ and $g$: it roughly
tracks how many more bits $f$ has than $g$, so a swap occurs
once $g$ is likely the smaller of the two. This allows
both integers to be reduced using only inexpensive arithmetic and control. Two auxiliary integers $f_v$ and $g_v$ track the input's contribution in a scaled
form: after $t$ steps, they satisfy $f_vw\equiv2^t f\pmod p$ and
$g_vw\equiv2^t g\pmod p$. Thus, at termination, $ff_v$ gives the desired inverse
up to the scaling factor $2^t$. The scaling keeps the coefficient updates
free of division. Starting from
$(z,f,g,f_v,g_v)=(0,p,w,0,1)$, each step applies one of three updates:
\begin{widetext}
\begin{subequations}\label{eq:divstep-recurrence}
\begin{empheq}[left={(z',f',g',f_v',g_v')=\empheqlbrace}]{alignat=2}
  &(-z,\ g,\ (g-f)/2,\ 2g_v,\ g_v-f_v)
    &&\quad\text{for } z\ge0,\ g\text{ odd},
    \label{eq:divstep-swap}\\
  &(z+1,\ f,\ (g+f)/2,\ 2f_v,\ g_v+f_v)
    &&\quad\text{for } z<0,\ g\text{ odd},
    \label{eq:divstep-odd-nonswap}\\
  &(z+1,\ f,\ g/2,\ 2f_v,\ g_v)
    &&\quad\text{for } g\text{ even}.
    \label{eq:divstep-even}
\end{empheq}
\end{subequations}
\end{widetext}

\looseness=-1
Classically, each update overwrites the current state.
Reversing the computation, however, requires knowing which of the three updates was applied.
A na\"ive
implementation would store this choice at every step, accumulating a history
stack alongside the arithmetic registers. The discovered circuit instead
recovers decisions from the current arithmetic state. Some decisions are easy
to recover from the updated registers, while others require a more involved
decoding procedure. We label each update by two bits: $s$ records whether the step swaps the roles
of $f$ and $g$, and $b$ records whether the incoming $g$ is odd. The updates in
\cref{eq:divstep-swap,eq:divstep-odd-nonswap,eq:divstep-even} have labels $(s,b)=(1,1),(0,1),(0,0)$, respectively. The task
during uncomputation is to recover these bits from the updated state.

\begin{widefigure}[t]
\centering
\includegraphics[width=\linewidth]{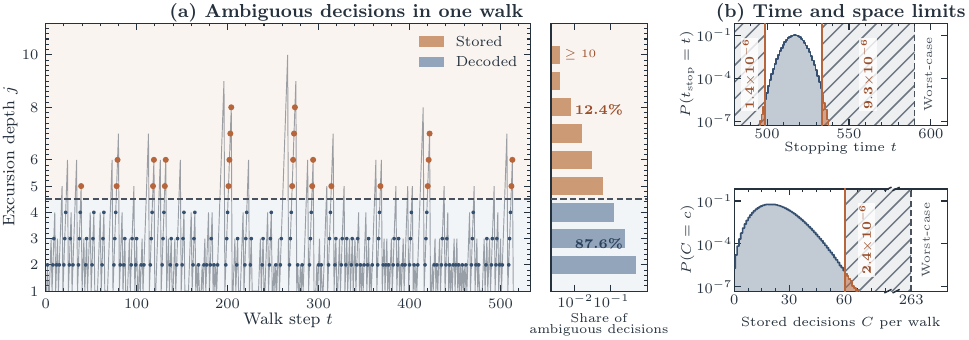}
\caption{Making divsteps reversible with minimal overhead.
(a) Excursion depth $j$ in one representative walk. Dots
mark ambiguous decisions. We use a cutoff $j_{\mathrm{store}}=5$, so all ambiguous decisions with $j \ge 5$ are stored rather than decoded. The histogram shows the distribution of excursion depths over uniformly sampled walks: only 12.4\% of ambiguous decisions reach $j \ge 5$. (b) Empirical stopping-time and
history-size distributions. The discovered circuit assumes that every walk terminates between steps $499$ and $533$; this assumption fails with probability $\sim 10^{-5}$. It also allocates enough qubits to store at most 60 branch decisions, which is sufficient for all but a $2.4\times10^{-6}$ fraction of walks.}
\label{fig:divstep}
\end{widefigure}

For many steps, recovering $b$ and $s$ requires only a few observations about
the updated registers. A swap (\cref{eq:divstep-swap}) always leaves a distinctive pattern in $f_v'$.
To see this, note that throughout the recurrence, $f_v$ is even and $g_v$ is
odd. A swap sets $f_v'=2g_v$, which has exactly one trailing zero bit, whereas
a nonswap sets $f_v'=2f_v$, which is divisible by four. Thus the updated
coefficient identifies every swap, recovering both $s=1$ and $b=1$; otherwise,
$s=0$, so the swap bit $s$ is \emph{always} cheaply recoverable. For nonswaps, two further tests determine $b$. If $z'\ge1$, the incoming $g$
must have been even, since an odd value would have triggered a swap. If the
step first reaches $g'=0$, the incoming $g$ must instead have been odd, since
halving a nonzero even integer cannot give zero. Together, these simple tests
recover about two thirds of the decisions in a typical walk (the sequence of divsteps applied to one input). The only
unresolved cases are nonterminating nonswaps with $z'\le0$, where the circuit
must determine whether the step added $f_v$ to $g_v$, that is, the bit $b$ in the update $(f_v',g_v')=(2f_v,g_v + b f_v)$.

Consider a swap followed by a sequence of these ambiguous nonswaps, which we
call an \emph{excursion}. Immediately after the swap, $f_v=2q$ for some odd $q$.
On the $j$th step of the excursion, $f_v=2^{j-1}q$ (hence $f_v'=2^j q$),
where the \emph{excursion depth} $j$ counts the initiating swap as step 1.
Equivalently, $j$ is the number of trailing zero bits in the updated
coefficient $f_v'$. The update to $g_v$ becomes
\begin{equation}
  g_v'=g_v+b2^{j-1}q.
\end{equation}
Note that this almost matches the binary-record construction in
\cref{eq:binary-history-example}. The only difference is the shared factor $q$: to recover $b$, we need to remove the factor $q$ from the update. Since $f_v'=2^j q$ is known, we can recover $q$ by removing the trailing zero bits from $f_v'$. To remove the same factor
from both sides of the update, we multiply by $q^{-1}$ modulo $2^j$; this inverse exists because $q$ is odd. The update rule then becomes
\begin{equation}
  g_v'q^{-1}\bmod2^j= (g_v q^{-1}\bmod2^j) +b2^{j-1}.
  \label{eq:divstep-history-record}
\end{equation}
Crucially, in \cref{app:divstep-recovery}, we show that $(g_v q^{-1}\bmod2^j)<2^{j-1}$. Therefore,
\begin{equation}
  b=\operatorname{bit}_{j-1}\!\left(g_v'q^{-1}\bmod2^j\right),
  \label{eq:divstep-decoding-rule}
\end{equation}
where bit positions are counted from zero at the least significant (rightmost) bit.
Every quantity on the right-hand side is determined by the updated registers, so the rule
recovers the branch decision without access to the previous state. This recovery
relies on a global invariant ensuring that earlier branch history cannot
contaminate the deciding bit. The LLM identified this invariant, which links
information accumulated across successive updates to the branch
decisions needed during uncomputation, and developed it into a
proof that the arithmetic state retains the information required for
recovery~\cite{gu2026lean}.

For a given $j$, the rule in \cref{eq:divstep-decoding-rule} depends only on
the lowest $j$ bits of $g_v'$ and
the lowest $j$ bits of $q$. Therefore, for each small value of $j$, we treat the
decoding rule as a truth table mapping $2j$ input bits to a single bit $b$,
and compile it directly into a Boolean circuit. For larger $j$, the decoder must inspect more bits
and becomes increasingly expensive. Although every decision is recoverable in
principle, decoding all of them does not give the cheapest circuit. A tunable cutoff $j_{\mathrm{store}}$ determines which ambiguous decisions
are decoded and which are stored. At depths $j<j_{\mathrm{store}}$, the circuit
uses a decoder; at depths $j\ge j_{\mathrm{store}}$, it writes the decision to
a history stack. Increasing the cutoff spends more Toffolis on decoding
while reducing the required history storage, giving a family of circuits that trade Toffoli count against qubit count
(SI~\ref*{sec:case-modinv}).

Further savings are possible when the application permits a small failure
probability. We use this failure budget to exclude rare, unusually demanding
inputs through two resource limits: a fixed history-stack size and a fixed
window of supported stopping times. The history-stack limit builds on the decoding trade-off. With
$j_{\mathrm{store}}=5$, only $12.4\%$ of ambiguous decisions are written to the stack
(\cref{fig:divstep}(a)). On average, the required stack size is 21 history qubits, but it can reach 263 in the worst
case. We instead reserve 60 and allow the circuit to fail on rare
inputs for which the stack overflows. Such overflows occur with probability
$2.4\times10^{-6}$ per inversion (\cref{fig:divstep}(b)), saving 203 history qubits relative
to the worst case. The stopping-time limit addresses a separate cost. A fixed quantum circuit
must schedule enough steps for every input it supports. Although 590 divsteps
suffice even in the worst case, stopping times are concentrated in a much
narrower range (\cref{fig:divstep}(b)). The discovered circuit therefore
schedules 533 steps and tests for termination only from step 499 onward;
inputs that terminate outside this window are allowed to fail.

Branch decoding and bounded resource allocation thus make divsteps
practical for reversible modular inversion. A complete benchmark circuit built on this inversion
(SI~\ref*{sec:case-modinv}) uses 4.0 million Toffolis and 935 qubits.

\section*{Discussion and outlook}\label{sec:discussion}

The DQI and modular inversion constructions illustrate how redesigning the
representation and uncomputation of an algorithm can substantially reduce its
quantum resource requirements. Similar gains arise across the other tasks studied here. For seven of the
eight tasks, the discovered
circuit uses both fewer Toffolis and fewer qubits than the corresponding baseline.
\Cref{tab:main-results} summarizes these improvements, and the \hyperref[sec:si]{Supplementary Information} gives details for each task.

\begin{widetable}
\caption{Toffoli--qubit results across the eight compilation tasks, comparing baselines against our discovered circuits (see \cref{tab:results} for details); $G_{TQ}=(TQ)_{\mathrm{baseline}}/(TQ)_{\mathrm{ours}}$. Where a baseline contains free parameters, we choose them to minimize its $TQ$-cost. HHL: Harrow--Hassidim--Lloyd.}
\label{tab:main-results}
\centering
\begin{threeparttable}
\begin{tabular}{lcc>{\columncolor{deltabg}\color{deltafg}\bfseries}c}
\toprule
\textbf{Task} & \textbf{Baseline} $\boldsymbol{(T,Q)}$ & \textbf{Ours} $\boldsymbol{(T,Q)}$ &
\multicolumn{1}{>{\columncolor{deltabg}}c}{$\boldsymbol{G_{TQ}}$} \\
\midrule
Reed--Solomon DQI decoder &
$(5.1{\cdot}10^6,\,1{,}938)$ &
$(4.4{\cdot}10^5,\,1{,}837)$ & $12$ \\
Hermitian DQI decoder &
$(1.0{\cdot}10^8,\,1.9{\cdot}10^5)$\tnote{\dag} &
$(9.3{\cdot}10^5,\,2{,}395)$ &
\multicolumn{1}{>{\columncolor{deltabg}\color{black}}c}{---\tnote{\dag}} \\
HHL rotation angle &
$(1.6{\cdot}10^4,\,300)$ &
$(1{,}348,\,115)$ & $32$ \\
Coulomb potential &
$(6{,}355,\,252)$ &
$(1{,}140,\,142)$ & $9.9$ \\
Fixed-point Gaussian &
$(2.4{\cdot}10^4,\,362)$ &
$(1{,}583,\,139)$ & $39$ \\
Fe$_2$S$_2$ alias--\textsc{Select} &
$(5.5{\cdot}10^4,\,139)$ &
$(1.5{\cdot}10^4,\,126)$ & $4.1$ \\
Nonabelian lattice-gauge simulation &
$(1.3{\cdot}10^6,\,1{,}927)$ &
$(1.7{\cdot}10^5,\,1{,}859)$ & $7.7$ \\
Prime-field modular inversion &
---\tnote{\ddag} &
$(4.0{\cdot}10^6,\,935)$ &
\multicolumn{1}{>{\columncolor{deltabg}\color{black}}c}{---} \\
\bottomrule
\end{tabular}
\begin{tablenotes}[flushleft]
\item[\dag] Na\"ive reversible implementation; no prior reversible compilation exists in the literature.
\item[\ddag] Costs are for the complete benchmark circuit; see SI~\ref*{sec:case-modinv} for the benchmark and prior constructions.
\end{tablenotes}
\end{threeparttable}
\end{widetable}

A recurring source of these gains is the recovery of classical algorithmic
efficiency under the constraints of quantum computation. In a na\"ive reversible
implementation, coherent control and cleanup can introduce enough overhead to
obscure the advantages of an efficient classical method. Our Berlekamp--Massey
and divstep constructions show how jointly choosing the algorithm,
representation, and uncomputation can recover those advantages. These results suggest that some
apparent costs of quantum arithmetic reflect limitations of existing reversible
constructions rather than of the algorithms themselves.

\looseness=-1
Finding such constructions requires reasoning about choices whose benefits
depend on one another. An arithmetic saving may become useful only after the
associated control and uncomputation costs have also been reduced. A candidate
design must therefore be carried through several
unresolved mathematical and
implementation problems before its resource advantage becomes apparent.

The framework is bounded by the choices that define each
task. One is the objective that ranks circuits: we use the $TQ$-cost
for a uniform comparison, but other objectives are equally reasonable. In practice, the relative importance of qubit count, depth, routing, and measurements
depends on details such as hardware constraints, decoder performance, and error-correction architecture. The scoring function can be
replaced with a hardware-appropriate combination of these resources, but
determining that objective remains a matter of application and hardware judgment. The circuit interface defines a second boundary. Here, the required
transformation is specified in a fixed computational basis, excluding
optimizations whose advantage comes from changing the encoding itself. For
example, coset representations can make modular addition cheaper by encoding
each residue as a superposition of congruent integers, which avoids conventional
comparisons and corrections at the cost of an exponentially small
error~\cite{zalka2006fewer,gidney2019approximate}, but such encodings lie outside
the present framework. More generally, the task boundary, input distribution, data representations, and accuracy target make a
search well posed while also encoding assumptions about how the circuit is
embedded in a larger algorithm (in particular, which states it receives and how its
outputs are consumed). %

\looseness=-1
The same choices that bound the search also make reliable programmatic feedback
possible: they specify what each candidate must preserve and which resource
trade-offs the evaluator should reward. Within these boundaries, our results
show that evaluator-guided LLM search can go beyond refining a fixed circuit
to uncover new reversible algorithms and representations. Related
approaches have produced new algorithms in combinatorics and linear
algebra~\cite{romera2024funsearch,novikov2025alphaevolve}. These observations suggest a
practical criterion for extending the approach: wherever a problem admits a
sufficiently expressive task interface and its correctness and resource costs can
be evaluated cheaply enough to sustain many rounds of search, the same evaluator-guided search
may provide a route to new algorithms.

These results, together with 
related progress in reinforcement learning and LLM-guided discovery of quantum
error-correcting codes and encoding
circuits~\cite{su2025quantum-codes,olle2024noise-aware,yan2026omniqec}, create an opportunity for end-to-end
co-design of fault-tolerant computation. Choices that are usually made
separately---the algorithm and its logical circuit, the error-correcting code
and its method for universal logic, the supply of magic states, and the physical
execution architecture---strongly determine one another's costs. A programmatic evaluator that propagates logical errors, resource
demand, and hardware requirements across these layers could allow an LLM search to
redesign them together, potentially discovering fault-tolerant algorithms and architectures
that cannot be found by optimizing each layer in isolation.

\section*{Acknowledgments}
We thank Pablo Bonilla, Arjun Kar, Nazli Ugur \mbox{Koyluoglu}, Rohan Mehta, Varun Menon, and Vinay Ramasesh for
feedback on the manuscript. AG conducted part of this work under the Anthropic STEM
Fellows Program.
Work at Harvard was supported by IARPA and the Army Research Office under the Entangled Logical Qubits program (Cooperative Agreement Number \mbox{W911NF-23-2-0219})
and by the National Science Foundation (NVQL grant PHY-2410716).

\clearpage
\onecolumngrid

\section*{Methods}\label{sec:platform}

A classical function cannot generally be applied directly within a quantum
algorithm, because quantum evolution must be reversible. Instead, a function
$f$ is embedded in a reversible map $\ket{x}\ket{y}\longmapsto\ket{x}\ket{y\oplus f(x)}$.
A natural model for these subroutines is the \kickmix{} circuit model of
\citet{babbush2026kickmix}, which augments reversible classical computation
with measurement-based uncomputation (MBU) and diagonal sign gates.
Despite these additional operations, \kickmix{} circuits remain efficiently
classically simulable. This allows candidate circuits to be tested for their
action on coherent quantum superpositions, providing the reliable correctness
signal needed for evolutionary search.

Circuits that implement the same specification can have very different time
and space requirements. We use Toffoli count and peak qubit count as
proxies for these costs. Toffoli gates require expensive
non-Clifford resources, so their count tracks a
large part of the logical work and often the
runtime~\cite{fowler2012time,gidney2019autoccz}. The peak qubit count measures the
data and ancillas that must coexist during execution and therefore
approximates the required logical processor
footprint~\cite{beverland2022assessing}. The two resources also typically trade against
each other: additional ancillas can avoid recomputation, while recycling
ancillas can reduce space at the cost of additional gates. We multiply these
quantities to obtain the $TQ$-cost, which provides a single objective for
comparing such designs. 

Together, the correctness tests and $TQ$-cost define a constrained
optimization problem over circuit-generating programs. Motivated by recent
advances in coding and reasoning agents and LLM-guided algorithmic discovery,
we use an LLM to propose coordinated, task-specific redesigns within an
evolutionary search~\cite{anthropic2026opus5,anthropic2026mythos5,novikov2025alphaevolve}.

\subsection*{Circuit model and validity}
\label{sec:mbu}

A compilation task specifies a reversible action on computational-basis
states: a permutation $\pi$ and, when needed, a Boolean phase $g$. A circuit
may use $a$ ancilla qubits, but must return them to zero:
\begin{equation}
  V\ket{x}\ket{0^a}
  =(-1)^{g(x)}\ket{\pi(x)}\ket{0^a}.
  \label{eq:kickmix-action}
\end{equation}
Every ancilla must return to $\ket{0}$ because
quantum operations must be unitary, and thus reversible. For a
numerical approximation, the required output relation may instead specify an
error bound; interface preservation, phase correctness, and ancilla uncomputation
remain exact. We separately allow some circuits to fail on a bounded fraction
of inputs drawn from a stated distribution; this fraction is the \emph{failure budget}. Following \citet{babbush2026kickmix}, the \kickmix{} model permits three
classes of operation:
\begin{enumerate}
  \item classical reversible gates: $X$, CNOT, Toffoli, and swap;
  \item diagonal sign gates: $Z$, CZ, CCZ, and global sign flips; and
  \item destructive $X$-basis measurements.
\end{enumerate}
Before any measurement, the evolution of a computational-basis input is
simple. Reversible gates update its stored bit string, while
diagonal gates update only its sign, so the state remains of the form
$(-1)^\phi\ket{x}$. A classical simulator therefore need only track the
current bits $x$ and the sign $\phi$, rather than a full state vector. The operations apply phase
\textbf{kick}back while \textbf{mix}ing basis states, giving \kickmix{} its
name.

To illustrate why uncomputation is costly, consider a temporary AND. A
Toffoli computes $t=uv$ into an ancilla. After $t$ has been used, applying the
same Toffoli again returns that qubit to zero:
\begin{equation}
  \ket{u,v}\ket{0}
  \xrightarrow{\mathbf{Toffoli}}
  \ket{u,v}\ket{uv}
  \xrightarrow{\text{use }uv}
  \ket{u,v}\ket{uv}
  \xrightarrow{\mathbf{Toffoli}}
  \ket{u,v}\ket{0}.
\end{equation}
Computing and uncomputing this temporary AND therefore costs two Toffolis.

MBU, introduced by
\citet{gidney2018halving}, replaces the second Toffoli with a destructive
$X$-basis measurement and a cheap phase correction. Measuring $t$ produces a
uniformly random outcome $r$ and releases the ancilla, which can be restored
to $\ket{0}$:
\begin{equation}
  \ket{u,v}\ket{0}
  \xrightarrow{\mathbf{Toffoli}}
  \ket{u,v}\ket{uv}
  \xrightarrow{\text{use }uv}
  \ket{u,v}\ket{uv}
  \xrightarrow{\ X\text{-measure }r\ }
  (-1)^{ruv}\ket{u,v}\ket{0}
  \xrightarrow{\mathrm{CZ}\ \text{if }r=1}
  \ket{u,v}\ket{0}.
\end{equation}
Because CZ is Clifford, the temporary AND is computed and cleared using only
one Toffoli. By linearity, the same reasoning applies component by component
when the retained bits are in superposition.

More generally, suppose an ancilla contains a Boolean function $f(y)$ of
retained basis data $y$. If $m$ records any earlier measurement outcomes, the
simulator represents the ancilla's destructive $X$-basis measurement by the transition
\begin{equation}
  ((y,f(y)),m,\phi)
  \longmapsto
  (y,mr,\phi\oplus r f(y)),
  \qquad \Pr(r=0)=\Pr(r=1)=\frac12.
  \label{eq:kickmix-measurement-transition}
\end{equation}
On a coherent input, the retained state is therefore
\begin{equation}
  \sum_y\alpha_y(-1)^{r f(y)}\ket{y}.
  \label{eq:kickmix-mbu}
\end{equation}
When $r=1$, a classically conditioned diagonal circuit for
$(-1)^{f(y)}$ cancels this sign.

For a fixed computational-basis input and a particular sequence of
measurement outcomes $m$, we call the resulting evolution a
\emph{computational-basis trajectory}.
A \kickmix{} circuit may perform several such measurements. Their outcomes
may be combined by classical XOR and may control diagonal corrections, but
cannot control gates that change the stored basis string. Every measurement
record therefore induces the same basis permutation:
\begin{equation}
  V_m\ket{x}=(-1)^{\phi(x,m)}\ket{\pi(x)}.
  \label{eq:kickmix-record}
\end{equation}
Exact coherent correctness additionally requires
\begin{equation}
  \phi(x,m)=g(x)
  \label{eq:kickmix-phase-correctness}
\end{equation}
for every allowed input and every record of nonzero probability. The
measurement register then factors from the logical output. Allowing an
outcome to control a basis-changing gate would instead make the permutation
depend on the record and could leave the logical output correlated with the
measurement environment.

\subsection*{Validation and statistical guarantees}
\label{sec:verification}

The circuits are ultimately applied to coherent superpositions, whereas the
evaluator follows one computational-basis trajectory at a time. Each test
must therefore capture enough information to determine the trajectory's
coherent action. For a basis input $x$ and measurement record $m$, the
evaluator checks that the output matches the target, required inputs are
unchanged, every ancilla has returned to $\ket{0}$, and the phase agrees with
the target exactly. These checks compare the complete final state rather than
only its classical output. An ancilla that retains input-dependent information
can become entangled with the logical registers, while an incorrect sign can
alter subsequent interference without changing any output bit. Resetting or tracing out a dirty
ancilla is therefore not a substitute for uncomputing it.

If every pair $(x,m)$ passes these tests, the measurement record factors from
the logical output, and linearity extends agreement on every trajectory to any
coherent input. For large input domains, however, the evaluator samples these
trajectories rather than checking them all.

Consider a coherent input $\ket{\psi}=\sum_x\alpha_x\ket{x}$, with $\mathcal D(x)=|\alpha_x|^2$, and let $\mathcal M$ be the set of possible measurement records. After
deferring the measurements, the candidate and target states can be written as
\begin{equation}
  \ket{\Phi_V}
  =\frac{1}{\sqrt{|\mathcal M|}}\sum_{x,m}
    \alpha_x V_m\ket{x}\ket{m},
  \qquad
  \ket{\Phi_U}
  =\frac{1}{\sqrt{|\mathcal M|}}\sum_{x,m}
    \alpha_x U\ket{x}\ket{m}.
\end{equation}
Let $B$ be the set of pairs $(x,m)$ on which $V_m\ket{x}\ne U\ket{x}$,
and let
\begin{equation}
  \varepsilon=\frac{1}{|\mathcal M|}
  \sum_{(x,m)\in B}|\alpha_x|^2
\end{equation}
be the probability of $B$ under $\mathcal D$ and the uniform distribution over measurement records.
Projecting onto the span of these defective trajectories gives
\begin{equation}
  \bigl\|\ket{\Phi_V}-\ket{\Phi_U}\bigr\|
  \le 2\sqrt{\varepsilon}.
  \label{eq:transfer}
\end{equation}
Consequently, the state fidelity is at least $(1-2\varepsilon)^2$. Any subsequent
measurement distribution changes by at most $2\sqrt{\varepsilon}$ in total variation
distance. We bound $\varepsilon$ by drawing $x$ from $\mathcal D$ and sampling every
measurement outcome uniformly, exactly as in a basis trajectory. For a
candidate fixed before the test draw, $k$ independent tests all pass with
probability $(1-\varepsilon)^k$. Thus, if all $k$ tests pass, then with confidence $1-\delta$
\begin{equation}
  \varepsilon\le 1-\delta^{1/k}
  \approx\frac{\ln(1/\delta)}{k}.
\end{equation}
More generally, if a test with $k$ random samples accepts at most $b$ failures, a candidate with true failure probability $\varepsilon$ passes with probability
\begin{equation}
  P_{k,b}(\varepsilon)
  =\sum_{j=0}^{b}\binom kj
    \varepsilon^j(1-\varepsilon)^{k-j}.
\end{equation}

\subsection*{Cost model and scoring}\label{sec:costmodel}

We define the $TQ$-cost of a circuit by
\begin{equation}
  TQ = \text{Toffoli count} \times \text{peak qubit count},
\end{equation}
where $T$ counts Toffoli and CCZ gates equally and $Q$ is the largest number of logical
qubits allocated at once.

We weight space and time equally in this cost function for simplicity. In
practice, their relative importance depends on the hardware and application.
When qubits are scarce, for example, a cost of the form
$(\text{qubit count})^{\kappa}\times\text{Toffoli count}$ with
$\kappa>1$ may be more appropriate; resource estimates for factoring
have used $\kappa$ between 1 and
3~\cite{gidney2021factor,gidney2025factor}. Either $T$ or $Q$ can also be replaced with other measures of time or space.
In particular, when there is limited magic-state-factory throughput (as is typically the case), runtime is bottlenecked by the Toffoli count, motivating its widespread use as a time proxy~\cite{fowler2012time,gidney2019autoccz}. However, with
enough factories to supply magic states in parallel, Toffoli depth rather than
Toffoli count limits runtime. Similarly, for architectures with cheap cold storage for inactive qubits~\cite{gidney2023yoked}, $Q$ can be replaced by active
volume~\cite{litinski2022active}, which counts only the qubits
participating in logical operations.

\looseness=-1
The search is not tied to the specific form of the $TQ$-cost: it can run with any
computable objective. A later case study
illustrates the available trade-offs: the modular-inversion search
(SI~\ref*{sec:case-modinv}) produces a family of circuits along a
Toffoli--qubit frontier (SI~\ref*{sec:modinv-divsteps}), and an
objective that weights qubit count more heavily selects a lower-qubit member
of that family.

\subsection*{Evolutionary search}
\label{sec:search}

We encode each candidate as Python code that constructs a circuit. Functions,
loops, and named registers preserve the organization of the computation,
allowing the LLM to reason about arithmetic subroutines and their composition
without reconstructing that structure from a flattened gate sequence. For
example, changing a shared arithmetic routine changes every place where it is
used, while reorganizing calls to that routine can change when intermediate
values are computed, retained, and uncomputed. The representation thus makes
coordinated changes expressible as revisions to the underlying algorithm.

We search over these circuit-generating programs with
OpenEvolve~\cite{openevolve}, using Claude
Mythos 5~\cite{anthropic2026mythos5} to propose architectures and program
rewrites. Each of the eight searches uses a fixed, hand-written task specification, and all
share the same generic circuit-building interface. During ideation, the LLM receives
only the specification, circuit model, and correctness criteria, and proposes several
distinct circuit architectures. The LLM has no access to the
evaluator beyond the feedback returned
after each submission. Starting from a blank program, it chooses the algorithm,
representation, and circuit structure and produces the discovered circuits
without substantive human intervention.

Each architecture proposed during ideation founds a separate family of
candidate programs. Later iterations improve descendants of the first working
implementation. Keeping the families separate at first gives each approach
time to mature, even if its earliest circuits are less efficient than those
based on another approach. Periodically, leading programs are shared across
families, allowing successful ideas from different approaches to combine.

The search retains candidate programs and their evaluation results, tracking
which family each belongs to and which earlier program it was derived from.
It also saves feedback on failures and short notes in which the LLM records
useful insights. These programs, results, and notes carry information between
iterations; the sessions do not share a continuously growing conversation.

Each iteration starts a new LLM agent session with the task specification, a
parent program to improve, and its evaluation feedback. The agent also
receives alternative programs from the archive and their resource costs, brief
summaries of earlier attempts, and recorded insights where available. This
context is limited; the agent never sees the complete search history. Within
the session, it may rewrite the program, run drafts on test inputs, query a
classical oracle for the required function on chosen inputs, and use the
results to revise its design. It may also delegate focused subtasks to subagents.
The iteration ends with submission of a revised
program. The evaluator rejects a circuit that does not
meet the task specification; a valid one is scored by its $TQ$-cost and
considered for retention in its family. Thus the LLM is free to propose
broad changes without having to judge their correctness: the evaluator
supplies an independent test of every completed design.
\Cref{fig:prompts} shows the ideation and rewrite calls and a concrete
feedback-and-repair cycle for a small finite-field inversion task.

\enlargethispage{2\baselineskip}
To decide which valid candidates remain available within each family, the
search retains circuits with different trade-offs between Toffoli count and
qubit count. We implement this selection with
MAP-Elites~\cite{mouret2015illuminating}, which divides the $(T, Q)$ resource plane
into cells and keeps the circuit with the lowest $TQ$-cost in each cell. A
space-efficient circuit can therefore remain available for later development
even when its higher Toffoli count gives it a worse $TQ$-cost.

Because the families preserve alternatives, progress need not be
monotonic in either resource. The Gaussian lineage in
\cref{fig:gaussian-lineage}, for example, contains descendants that
temporarily worsen Toffoli count or qubit count before reaching a region of lower
$TQ$-cost. We report the lowest $TQ$-cost reached across all families at the end
of the search.
\clearpage

\ifrevtex\bibliography{refs}\else\printbibliography\fi

\clearpage

\section*{Extended Data}
\begingroup
\setcounter{figure}{0}
\captionsetup[figure]{name=Extended Data Figure}
\renewcommand{\thefigure}{\arabic{figure}}
\renewcommand{\theHfigure}{ED.\arabic{figure}}

\begin{figure}[!htbp]
\centering
\input{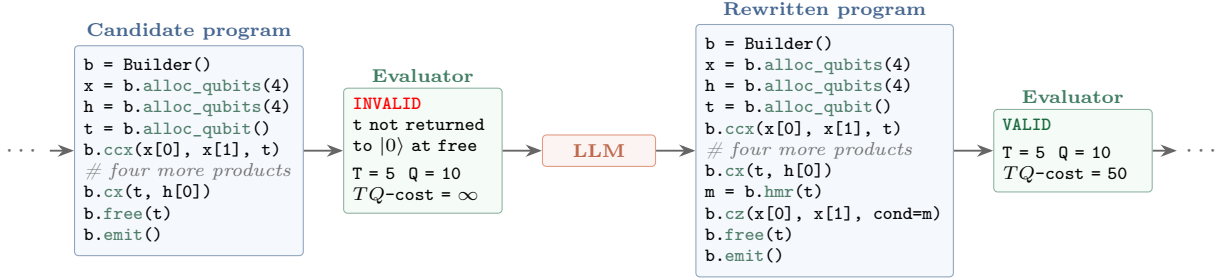}
\caption{Sample inputs and outputs for inversion in
$\mathbb F_{2^4}$. (a) During ideation, the LLM is given the complete task
specification but no prior implementation and returns distinct candidate
architectures. (b) During evolution, the input additionally contains a
selected parent program, its evaluator feedback, and limited archive context;
the LLM returns a proposed rewrite. (c) One such rewrite computes the inverse with five nonlinear
products but fails to return an ancilla to $\ket{0}$, so the evaluator rejects
it. Seeing this feedback, the LLM adds MBU; the
resulting circuit is valid with $T=5$, $Q=10$, and a $TQ$-cost of 50.}
\label[extendeddatafigure]{fig:prompts}
\end{figure}

\begin{figure}[!htbp]
\centering
\includegraphics[width=0.5\linewidth]{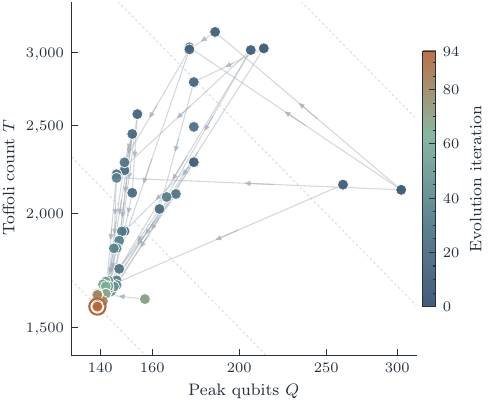}
\caption{Search lineage for the Gaussian circuit in the resource plane.
Each marker is a distinct valid circuit, colored by the iteration in which it was
discovered; arrows point from each parent to its mutated child. Dashed
curves are contours of constant $TQ$-cost, which decreases toward
the lower left. The ring marks the discovered circuit at $T=1{,}583$ and
$Q=139$.}
\label[extendeddatafigure]{fig:gaussian-lineage}
\end{figure}
\endgroup

\clearpage
\ifrevtex
\makeatletter
\edef\si@portrait{\pdfpagewidth=\the\pdfpagewidth \pdfpageheight=\the\pdfpageheight
  \textwidth=\the\textwidth \textheight=\the\textheight
  \oddsidemargin=\the\oddsidemargin \evensidemargin=\the\evensidemargin
  \topmargin=\the\topmargin \headheight=\the\headheight
  \headsep=\the\headsep \footskip=\the\footskip\relax}
\newgeometry{layoutsize={11in,8.5in},margin=1in}
\edef\si@landscape{\pdfpagewidth=11in \pdfpageheight=8.5in
  \textwidth=\the\textwidth \textheight=\the\textheight
  \oddsidemargin=\the\oddsidemargin \evensidemargin=\the\evensidemargin
  \topmargin=\the\topmargin \headheight=\the\headheight
  \headsep=\the\headsep \footskip=\the\footskip\relax}
\appdef\@outputpage@head{%
  \ifnum\value{page}=\getpagerefnumber{ed-landscape}\relax\si@landscape\else\si@portrait\fi}
\makeatother
\else
\pdfpagewidth=11in \pdfpageheight=8.5in
\newgeometry{layoutsize={11in,8.5in},margin=1in}
\fi
\label{ed-landscape}
\vspace*{3em}
\begingroup %
\setlength{\tabcolsep}{2.1pt}
\renewcommand{\arraystretch}{1.0}
\captionsetup{hypcap=false}
\captionsetup[table]{name=Extended Data Table}
\setcounter{table}{0}
\renewcommand{\thetable}{\arabic{table}}
\renewcommand{\theHtable}{ED.\arabic{table}}
\captionof{table}{Main results
across eight compilation tasks. Input/output gives the number of qubits in
read-only input registers and target registers, respectively; a target may
carry an initial value and be updated in place, and ``---'' denotes a task
with no target register. The correctness columns give the
error tolerance and the number of tests. The resource columns report
Toffoli count $T$ and peak qubit count $Q$ for the baseline and discovered
circuits; $G_{TQ}=(TQ)_{\mathrm{baseline}}/(TQ)_{\mathrm{ours}}$. HHL: Harrow--Hassidim--Lloyd.}
\label[extendeddatatable]{tab:results}
\begin{center}
\begin{threeparttable}
\begin{tabular}{lcclccc>{\columncolor{deltabg}\color{deltafg}\bfseries}c}
\toprule
& & & \multicolumn{2}{c}{\textbf{Correctness criteria}} &
\multicolumn{3}{c}{\textbf{Resources}} \\
\cmidrule(lr){4-5}\cmidrule(lr){6-8}
\textbf{Task} & \textbf{Details} & \textbf{Input/output} &
\multicolumn{1}{c}{\textbf{Error tolerance}} &
\textbf{Tests} &
\textbf{Baseline} $\boldsymbol{(T,Q)}$ &
\textbf{Ours} $\boldsymbol{(T,Q)}$ &
\multicolumn{1}{>{\columncolor{deltabg}}c}{$\boldsymbol{G_{TQ}}$} \\
\midrule
Prime-field modular inversion &
SI~\ref*{sec:case-modinv} & $512/512$ &
Failure budget $\le2{\cdot}10^{-4}$~\cite{babbush2026kickmix,schrottenloher2026pointaddition} & $10^5$ &
---\tnote{\dag} &
$(4.0{\cdot}10^6,\,935)$ &
\multicolumn{1}{>{\columncolor{deltabg}\color{black}}c}{---} \\
Reed--Solomon DQI decoder &
SI~\ref*{sec:case-dqi-rs} & $840/\text{---}$ &
Failure budget $\le10^{-3}$ & $10^5$ &
$(5.1{\cdot}10^6,\,1{,}938)$~\cite{khattar2025optimized-dqi} &
$(4.4{\cdot}10^5,\,1{,}837)$ & $12$ \\
Hermitian DQI decoder &
SI~\ref*{sec:case-dqi-hopi} & $738/\text{---}$ &
Failure budget $\le10^{-3}$ & $10^5$ &
$(1.0{\cdot}10^8,\,1.9{\cdot}10^5)$\tnote{\ddag}~\cite{sakata1995feng-rao} &
$(9.3{\cdot}10^5,\,2{,}395)$ &
\multicolumn{1}{>{\columncolor{deltabg}\color{black}}c}{---\tnote{\ddag}} \\
HHL rotation angle &
SI~\ref*{sec:case-hhl-angle} & $21/33$ &
$<10^{-7}$ rad & $9{\cdot}10^3$ &
$(1.6{\cdot}10^4,\,300)$~\cite{haner2018optimizing} &
$(1{,}348,\,115)$ & $32$ \\
Coulomb potential &
SI~\ref*{sec:case-coulomb} & $48/25$ &
$\le2{\cdot}10^{-6}$ absolute error & $9{\cdot}10^3$ &
$(6{,}355,\,252)$~\cite{haner2018optimizing} &
$(1{,}140,\,142)$ & $9.9$ \\
Fixed-point Gaussian &
SI~\ref*{sec:case-gaussian} & $39/39$ &
$<10^{-7}$ absolute error & $9{\cdot}10^3$ &
$(2.4{\cdot}10^4,\,362)$~\cite{haner2018optimizing} &
$(1{,}583,\,139)$ & $39$ \\
Fe$_2$S$_2$ alias--\textsc{Select} &
SI~\ref*{sec:case-alias-fe2s2} & $32/40$ &
Exact on the entire domain & $2{\cdot}10^7$ &
$(5.5{\cdot}10^4,\,139)$~\cite{babbush2018encoding} &
$(1.5{\cdot}10^4,\,126)$ & $4.1$ \\
Nonabelian lattice-gauge simulation &
SI~\ref*{sec:case-nonabelian} & $1{,}728/51$ &
Exact on the entire domain & $10^5$ &
$(1.3{\cdot}10^6,\,1{,}927)$~\cite{osorio2025sigma72} &
$(1.7{\cdot}10^5,\,1{,}859)$ & $7.7$ \\
\bottomrule
\end{tabular}
\begin{tablenotes}[flushleft]
\item[\dag] Costs are for the complete benchmark circuit of SI~\ref*{sec:case-modinv}, whose cost is dominated by modular inversion.
Several prior implementations trade Toffolis against qubits
differently (\cref{tab:modinv-frontier})~\cite{babbush2026kickmix,schrottenloher2026pointaddition,luo2026space}, so we report no single baseline or gain.
Separately from the divstep design reported in the table, we specialize Luo's space-efficient
construction, reducing its Toffoli count by a factor of $3.4$ at the cost of five additional qubits (SI~\ref*{sec:modinv-eea-specialized}).
\item[\ddag] Na\"ive reversible implementation; no prior reversible compilation exists in the literature.
\end{tablenotes}
\end{threeparttable}
\end{center}
\vspace*{\fill}
\endgroup
\restoregeometry
\ifrevtex\else\pdfpagewidth=8.5in \pdfpageheight=11in\fi
\clearpage

\setcounter{page}{1}
\ifrevtex\makeatletter\let\si@landscape\si@portrait\makeatother\fi
\renewcommand{\thepage}{\texorpdfstring{\arabic{page}}{S\arabic{page}}}
\section*{Supplementary Information}
\phantomsection\label{sec:si}
\renewcommand{\thesection}{\Roman{section}}
\setcounter{figure}{0}
\setcounter{table}{0}
\setcounter{algorithm}{0}
\renewcommand{\thefigure}{S\arabic{figure}}
\renewcommand{\thetable}{S\arabic{table}}
\renewcommand{\thealgorithm}{S\arabic{algorithm}}
\renewcommand{\theHfigure}{S.\arabic{figure}}
\renewcommand{\theHtable}{S.\arabic{table}}

\startcontents[si]
\printcontents[si]{}{1}{%
  \renewcommand{\maintocdepth}[1]{}\let\appendix\relax\setcounter{tocdepth}{2}}
\bigskip

\section{Syndrome decoding for DQI}\label{sec:case-dqi}

Consider an optimization problem over a finite field $\mathbb F_q$. An
assignment $x\in\mathbb F_q^n$ produces $m$ linear values $Bx$. Given a set
$F_i\subset\mathbb F_q$ of acceptable values for each constraint, let
$f(x)=|\{i:(Bx)_i\in F_i\}|$ count the constraints satisfied by $x$; the goal
is to maximize $f(x)$. This family, known as max-LINSAT, includes linear
constraint-satisfaction problems such as Max-Cut and max-XORSAT, as well as
nearest-codeword problems and polynomial fitting over finite fields. Decoded
quantum interferometry (DQI) is a quantum algorithm that aims to solve these
optimization problems by arranging for constructive interference on assignments satisfying many
constraints and destructive interference on assignments satisfying
fewer~\cite{jordan2024dqi}.

DQI prepares a superposition over sparse vectors
$e\in\mathbb F_q^m$, each supported on at most $\ell$ constraints. Here
$H=B^{\mathsf T}$ can be viewed as a parity-check matrix, making $e$ a
low-weight error and $He$ its syndrome. If the associated code has minimum
distance $d_{\min}\ge2\ell+1$, a decoder can recover $e$ uniquely from that syndrome.
DQI uses the decoder in the following state-preparation sequence:
\begin{equation}
\begin{aligned}
  \sum_{\operatorname{wt}(e)\le\ell}a_e\ket e\ket0
  &\xrightarrow{\ s=He\ }
  \sum_{\operatorname{wt}(e)\le\ell}a_e\ket e\ket{He}\\
  &\xrightarrow{\text{uncompute }e}
  \ket0\sum_{\operatorname{wt}(e)\le\ell}a_e\ket{He}
  \xrightarrow{\operatorname{QFT}^{-1}}
  \sum_x A_x\ket x.
\end{aligned}
\label{eq:dqi-state-preparation}
\end{equation}
The middle step must erase $e$ while retaining its syndrome $s=He$; otherwise the
different contributions remain distinguishable and cannot interfere. The
inverse quantum Fourier transform (QFT) then produces amplitudes $A_x$ that favor assignments
satisfying more constraints. The dominant cost in \cref{eq:dqi-state-preparation} is uncomputing $e$, which requires solving a classical decoding problem; DQI thus reduces optimization to
decoding.

The space-efficient construction of \citet{khattar2025optimized-dqi} replaces
the middle decoding-and-uncomputation step with a phase correction. It
generates one field symbol $e_i$ at a time, adds its contribution to $s$, and
then removes it with an $X$-basis measurement. This avoids retaining all $m$
symbols of $e$, but the measurement outcomes $c_i$ leave a known phase. We expand each symbol of $e \in \mathbb F_q^m$ in the fixed binary basis used by the circuit,
so that $c$ and $e$ are bit strings of equal length; the phase is then
$(-1)^{\langle c,e\rangle_2}$, where $\langle \cdot, \cdot \rangle_2$ denotes the binary inner product. We call $c$ the \emph{mask}. In this section, we use a fixed, randomly chosen mask. A reversible decoder must remove
this phase using $s$, after which the inverse
Fourier transform produces the desired interference over assignments
$x$. The case studies below isolate and optimize this decoder-dependent phase
correction. We require
\begin{equation}
  \ket{s}\longmapsto
  (-1)^{\langle c,e(s)\rangle_2}\ket{s}.
  \label{eq:dqi-phase-fixup}
\end{equation}

We study two instances of DQI. In optimal polynomial intersection (OPI), $x$
specifies a univariate polynomial and the constraints restrict its values
at prescribed points on a line; the associated code is a Reed--Solomon code.
Hermitian optimal polynomial intersection (HOPI) generalizes this picture from
a line to a Hermitian curve: $x$ specifies a function on the curve, the
constraints restrict its values at curve points, and the corresponding code is a Hermitian
code~\cite{gu2025hermitian-dqi}.
\Cref{tab:dqi-instances} lists the parameters of the two benchmark instances.

\begin{table}[htbp]
\caption{Parameters of the two DQI decoding instances. The syndrome-input
column gives the number of syndrome symbols and the bits per symbol; $\ell$ is
the maximum decoded error weight. The last column compares expected satisfied
fractions for ideal DQI and a single trial of Prange's information-set decoding~\cite{prange1962information}, which is the best known efficient classical algorithm for these instances.}
\label{tab:dqi-instances}
\centering
\begin{tabular}{lcrcrc}
\toprule
\textbf{Problem} & \textbf{Code field} & \textbf{Positions} &
\textbf{Syndrome input} & $\boldsymbol{\ell}$ & \textbf{DQI/Prange} \\
\midrule
OPI & $\mathbb F_{2^{12}}$ & $4{,}095$ & $70\times12$ & $35$ &
$0.5763/0.5009$ \\
HOPI & $\mathbb F_{64}$ & $512$ & $123\times6$ & $47$ &
$0.7145/0.5726$ \\
\bottomrule
\end{tabular}
\end{table}

We test implementations on syndromes drawn from the DQI input distribution (defined by $|a_e|^2$) and accept
implementations with at most 100 detected wrong signs in 100,000 tests. Let
$\varepsilon$ be the failure probability: the total probability, under this distribution, of the inputs
on which the circuit fails. A failure applies the wrong sign or resets an
ancilla that still holds input-dependent data. The discovered circuit
satisfies~\cite{gu2026lean}
\begin{equation}
\bigl|\langle f\rangle_{\mathrm{impl}}-\langle f\rangle\bigr|
\le 4\varepsilon\langle f\rangle.
  \label{eq:dqi-budget-bound}
\end{equation}

Under our testing framework, a wrong sign is detected with probability at
least $\tfrac12$: for any fixed input, whether the sign is wrong is given by a
constant plus a linear function (modulo 2) of the uniformly random
measurement record, and a reset of a
nonzero qubit adds a fresh uniformly random bit to the sign. Thus the
acceptance limit of 100
failures in 100,000 tests gives $\varepsilon\le2.434\times10^{-3}$ at $97.5\%$
confidence. At that value, \cref{eq:dqi-budget-bound} bounds the loss in
satisfied fraction by $5.6\times10^{-3}$ (Reed--Solomon) and
$7.0\times10^{-3}$ (Hermitian)---at most $7.4\%$ and $4.9\%$ of the two
instances' advantage over single-trial Prange.

\subsection{Reed--Solomon decoding for OPI}\label{sec:case-dqi-rs}

We use the Reed--Solomon instance and reversible decoder of
\citet{khattar2025optimized-dqi} as our baseline. Let $\gamma$ be a primitive
element of $\mathbb F_{2^{12}}$. An error word
$e\in\mathbb F_{2^{12}}^{4095}$ produces 70 syndromes
\begin{equation}
  S_i=\sum_{j=0}^{4094}e_j\gamma^{ij},
  \qquad i=1,\ldots,70.
  \label{eq:rs-syndromes}
\end{equation}
This syndrome map defines a $[4095,4025,71]_{2^{12}}$ Reed--Solomon code.
The 70 twelve-bit syndromes occupy 840 qubits. Assuming the error weight is at most 35, the decoding task is to recover the error from the syndromes.

Decoding proceeds in two stages: first find the positions of the errors, then
recover their values in $\mathbb F_{2^{12}}$. The first stage packages the
unknown set of error positions into a single polynomial. For each nonzero
error symbol $e_j$, include one factor $1-\gamma^jz$, giving
\begin{equation}
  \sigma(z)=\prod_{j:e_j\ne0}(1-\gamma^jz).
  \label{eq:rs-locator}
\end{equation}
The factor associated with position $j$ vanishes at $z=\gamma^{-j}$, so
the roots of $\sigma$ identify exactly the positions at which the error word
is nonzero. The decoder's first task is therefore to recover this
\emph{locator polynomial} from the syndromes. The input $\ket{s}$ supplies the
first 70 syndromes, which we collect into a polynomial
\begin{equation}
  S(z)=S_1+S_2z+\cdots+S_{70}z^{69}.
  \label{eq:rs-syndrome-polynomial}
\end{equation}
Summing the geometric series in \cref{eq:rs-syndromes} shows that the product
$\sigma(z)S(z)$ agrees, through degree 69, with a polynomial $\omega(z)$ of
degree less than $\deg\sigma$, called the \emph{error-evaluator polynomial}:
\begin{equation}
  \sigma(z)S(z)=\omega(z)\pmod {z^{70}}.
  \label{eq:rs-key-equation}
\end{equation}
The decoder recovers the locator $\sigma$ of least degree and the
corresponding evaluator $\omega$ from this \emph{key equation}. The baseline
does so with the extended Euclidean algorithm (EEA) applied to $S(z)$ and
$z^{70}$. It repeatedly divides the higher-degree polynomial by the
lower-degree one, replaces the former by the remainder, and stops when the
remainder has degree below 35. That remainder is $\omega$, and the
coefficient that accompanies it is $\sigma$, up to normalization. The promise
of at most 35 errors ensures $\deg\sigma\le35$ and determines this pair
uniquely.

The baseline next locates the errors by evaluating $\sigma$ at every possible
position. Associate position $j$ with the field point
$\beta_j=\gamma^{-j}$ (since $\gamma$ is primitive, these points enumerate
all 4,095 nonzero field elements). The scan visits them one by one and tests whether
$\sigma(\beta_j)=0$. It need not evaluate $\sigma$ independently each time:
advancing to $\beta_{j+1}=\gamma^{-1}\beta_j$ multiplies its degree-$r$ term
by the fixed constant $\gamma^{-r}$. This efficient scan is known as a Chien
search~\cite{chien1964search}. Finding a root reveals an error's position but
not its value. At $\beta_j$, the evaluator $\omega(\beta_j)$ contains the
error value multiplied by the contributions of the other locator factors,
while the formal derivative $\sigma'(\beta_j)$---the usual coefficientwise
polynomial derivative over the field---contains those same factors without
the error value. Their ratio therefore isolates $e_j$, up to a fixed
factor $\kappa_j$ that depends only on the position:
\begin{equation}
  e_j=\kappa_j\frac{\omega(\beta_j)}{\sigma'(\beta_j)}.
  \label{eq:rs-forney}
\end{equation}
This recovery step is Forney's formula~\cite{forney1965bch}. The circuit
immediately uses that value and the matching twelve-bit block $c_j$ to apply the phase
\eqref{eq:dqi-phase-fixup}. It then reverses the computation to restore the
syndrome and clear its ancillas.

The EEA is the dominant cost of this baseline because its classical control flow
depends on the syndromes: different inputs require different shifts, swaps,
and numbers of iterations. A quantum circuit must instead apply one fixed
sequence of gates across the syndrome superposition. The baseline therefore
spreads the Euclidean solve over 222 synchronized cycles, controlling each
operation so that it acts only when needed. Our implementation of this baseline, based on the authors' released code~\cite{khattar2025dqi-data}, uses 5.1 million
Toffoli gates and 1,938
qubits at peak, with the forward Euclidean solve and its reversal accounting for 3.8 million
of the Toffolis~\footnote{This agrees with the decoder-stage estimate recomputed
from the authors' released code~\cite{khattar2025dqi-data} (5.0 million
Toffolis, 1,885 qubits) within about $2\%$ in Toffoli count; the qubit counts differ because
the estimate takes stage-wise maxima while the composed circuit retains the
Euclidean state during the scan.}.

The discovered circuit changes both stages of this baseline. First, it replaces the
Euclidean solve with the arithmetically cheaper Berlekamp--Massey (BM) recurrence
and arranges the recurrence so that reversibility does not require retaining
much history state. Second, although the decoded error word contains 49,140
bits, DQI extracts only the single bit $\langle c,e\rangle_2$ from it. The
circuit therefore computes this phase bit directly from the locator and
syndromes rather than materializing the full decoded error, and it
precomputes one reciprocal polynomial so that no field division is performed
during the scan.

An alternative to solving the key equation by EEA is the standard
BM algorithm, which recovers the locator from a linear
recurrence among the syndromes. We first review how BM reads
\cref{eq:rs-key-equation} one coefficient at a time. Write
$\sigma(z)=\sum_{k=0}^L\sigma_kz^k$, where $L=\deg\sigma$. The coefficient
of $z^{r-1}$ in $\sigma(z)S(z)$ is the convolution
$\sum_{k=0}^L\sigma_kS_{r-k}$. Because $\deg\omega<L$, the key equation
requires every such coefficient from degree $L$ through degree 69 to vanish:
\begin{equation}
  0=[z^{r-1}]\,\sigma(z)S(z)
   =\sum_{k=0}^L\sigma_kS_{r-k},
  \qquad r=L+1,\ldots,70.
  \label{eq:rs-syndrome-recurrence}
\end{equation}
Since we can set $\sigma_0=1$ without loss of generality, and subtraction is the same as addition in
$\mathbb F_{2^{12}}$, every syndrome after the first $L$ satisfies
\begin{equation}
  S_r=\sigma_1S_{r-1}+\cdots+\sigma_LS_{r-L}.
  \label{eq:rs-syndrome-prediction}
\end{equation}
Thus the locator coefficients give a rule for predicting each syndrome from
the preceding $L$ syndromes. We call $L$ the \emph{order} of this recurrence. BM finds this recurrence by maintaining a
candidate locator $\sigma$ and refining it as the syndromes are
incorporated~\cite{massey1969shift}. At the start of round
$r$, the candidate defines a lowest-order recurrence for the already processed
syndromes $S_1,\ldots,S_{r-1}$. The round tests whether that recurrence also
fits $S_r$. The discrepancy $\Delta=\sum_{k=0}^{L}\sigma_kS_{r-k}$ is the sum in
\cref{eq:rs-syndrome-recurrence} for round $r$, computed with the current candidate. It measures the amount by which
the recurrence fails: $\Delta=0$ means that the current candidate already fits
$S_r$, while $\Delta\ne0$ requires the candidate to be refined.

When $\Delta\ne0$, BM needs a correction that fixes the equation for $S_r$
without invalidating any of the earlier equations. For this purpose, it retains an
earlier candidate $\tau$ as a correction template, together with the nonzero
discrepancy $\delta$ that candidate produced. At the start of each subsequent
round, BM multiplies $\tau$ by $z$, shifting every coefficient one degree
higher as the recurrence advances to the next syndrome. Thus, if the template
was saved $k$ rounds earlier, it has acquired a factor $z^k$ and is aligned
with the current syndrome equation while leaving the earlier equations unchanged. BM
then scales this aligned template by $\alpha=\Delta/\delta$ and applies
$\sigma\gets\sigma-\alpha\tau$: at $S_r$, the correction contributes exactly
$\alpha\delta=\Delta$ and so cancels the discrepancy.

On a round with a nonzero discrepancy, BM must also decide whether the
current candidate takes over as the template. A candidate of order $L$ that
fits $n$ syndromes has predicted $n-L$ of them. BM seeks the recurrence of
lowest order, and a correction raises the order less when its template had
predicted more syndromes at the time it was saved. BM therefore replaces the
template when the current candidate has predicted more syndromes than the
stored template had when it was saved, and retains it when the current candidate has
predicted fewer. When the two counts are equal, either choice gives a correct
decoder~\cite[p.~125]{massey1969shift}. Textbook BM retains in that case. The
circuit instead chooses to replace, which reduces the history it must store
(see ``Reconstructing BM history'' below).

The circuit makes this decision with a signed counter $d$. The counter is
nonnegative whenever a replacement is required and negative whenever the
template must be retained; when the counts are equal it may have either sign.
Starting from $(\sigma,\tau,\delta,d)=(1,1,1,0)$, each round
$r=1,\ldots,70$ applies
\begin{equation}
\begin{gathered}
\widetilde\tau=z\tau,
\qquad
\Delta=\sum_{i=0}^{\min(\deg\sigma,r-1)}\sigma_iS_{r-i},\\[1mm]
(\sigma,\tau,\delta,d)\gets
\begin{cases}
 (\sigma,\widetilde\tau,\delta,d+1),
   & \Delta=0,\\
 (\sigma-(\Delta/\delta)\widetilde\tau,\sigma,\Delta,-d),
   & \Delta\ne0,\ d\ge0,\\
 (\sigma-(\Delta/\delta)\widetilde\tau,\widetilde\tau,\delta,d+1),
   & \Delta\ne0,\ d<0.
\end{cases}
\end{gathered}
\label{eq:rs-bm-round-update}
\end{equation}
After the final round, $\sigma$ is the error locator.

\paragraph{Adaptive coefficient ranges.}
A na\"ive implementation would apply every BM update across all 36
possible coefficient positions in every round. Even in the worst case, this
is unnecessary: after $r$ syndromes have been processed, only coefficients
through degree $\min(r,35)$ can be active, so the update range need only grow
by one position per round. The dashed line in \cref{fig:dqi-rs}(b) shows this
worst-case bound. On the DQI input distribution, however, the candidate
polynomials typically grow more slowly still. A degree-$L$ locator $\sigma$
has $L$ unknown coefficients after fixing $\sigma_0=1$, while its evaluator
$\omega$, whose degree is less than $L$, has at most $L$ unknown
coefficients. Since each syndrome supplies one coefficient constraint, $r$
syndromes typically determine a locator of degree $L \approx r/2$, consistent
with the samples in \cref{fig:dqi-rs}(b).

The discovered circuit exploits this typical behavior, at the price of
failing on rare inputs. Before compiling the circuit, we
classically simulate BM on the DQI input distribution to choose a
cutoff $k$ for each round. The circuit then updates only the prefix
$\sigma_0,\ldots,\sigma_k$, treating higher coefficients as zero. The prefixes expand as the rounds
progress, which produces the
rising per-round cost in \cref{fig:dqi-rs}(a).

\paragraph{Shared BM updates.}\label{sec:rs-shared-updates}
Compiled directly, the two correction cases of \cref{eq:rs-bm-round-update}
would produce two separately controlled arithmetic paths. Only one path acts
on any basis state, but both would appear in the fixed circuit and contribute
to its Toffoli count. The main text shows how the circuit avoids this
(\cref{eq:rs-bm-rescaled-overview}). Under a flag that marks the replace case,
it exchanges $\sigma$ with the shifted template $\widetilde\tau$. One shared
update then subtracts a multiple of the template register from the locator
register. The multiplier is $\alpha$ in the retain case and $\alpha^{-1}$ in
the replace case. The replace result differs from the standard update only by
the nonzero factor $-\alpha^{-1}$ and therefore has the same roots.

The circuit also stores no discrepancy $\delta$. The rescaled update
maintains the invariant
\begin{equation}
  \sigma_0\delta=1.
  \label{eq:rs-scale-invariant}
\end{equation}
The invariant holds initially. The shifted template has zero constant
coefficient, so the retain case changes neither $\sigma_0$ nor $\delta$. In
the replace case, the new constant coefficient comes entirely from the
correction term: $\sigma_0\gets\alpha^{-1}\sigma_0=(\delta/\Delta)\sigma_0
=\Delta^{-1}$, while $\delta\gets\Delta$. The circuit therefore reads
$\delta=\sigma_0^{-1}$ from the locator. In each round it forms the product
$\Delta\sigma_0$, which equals $\alpha$,
and inverts it. A controlled selection between
the product and its inverse supplies the multiplier for the shared update.

\begin{figure}[t]
\centering
\includegraphics[width=\linewidth]{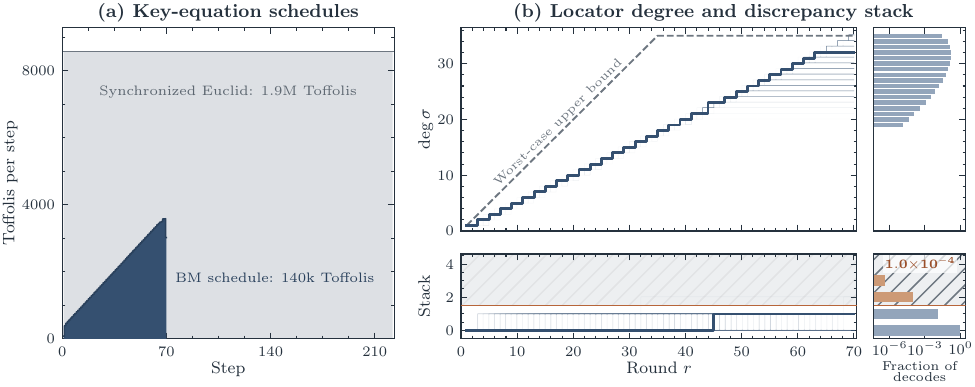}
\caption{The Reed--Solomon decoder. (a) Measured forward Toffoli cost per
step: one cycle of the 222-cycle synchronized EEA baseline of
\citet{khattar2025optimized-dqi}, or one round of our 70-round BM decoder. (b) Locator degree (top) and number of stored discrepancies (bottom) by round, across 20,000
samples drawn from the DQI input distribution, with one sample highlighted;
the dashed line is the worst-case degree bound $\min(r,35)$.
The right-hand panels show the final locator-degree distribution and the distribution of maximum
discrepancy-stack depth. Inputs in the hatched region need more than the one
allocated stack entry and may produce a wrong sign; this occurs with probability $1.0 \times 10^{-4}$.}
\label{fig:dqi-rs}
\end{figure}

\paragraph{Reconstructing BM history.}
Reversing a round requires the discrepancy that the round used, which the
round computes into a temporary register. Keeping a
twelve-bit history word for every round would raise the qubit count to 2,677.
A replace round leaves $\sigma_0=\Delta^{-1}$ by
\cref{eq:rs-scale-invariant}, so the locator determines the round's
discrepancy. The circuit therefore erases the temporary discrepancy register
by MBU, and the reverse pass recomputes the erased value from the
locator. A round with $\Delta=0$ changes nothing and has nothing to recover.

A retain round with $\Delta\ne0$ is the exception. No exchange occurs and
$\sigma_0$ is unchanged, so the constant coefficient that reveals the
discrepancy on a replace round carries no trace of it here. The current
$\Delta$ remains only in the temporary discrepancy register, and the circuit
stores it on the history stack before erasing that register.
A stack that never overflows needs 35 entries, or 665 qubits.

Retain rounds are rare on the DQI input distribution. Starting from $d=0$,
a run of nonzero discrepancies repeatedly selects replace and leaves $d=0$. A zero
discrepancy raises $d$; a later replacement can then make it negative,
producing retain rounds while $d$ returns to zero. Such accidental zero
discrepancies are uncommon over $\mathbb F_{2^{12}}$ before the locator reaches its
final degree. Once it has, subsequent discrepancies stay zero and do not
start new runs of retain rounds. Consequently, most inputs need few saved
discrepancies. The discovered circuit allocates only one stack entry,
which costs 19 qubits (\cref{fig:dqi-rs}(b)).

Across 400 million samples from the DQI input
distribution, the two failure modes (overflow of the one-entry stack and activation of a coefficient beyond the cutoff) occur with a combined probability of
$1.4\times10^{-4}$. This completes locator construction; the circuit next
forms the evaluator $\omega$, which will be used together with $\sigma$ to
compute the required phase.

\paragraph{In-place construction of the evaluator polynomial.}
The key equation~\eqref{eq:rs-key-equation} gives $\omega$ by multiplying the
locator by the syndrome polynomial. The circuit first normalizes $\sigma$ so
that $\sigma_0=1$. This leaves the locator roots unchanged, and the
corresponding $\omega$ acquires the same nonzero factor, so the ratio
$\omega/\sigma'$ used to recover the error values is unchanged. Write
$\omega_{<35}=(\omega_0,\ldots,\omega_{34})$ for the first 35 coefficients of $\omega$:
\begin{equation}
  (\sigma;S_1,\ldots,S_{35};S_{36},\ldots,S_{70})
  \longmapsto
  (\sigma;\omega_{<35}=\sigma S\bmod z^{35};S_{36},\ldots,S_{70}).
  \label{eq:rs-inplace-evaluator}
\end{equation}
With this normalization, each output coefficient is
\begin{equation}
  \omega_i=S_{i+1}+\sum_{j=1}^{i}\sigma_jS_{i-j+1},
  \qquad 0\le i<35,
  \label{eq:rs-evaluator-triangular}
\end{equation}
so it equals the syndrome symbol being overwritten plus a function of earlier
symbols. The transformation is therefore reversible:
the circuit can overwrite the first 35 syndrome symbols with the first 35
coefficients of $\omega$ and later restore them. This avoids a separate
420-qubit register for $\omega$.
The same map extends to all 70 syndrome symbols; by the key
equation, the coefficients of $\sigma S$ from degree 35 through 69 vanish
whenever the locator is correct, so the last 35 become zero and serve as
workspace until the map is reversed.

\paragraph{Phase-only error recovery.}
At this point the decoder could use $\sigma$ and $\omega$ to recover every
error value. An ordinary decoder does so by applying Forney's formula
\eqref{eq:rs-forney} to construct the twelve-bit field element $e_j$ at each
located error. The DQI task~\eqref{eq:dqi-phase-fixup}, however, needs only the
single accumulated parity $\sum_j\langle c_j,e_j\rangle_2$. Constructing all
twelve bits of each $e_j$ only to retain one parity would waste most of this
work, so the discovered circuit computes the required phase bit directly.

Any parity of a field element can be expressed as a field trace after
multiplication by a fixed constant. In particular, with
$\Tr(x)=\sum_{r=0}^{11}x^{2^r}\in\mathbb F_2$, the twelve mask bits $c_j$
determine a unique element $\mu_j\in\mathbb F_{2^{12}}$ such that
\begin{equation}
  \langle c_j,e_j\rangle_2
  =\Tr(\mu_j e_j).
\end{equation}

Part of the total phase is linear in the supplied syndromes and can be applied
with Clifford $Z$ gates. Choosing that linear part offline cancels the mask at 71
positions and leaves residual constants $\widetilde\mu_j$ at the remaining
4,024 positions to scan.

At each remaining position the circuit tests whether $\beta_j$ is a
root of the locator. If it is, the error word is nonzero there. Absorb the fixed factor
$\kappa_j$ from Forney's formula into
$\nu_j=\widetilde\mu_j\kappa_j$. Substituting that formula directly into the
trace gives the residual phase bit
\begin{equation}
  \Tr(\widetilde\mu_j e_j)
  =\Tr\!\left(
    \nu_j
    \frac{\omega(\beta_j)}{\sigma'(\beta_j)}
  \right),
  \label{eq:rs-contracted-forney}
\end{equation}
without materializing $e_j$.

\paragraph{Precomputed reciprocal.}
The quotient in \cref{eq:rs-contracted-forney} is required only at the roots
of $\sigma$, of which there are at most 35. A quantum circuit, however,
applies one fixed gate sequence to every syndrome in the superposition, so it
cannot skip the positions that are not roots. A direct implementation would therefore perform the
division at each of the 4,024 positions, and this division would dominate
the decoder's cost.

The circuit instead performs one division before the scan. Because only the
values at the roots matter, any function that agrees with $1/\sigma'$ there
serves equally well. The roots of $\sigma$ are distinct, so
$\sigma'$ is nonzero at each of them, and a single polynomial $\Psi$ of degree
less than $\deg\sigma$ takes the value $1/\sigma'(\beta)$ at every root
$\beta$. Equivalently, $\Psi\sigma'\equiv1\pmod\sigma$. Evaluating $\Psi$ and
$\omega$ at a fixed point is linear in their coefficients and uses only
CNOTs, so the residual phase bit becomes
\begin{equation}
  \Tr(\widetilde\mu_j e_j)
  =\Tr\!\bigl(\nu_j\,\omega(\beta_j)\,\Psi(\beta_j)\bigr),
  \label{eq:rs-psi-pairing}
\end{equation}
which is bilinear in $\omega(\beta_j)$ and $\Psi(\beta_j)$.
In a basis of $\mathbb F_{2^{12}}$ over $\mathbb F_2$, it splits into
12 terms, each the product of one parity of the bits of $\omega(\beta_j)$ with
one parity of the bits of $\Psi(\beta_j)$. The parities cost only CNOTs, and each
product costs one Toffoli, applied only when the root test succeeds. The root
test itself is a zero test on the twelve bits of $\sigma(\beta_j)$ and costs
11 Toffolis, so each position costs $11+12=23$ Toffolis.

Computing $\Psi$ is itself a division: $\Psi$ is the inverse of $\sigma'$
modulo $\sigma$. The circuit computes it with the EEA (the
algorithm that BM replaced for the key equation) on a smaller
problem. Two observations make it affordable.

The Toffoli cost of the EEA grows roughly with the square of
the operand degree. The first observation halves that degree, which cuts the
Toffoli cost to about a quarter. In characteristic 2, squaring is linear: $(a+b)^2=a^2+b^2$. It therefore acts
on the twelve bits of a field element as a fixed linear map, as does its
inverse, the square root, and both cost only CNOTs. Group the coefficients of
$\sigma$ by the parity of their degree:
\begin{equation}
  \sigma(z)=\bigl(\sigma_0+\sigma_2z^2+\sigma_4z^4+\cdots\bigr)
    +z\bigl(\sigma_1+\sigma_3z^2+\sigma_5z^4+\cdots\bigr).
\end{equation}
Each group is a polynomial in $z^2$ and hence a square,
$\sigma_0+\sigma_2z^2+\cdots=\bigl(\sqrt{\sigma_0}+\sqrt{\sigma_2}\,z+\cdots\bigr)^2$,
so $\sigma(z)=E(z)^2+z\,O(z)^2$ with $E_i=\sqrt{\sigma_{2i}}$,
$O_i=\sqrt{\sigma_{2i+1}}$, and $\deg E,\deg O\le17$. In characteristic 2 the derivative of any square
is zero, since $(f^2)'=2ff'$, so $\sigma'=O^2$. The polynomials $E$ and $O$ share no factor, because
a common factor would divide both $\sigma$ and $\sigma'$. The EEA on this half-degree pair $(E,O)$ returns $A$ and $B$ of
degree at most 17 with $AE+BO=1$. Then $BO=1+AE$, so
$(B^2+zA^2)\,\sigma'=(BO)^2+zA^2O^2=1+A^2E^2+zA^2O^2=1+A^2\sigma$,
and $B^2+zA^2$ agrees with $\Psi$ at every root of $\sigma$. Forming $E$ and
$O$ from $\sigma$, and $B^2+zA^2$ from $(A,B)$, uses only square roots and
squares. The only costly step, the EEA, thus runs on
degree-17 rather than degree-35 operands.

The second observation concerns qubit usage. A fresh register for $\Psi$ would raise the peak
qubit count by 420. Instead, the circuit makes room inside the syndrome
register itself. As described above, the in-place product over all 70
syndrome symbols leaves $\omega$ in the first 35 and zeros in the last 35 on
every correct decode. The circuit stores $A$ and $B$
in these zeroed words and evaluates
$\Psi(\beta_j)=B(\beta_j)^2+\beta_jA(\beta_j)^2$ at each position with CNOTs
only. After the scan, it reverses the Euclidean pass, which returns these
words to zero, and then reverses the product to restore the syndromes.

Unlike the baseline's Euclidean solve, this pass follows a fixed schedule.
Every input runs the same 17 elimination blocks. Each
block parks its quotients in the coefficient words that it clears, so the pass
needs no separate history.

In total, the 4,024-position scan costs 93,000 Toffolis, and the Euclidean
pass 44,000 in each direction, compared with 1.3 million for the baseline's
scan. The scan's Toffoli count barely depends on the choice of $c$: every scanned position uses
the same 23-Toffoli circuit, and a different $c$ changes the number of scanned
positions by at most one.

The discovered circuit uses 440,000 Toffolis and 1,837
qubits, for a $TQ$-cost of $8.1\times10^8$. Relative to the Reed--Solomon
baseline, this reduces the Toffoli count by $12\times$, the qubit count by
$5.2\%$, and the $TQ$-cost by $12\times$. In 100,000 tests of the complete circuit,
seven produce the wrong sign.

\subsection{Hermitian-code decoding for HOPI}\label{sec:case-dqi-hopi}

HOPI requires the same phase~\eqref{eq:dqi-phase-fixup} as the Reed--Solomon
instance above, but its error word belongs to a 512-symbol Hermitian code over
$\mathbb F_{64}$. A Reed--Solomon code with 512 positions would require
symbols of about nine bits, since each position needs a distinct field element.
A Hermitian code instead obtains 512 positions while retaining six-bit
symbols by using the points $(x,y)\in\mathbb F_{64}^2$ on
\begin{equation}
  y^8=x^9+y.
  \label{eq:hermitian-curve}
\end{equation}
Polynomials in $x$ and $y$ are ordered by a \emph{weighted degree}, the
analogue of polynomial degree for Reed--Solomon codes. Here $x$
has weight 8 and $y$ has weight 9, so the monomial $x^ay^b$ has weighted degree
$8a+9b$. These weights give the leading terms $y^8$ and $x^9$ of
\cref{eq:hermitian-curve} the same weighted degree, 72.
Each finite point $P_j=(x_j,y_j)$ is one code position. The syndrome
construction is the two-variable analogue of the Reed--Solomon map above. A function
$\psi(x,y)$ assigns a field value to every position and therefore defines one
parity check
\begin{equation}
  S(\psi)=\sum_{j=1}^{512}e_j\psi(P_j).
  \label{eq:hopi-generalized-syndrome}
\end{equation}
The value $S(\psi)$ is one syndrome symbol. This instance uses the 123
monomials
$\psi_{a,b}(x,y)=x^ay^b$ of weighted degree
$8a+9b\le150$~\cite{gu2025hermitian-dqi}. We enumerate them as
$\psi_0,\ldots,\psi_{122}$ in increasing weighted degree~\footnote{Here $a\ge0$ and $0\le b\le7$.}. The input $\ket{s}$ supplies the
123 symbols
\begin{equation}
  S_k:=S(\psi_k), \qquad k=0,\ldots,122.
\end{equation}
This syndrome map defines a $[512,389,d_{\min}\ge96]_{64}$ Hermitian code, so the syndrome uniquely determines every error affecting at most
47 positions. The 123 six-bit syndromes occupy 738 qubits. \Cref{tab:hopi-dictionary} lists the
correspondences with the Reed--Solomon decoder that the rest of this
subsection uses.

\begin{table}[htbp]
\caption{Correspondence between the Reed--Solomon and Hermitian decoders.}
\label{tab:hopi-dictionary}
\centering
\begin{tabular}{lll}
\toprule
 & \textbf{Reed--Solomon} & \textbf{Hermitian} \\
\midrule
Positions & 4,095 nonzero field elements & 512 points $(x,y)$ on the curve \\
Symbol size & 12 bits & 6 bits \\
Check functions & $z^i$ & $x^ay^b$ \\
Ordered by & degree $i$ & weighted degree $8a+9b$ \\
Syndromes & 70 & 123 \\
Locator & polynomial $\sigma(z)$ & function $\sigma(x,y)$ \\
Evaluator & $\omega=\sigma S$ & $\omega=\sigma\widehat S$ \\
Error value & $\kappa_j\,\omega/\sigma'$ & $\omega/\sigma'$ \\
\bottomrule
\end{tabular}
\end{table}

As in Reed--Solomon decoding, the first task is to find a locator $\sigma$ that
vanishes at every error point; here it is a function of two variables on
the curve. Multiplying each root condition by $e_j\psi_{k'}(P_j)$ and summing over
positions gives, for $\sigma=\sum_k\sigma_k\psi_k$,
\begin{equation}
  0=S(\sigma\psi_{k'})
   =\sum_k\sigma_k S(\psi_k\psi_{k'}).
  \label{eq:hopi-locator-recurrence}
\end{equation}
This has the same form as the Reed--Solomon recurrence
\eqref{eq:rs-syndrome-recurrence}: a sum of syndromes weighted by the locator
coefficients. Whereas multiplying Reed--Solomon
monomials simply shifts the syndrome index, the curve relation
\eqref{eq:hermitian-curve} reduces each product $\psi_k\psi_{k'}$ to one or two
basis monomials. Beyond error weight 33, some of the syndromes that the recurrence needs have weighted
degree above 150, which the input does not supply. The Feng--Rao baseline reconstructs these
missing syndromes by majority voting among recurrence predictions, then uses
the extended sequence to locate and evaluate the errors~\cite{sakata1995feng-rao}.

Our reversible implementation of this baseline uses 100 million Toffolis at a peak of
190,000 qubits, for a $TQ$-cost of
$2.0\times10^{13}$. The
large qubit count comes from retaining the intermediate states of the recurrence updates until the circuit reverses the computation.

The discovered circuit instead turns locator finding into a null-space problem
for a symmetric matrix $A$ built directly from the supplied
syndromes. It never stores $A$; instead, it computes the effect of multiplying
a vector by $A$ directly from the syndromes. Repeating this operation produces
a sequence of vectors that must eventually become linearly dependent, and
such a dependence yields the desired null vector. The circuit processes this
sequence incrementally, retaining only a few vectors at a time and grouping
several consecutive vectors when a vector cannot be processed alone. Once
the null vector supplies a locator $\sigma$, the circuit still avoids constructing the
whole evaluator $\omega$ or the decoded error word. It scans the curve points in small batches called \emph{stripes},
combines the locator and syndromes only for the current stripe, applies those
points' phase contributions, and clears the temporary data before continuing.
The circuit thus has two stages: a locator walk that never stores the
matrix, and a space-efficient phase scan. The walk freezes on the rare inputs that need
extra work and repairs them at scheduled rounds. Its reversal recomputes most
decisions and stores one bit per round.

\paragraph{The locator as a null vector.}
The discovered circuit looks for the locator in a fixed space that is large
enough to always contain one. Following the linear-algebra strategy of
\citet{skorobogatov1990decoding}, it uses the span of the 48 monomials
$\psi_0,\ldots,\psi_{47}$ of weighted degree at most 75. Writing
$\sigma=\sum_{k=0}^{47}\sigma_k\psi_k$, each error point imposes one linear
condition $\sum_{k=0}^{47}\sigma_k \psi_k(P_j)=0$ on the 48 coefficients.
There are at most 47 conditions on 48 unknowns, so a nonzero locator always
exists.

The error points are unknown, so the circuit cannot impose these
conditions directly. It replaces them with a condition that it can compute
from the syndromes. Define the symmetric $48\times48$ matrix
$A_{kk'}=S(\psi_k\psi_{k'})$. Each product $\psi_k\psi_{k'}$ has weighted degree at most 150,
so every entry of $A$ is a linear function of the supplied syndromes. The coefficient vector
$\boldsymbol{\sigma}=(\sigma_0,\ldots,\sigma_{47})^{\mathsf T}$ of every
locator lies in the null space of $A$, and almost always the converse holds
as well.

To see this, let $\mathcal E$ be the set of error points and let
$M_{jk}=\psi_k(P_j)$ for $P_j\in\mathcal E$. Then
\begin{equation}
  A=M^{\mathsf T}\operatorname{diag}(e_j)\,M,\qquad
  M\boldsymbol{\sigma}=(\sigma(P_j):P_j\in\mathcal E)^{\mathsf T}.
  \label{eq:hopi-null-vector}
\end{equation}
A locator has $M\boldsymbol{\sigma}=0$ and therefore $A\boldsymbol{\sigma}=0$.
Conversely, when the conditions $\sigma(P_j)=0$ are independent, $M$ has full
row rank. The error values are nonzero, so $A\boldsymbol{\sigma}=0$ then
implies $M\boldsymbol{\sigma}=0$, and $\sigma$ is a locator. The conditions
are dependent with probability $\approx8\times10^{-8}$ under the DQI input
distribution.

\paragraph{Applying $A$ without storing it.}
The circuit needs only the product $A\boldsymbol u$ for
a given vector $\boldsymbol u$. Computing it directly would take
$48^2=2{,}304$ field multiplications. Let $U=\sum_k u_k\psi_k$ be the function
with coefficient vector $\boldsymbol u$. Then $(A\boldsymbol u)_{k'}=S(\psi_{k'}U)$:
each entry is the syndrome of a product of two functions. Every such product
lies in the span of the 123 monomials, so its values at 123 suitably chosen
curve points $R_1,\ldots,R_{123}$ determine it. On this span the syndrome is
therefore a weighted sum of point values,
$S(f)=\sum_{i=1}^{123}\eta_i f(R_i)$, with weights $\eta_i$ that are a fixed
linear function of the supplied syndromes. Writing
$\Phi_{ik}=\psi_k(R_i)$ gives
\begin{equation}
  A\boldsymbol{u}
  =\Phi^{\mathsf T}\operatorname{diag}(\boldsymbol{\eta})\Phi\boldsymbol{u}.
  \label{eq:hopi-walk-factorization}
\end{equation}
Here $\Phi\boldsymbol{u}$ evaluates $U$ at the points, the diagonal map
multiplies each value by its weight, and $\Phi^{\mathsf T}$ returns the
weighted values to the 48 coefficients. The maps $\Phi$ and
$\Phi^{\mathsf T}$ are fixed and linear, so they cost only CNOTs. The only
Toffolis are in the 123 pointwise products, each a 15-Toffoli multiplication
in $\mathbb F_{64}$~\cite{toom1963complexity,cenk2010multiplication}. One
application of $A$ thus costs $123\times15=1{,}845$ Toffolis. The same
pattern recurs below, and we call it \emph{evaluation--interpolation}: the
circuit evaluates at fixed points, multiplies the values pointwise, and maps
the products back linearly.

\paragraph{Lanczos recurrence.}
A null vector of $A$ can be found from repeated applications of $A$. Starting
from a fixed seed $\boldsymbol{y}_0$, generate the \emph{Krylov vectors}
$\boldsymbol{y}_0,A\boldsymbol{y}_0,A^2\boldsymbol{y}_0,\ldots$. They all lie
in a 48-dimensional space, so they must become linearly dependent: there are
coefficients, not all zero, with
$\lambda_0\boldsymbol{y}_0+\lambda_1A\boldsymbol{y}_0+\cdots+
\lambda_NA^N\boldsymbol{y}_0=0$. If the seed is not orthogonal to the entire
null space, symmetry forces $\lambda_0=0$. To see this, choose
$\boldsymbol{v}$ in the null space of $A$ with $\boldsymbol{v}^{\mathsf T}\boldsymbol{y}_0\ne0$.
Symmetry gives $\boldsymbol{v}^{\mathsf T}A^r\boldsymbol{y}_0=0$ for $r\ge1$,
so the inner product of the dependence with $\boldsymbol v$ leaves
$\lambda_0\boldsymbol{v}^{\mathsf T}\boldsymbol{y}_0=0$. Then
\begin{equation}
  \boldsymbol{\sigma}=\sum_{r=1}^N \lambda_rA^{r-1}\boldsymbol{y}_0,
  \qquad
  A\boldsymbol{\sigma}=\sum_{r=1}^N \lambda_rA^r\boldsymbol{y}_0=0
  \label{eq:hopi-walk-kernel}
\end{equation}
is the desired null vector. In the rare case where $\lambda_0\ne0$, the
decoding fails.

Used directly, this would require storing every Krylov vector until the
dependence appears. The Lanczos algorithm avoids that. It maintains an
orthogonal basis $\boldsymbol{w}_0,\boldsymbol{w}_1,\ldots$ whose first $r+1$ vectors span
the first $r+1$ Krylov vectors: it applies $A$ to the newest basis vector and subtracts the
components along the existing basis. Because $A$ is symmetric,
$A\boldsymbol{w}_r$ is already orthogonal to every basis vector except
$\boldsymbol{w}_{r-1}$ and $\boldsymbol{w}_r$, so
\begin{equation}
  A\boldsymbol{w}_r\in
  \operatorname{span}\{\boldsymbol{w}_{r-1},\boldsymbol{w}_r,
  \boldsymbol{w}_{r+1}\}.
  \label{eq:hopi-three-term}
\end{equation}
Constructing the next vector therefore needs only the current and the
preceding basis vectors. A zero result means that the Krylov vectors have
become dependent.

The null vector can be built in the same pass. By \cref{eq:hopi-three-term},
the coefficients $\xi_r$ of a null vector
$\boldsymbol{\sigma}=\sum_r\xi_r\boldsymbol{w}_r$ satisfy a three-term
recurrence of their own, so each follows from the two before it. The circuit
computes each coefficient while its basis vector is still stored and
immediately adds $\xi_r\boldsymbol{w}_r$ into a 48-symbol register, the
\emph{locator register}. When the
walk ends, that register holds $\boldsymbol{\sigma}$.

\paragraph{Breakdowns and scheduled repairs.}
Lanczos removes the component of a vector $\boldsymbol v$ along a basis
vector $\boldsymbol w$ by subtracting $\rho\boldsymbol w$, with
$\rho=(\boldsymbol w^{\mathsf T}\boldsymbol v)/
(\boldsymbol w^{\mathsf T}\boldsymbol w)$. Over $\mathbb F_{64}$, however, a
nonzero vector can be orthogonal to itself, which leaves this quotient
undefined. Block Lanczos repairs such a breakdown by treating the consecutive
Krylov vectors $\boldsymbol w,A\boldsymbol w,\ldots,A^{d-1}\boldsymbol w$ as
one block and removing the components along all of them at once. This is possible once their Gram
matrix, the matrix of their inner products, is invertible:
\begin{equation}
  G_d=\bigl[(A^i\boldsymbol w)^{\mathsf T}
             (A^j\boldsymbol w)\bigr]_{i,j=0}^{d-1}
      =\bigl[g_{i+j}\bigr]_{i,j=0}^{d-1},
  \qquad g_k=\boldsymbol w^{\mathsf T}A^k\boldsymbol w.
  \label{eq:hopi-walk-block-pairing}
\end{equation}
We call $d$ the \emph{depth} of the block. The circuit supports depths two
and three.

A breakdown is rare, and the machinery that repairs it is expensive. An
ordinary round of the walk costs 5,400 Toffolis; a round that can also repair
at depth two costs 12,000, and one that can repair at depth three costs
17,000. A fixed circuit that repaired each breakdown where it occurred would
carry this machinery in every round. The circuit instead concentrates it in
scheduled \emph{repair rounds}. It conditions the ordinary update on
$g_0=\boldsymbol w^{\mathsf T}\boldsymbol w\ne0$. This test is cheap: in
characteristic 2, $\boldsymbol w^{\mathsf T}\boldsymbol w$ is the square of the
sum of the components of $\boldsymbol w$, which needs only CNOTs. On an input with $g_0=0$ the update gates do not act,
and the walk state stays frozen until the next repair round, which forms the
block and advances the walk past it in one step. The walk has 52 rounds, four more than the dimension, so that delayed
inputs can finish. Of these, 17 carry repair machinery. An input that needs a deeper block, or that
has not finished after 52 rounds, fails.

\paragraph{Ending and reversing the walk.}
For a given input, the walk ends when $g_0=g_1=g_2=0$, because no block of depth up
to three is then invertible. The circuit freezes the walk there and records the
round, as the divstep circuit does
(SI~\ref*{sec:modinv-divsteps}). After the last round, it tests which of three
cases holds and routes the corresponding vector into the locator register. If
$\boldsymbol w=0$, the Krylov vectors have become dependent, and the
accumulated vector is the null vector. If $\boldsymbol w\ne0$ and
$A\boldsymbol w=0$, then $\boldsymbol w$ itself is a null vector. If
$A\boldsymbol w\ne0$ and $A^2\boldsymbol w=0$, then $A\boldsymbol w$ is a
null vector. An input in none of these cases fails.

After the phase scan described below, the circuit must undo the walk. It
first measures the locator register in the $X$ basis. The walk built that
register by adding multiples of basis vectors, so the reverse walk can
correct the phase left by the measurement as it retraces the rounds: each
addition becomes a set of controlled-$Z$ gates conditioned on the measurement
outcomes, at no Toffoli cost. The reverse walk recomputes most of the
decisions of the forward walk, such as its zero tests, while the vectors they
depend on are present. One decision per round cannot be recomputed in time:
whether the round advanced the walk, which controls an exchange of the two
stored basis vectors. The forward walk stores that decision in one history
qubit per round, 52 in all, and the reverse walk clears them. The reverse walk
uses 52,000 fewer Toffolis than the forward walk.

\paragraph{Error values.}
The locator gives the possible error points, and the phase also needs the
error values. The circuit uses two facts, both counterparts of the
Reed--Solomon decoder (\cref{tab:hopi-dictionary}). First, there is an
evaluator $\omega$ whose coefficients are bilinear in the locator and the
syndromes, as $\omega=\sigma S$ is for Reed--Solomon codes. Second, the error
value at a simple zero of the locator is
\begin{equation}
  e_j=\frac{\omega(P_j)}{\sigma'(P_j)},
  \qquad \sigma'=\partial_x\sigma+x^8\partial_y\sigma,
  \label{eq:hopi-simple-residue}
\end{equation}
the counterpart of Forney's formula~\eqref{eq:rs-forney}. Here $\sigma'$ is
the derivative of the locator along the curve. Differentiating
\cref{eq:hermitian-curve} in characteristic 2 gives $dy=x^8dx$, so a step $dx$
along the curve changes $\sigma$ by $\sigma'\,dx$.
The evaluator is $\omega=\sigma\widehat S$, where $\widehat S$ is the
counterpart of the syndrome polynomial $S(z)$: a series whose leading
coefficients, in decreasing weighted degree, are the 123
syndromes~\cite{stichtenoth2009function-fields}.

The locator comes from a 48-dimensional space and is not necessarily the
minimal-degree locator. It can therefore also vanish at a point that carries no error,
or have a repeated zero at an error point. At a point with no error,
$\omega$ vanishes as well and the point contributes nothing. At a repeated
zero the quotient becomes $0/0$. As in L'H\^opital's rule, the error value
then comes from the first few coefficients of the expansions of $\omega$ and
$\sigma$ at that point, which we call \emph{local coefficients}. The circuit
supports repeated zeros up to multiplicity three.

\paragraph{Striped phase scan.}
For a curve point $P_j$, write $c_j$ for its six mask bits and $e_j$ for its
error value. As in the Reed--Solomon decoder, the mask bits define
$\mu_j\in\mathbb F_{64}$ with
$\langle c_j,e_j\rangle_2=
\Tr_{\mathbb F_{64}/\mathbb F_2}(\mu_je_j)$, and part of the phase is linear
in the syndromes and costs only $Z$ gates. The residual constant
$\widetilde\mu_j$ vanishes at 128 points, which leaves 384 points to scan.

At each of these points the circuit needs three local coefficients of
$\omega$. Computing
$\omega$ once and storing it would hold more than 100 field symbols through
the whole scan. The circuit instead works in 24 stripes of 16 points.
For each stripe it computes only the $16\times3=48$ symbols that the stripe
needs, uses them, and erases them by MBU at no Toffoli cost. The
matching local coefficients of the locator are linear in its coefficient
vector and cost only CNOTs.

Recomputing for every stripe is affordable only because $\omega$ is
cheap to form. A direct construction multiplies the expansion of
$\widehat S$ by $\sigma$ coefficient by coefficient, which uses 2,452 field
multiplications, or 37,000 Toffolis. The circuit uses
evaluation--interpolation instead, as in \cref{eq:hopi-walk-factorization}.
One obstacle is that the syndromes give $\widehat S$ only as the leading
coefficients of an expansion, which cannot be evaluated at a curve point. The
circuit therefore replaces $\widehat S$ by a function that can be evaluated:
a fraction $h_{\boldsymbol s}=n_{\boldsymbol s}/q$ of two polynomials in $x$
and $y$, with a fixed denominator and a numerator that is linear in the
syndromes $\boldsymbol s=(S_0,\ldots,S_{122})^{\mathsf T}$. The fraction is
chosen so that its product with any locator agrees with $\sigma\widehat S$ in
every coefficient that the scan uses. All products $\sigma h_{\boldsymbol s}$
lie in a fixed 206-dimensional space, so their values at 206 suitably chosen
curve points determine those coefficients, at $206\times15=3{,}090$ Toffolis
per stripe. The 24 stripes
cost 74,000 Toffolis in total. This is more than the 37,000 Toffolis of one
direct construction, but it avoids storing $\omega$.

At each point, the circuit then forms the phase bit
$\Tr(\widetilde\mu_j\,\omega/\sigma')$ from the local coefficients. Division
in $\mathbb F_{64}$ is cheap: the circuit computes the inverse with arithmetic in the smaller
field $\mathbb F_8$, and the whole computation costs 33 Toffolis per point.

On 100,000 random inputs, 30 produce the
wrong sign. The circuit uses 930,000 Toffolis and 2,395 qubits, for a $TQ$-cost
of $2.2\times10^9$. Relative to the Feng--Rao baseline, this reduces
the Toffoli count by $110\times$, the qubit count by $81\times$, and the $TQ$-cost by
$9{,}100\times$.

\section{Prime-field modular inversion}\label{sec:case-modinv}

We benchmark modular inversion within a complete arithmetic circuit: point
addition on an elliptic curve over a 256-bit prime field. Each point on this curve
is a pair of integers $(x,y)$ satisfying the curve equation
\begin{equation}
  y^2\equiv x^3+7\pmod p,
  \label{eq:modinv-curve}
\end{equation}
where $p=2^{256}-2^{32}-977$ is a 256-bit prime number. The curve defines \emph{point addition}, an operation that combines two points $P$
and $Q$ into a third point $R=P+Q$ on the same curve (note that this addition is not the usual coordinate-wise addition). Geometrically, one draws the line through $P$ and $Q$, finds its third
intersection with the curve, and reflects that point across the horizontal
axis. Algebraically, for the inputs considered here, with $Q_x\ne P_x$, one first
computes the slope
\begin{equation}
  \lambda\equiv(Q_y-P_y)(Q_x-P_x)^{-1}\pmod p,
  \label{eq:modinv-slope}
\end{equation}
where $(Q_x-P_x)^{-1}$ denotes the inverse modulo $p$. The coordinates of the sum are then
\begin{equation}
  R_x=\lambda^2-P_x-Q_x,\qquad
  R_y=\lambda(P_x-R_x)-P_y.
\end{equation}

We study the reversible point adder~\cite{proos2003quantum}: given an
input point $P$ and another, classically supplied point $Q$, it performs
\begin{equation}
  \ket{P}\longmapsto\ket{P+Q}.
  \label{eq:modinv-point-add-task}
\end{equation}
The modular inverse required by \cref{eq:modinv-slope} is the dominant cost of point addition: once the slope $\lambda$ is known, the remaining coordinate
updates require only a fixed sequence of relatively cheap modular additions and
multiplications. In contrast, modular inversion follows a long input-dependent sequence of
arithmetic transitions, which the circuit must later recover to reverse the
computation and clear its ancillas.

Prior point-addition
constructions~\cite{babbush2026kickmix,schrottenloher2026pointaddition}
use the binary greatest common divisor (gcd) algorithm for modular inversion.
The connection between gcd and inversion comes from expressing the gcd in
terms of the original inputs. Because $p$ is prime and $w=Q_x-P_x\ne0$, there are
integers $\alpha$ and $\beta$ such that
$\gcd(p,w)=1=\alpha p+\beta w$.
Reducing this identity modulo $p$ gives $\beta w\equiv1$, so
$\beta\equiv w^{-1}\pmod p$. The extended binary-gcd algorithm finds these
coefficients by tracking expressions for the two current values $u$ and $v$ as integer
combinations of $p$ and $w$. Binary gcd rests on two elementary identities. If $u$ is even and $v$ is odd, then
$\gcd(u,v)=\gcd(u/2,v)$; if $u$ and $v$ are both odd, then, for $u>v$,
$\gcd(u,v)=\gcd((u-v)/2,v)$. Applying these identities symmetrically gives
the reduction
\begin{equation}
  (u,v)\longmapsto
  \begin{cases}
    (u/2,v) & \text{for } u\text{ even},\\
    (u,v/2) & \text{for } v\text{ even},\\
    ((u-v)/2,v) & \text{for } u>v,\quad u,v\text{ odd},\\
    (u,(v-u)/2) & \text{for } v>u,\quad u,v\text{ odd}.
  \end{cases}
  \label{eq:modinv-binary-gcd}
\end{equation}
Each case preserves the gcd and makes one value smaller.

Making this reduction reversible is difficult because its control flow is
input-dependent: which line of \cref{eq:modinv-binary-gcd} applies changes from
one step to the next, and the reduced pair does not in general reveal which
line produced it. A quantum circuit therefore cannot simply run the reduction
backward to restore its input and clear its workspace; it must preserve enough
information to identify each step. The prior constructions do so using the
\emph{dialog} introduced by \citet{khattar2025optimized-dqi}. The direct
approach would update the coefficients that become $\alpha$ and $\beta$
alongside the two values throughout the reduction. The dialog avoids holding
these growing coefficients in memory. Its first pass updates only the pair
being reduced and records a compact symbol identifying each subtraction,
swap, and halving choice. A second pass replays the same transformations to
construct the required coefficient. Initializing this replay from the
numerator $a=Q_y-P_y$ produces $a\beta=aw^{-1}\bmod p$ directly, without ever storing
$w^{-1}$ in a separate register. The simple binary updates and the replay trick
reduce the gate count, but the complete dialog record
requires hundreds of additional qubits. \Cref{tab:modinv-frontier} compares
these constructions with the two circuits reported below.

\begin{table}[htbp]
\caption{Resource comparison for point addition. All circuits use the
same affine-coordinate input and output; counts include reversible cleanup but exclude
application-specific storage or lookup used to supply the classical addend.}
\label{tab:modinv-frontier}
\centering
\begin{tabular}{@{}llrrr@{}}
\toprule
\textbf{Construction} & \textbf{Inversion strategy} & \textbf{Toffolis} &
\textbf{Qubits} & $\boldsymbol{TQ}$ \\
\midrule
\citet{schrottenloher2026pointaddition} & Binary gcd and dialog trick~\cite{khattar2025optimized-dqi} & $1.9{\cdot}10^6$ & 1,446 & $2.7{\cdot}10^9$ \\
\textbf{Divstep circuit,} SI~\ref*{sec:modinv-divsteps} & \textbf{Divsteps with decoding} & $\boldsymbol{4.0{\cdot}10^6}$ & \textbf{935} & $\boldsymbol{3.7{\cdot}10^9}$ \\
\citet{luo2026space} & Packed Euclidean walk & $7.0{\cdot}10^7$ & 838 & $5.8{\cdot}10^{10}$ \\
\textbf{Specialized Euclidean,} SI~\ref*{sec:modinv-eea-specialized} & \textbf{Specialized packed walk} & $\boldsymbol{2.1{\cdot}10^7}$ & \textbf{843} & $\boldsymbol{1.7{\cdot}10^{10}}$ \\
\bottomrule
\end{tabular}
\end{table}

\subsection{Divsteps}\label{sec:modinv-divsteps}

Rather than reorganize the binary-gcd constructions above, the search finds a new reversible
prime-field inversion architecture built on
\emph{divsteps}~\cite{bernstein2019divsteps,safegcd-explained}. In the divstep
recurrence of \cref{eq:divstep-recurrence}, a small signed counter $z$, called the half
delta, replaces the repeated magnitude comparisons in
\cref{eq:modinv-binary-gcd}: the sign of $z$ and the parity of one value alone
determine the next update. We call the choice of update at a step its
\emph{decision}. \Cref{app:divstep-recovery} proves that every decision can
be computed exactly from the registers as they stand after the step. When a decision
is cheap to decode, the circuit decodes it from the updated state. Otherwise,
the circuit stores the decision as a single bit on a fixed-size history stack.
We call the sequence of divsteps for one input a \emph{walk}.
The point adder uses this inversion twice, once to form the slope $\lambda$
and once to erase it after the coordinate update. Each use runs a forward
walk and a reverse walk, for four walks in total.

Unlike the quotient-based extended Euclidean algorithm (EEA) of
SI~\ref*{sec:modinv-eea-specialized}, which divides one remainder by the other at
every iteration, a divstep never computes a
multi-bit quotient or adds a multiple of one full register to another.
A divstep walk needs more steps than the quotient-based EEA needs divisions, but each step is
much cheaper. Reversible divsteps were previously
developed over binary fields and polynomial
rings~\cite{banegas2021binary,khattar2025optimized-dqi}. There, each step
shrinks the polynomial pair by one coefficient, and the freed space can hold
that step's record. Integer divsteps free no such fixed space, so a stored
history costs additional qubits. To our knowledge, this
is the first reversible divstep inversion circuit over a prime field.

We recall the notation of the main text. The walk state is
$(z,f,g,f_v,g_v)$, and each step applies one of the three updates in
\cref{eq:divstep-recurrence}. Two bits label the update: $s$ records whether
the step swaps the roles of $f$ and $g$, and $b$ records whether the incoming
$g$ is odd. To run a step backward, the circuit must know $s$ and $b$, and the
updated state does not hold them explicitly. Cheap tests on a few bits of the
updated state recover $s$ at every step and $b$ at about two thirds of the
steps. We call the remaining decisions \emph{ambiguous}. Each ambiguous
decision has a \emph{depth} $j$, the number of trailing zero bits of the
updated coefficient $f_v'$, and we write $f_v'=2^jq$ with $q$ odd. A swap
followed by a run of ambiguous decisions is an \emph{excursion}, and the depth
counts the steps of the excursion so far, with the swap as step 1. The cost of a
depth-$j$ decoder grows with $j$. The cutoff $j_{\mathrm{store}}$ is the depth from
which the circuit stops decoding: it decodes an ambiguous decision at depth
$j<j_{\mathrm{store}}$ and stores one at depth $j\ge j_{\mathrm{store}}$.

The cutoff $j_{\mathrm{store}}$ trades qubits for Toffolis: a higher cutoff decodes more and stores less. At $j_{\mathrm{store}}=0$, the circuit uses only the cheap tests and
stores every remaining decision, requiring 198 history qubits for the 533-step
schedule and 1,065 qubits in total. The discovered circuit chooses the
intermediate cutoff $j_{\mathrm{store}}=5$: it decodes depths two through four
with the rules given below. Storing only the deeper decisions shrinks the
history stack from 198 to 60 qubits. At the low-qubit
extreme, $j_{\mathrm{store}}=9$ decodes through $j=8$ and reserves 24 history
qubits.
Varying the cutoff and the workspace produces a family of circuits that spans from
3.6 million Toffolis and 1,031 qubits to 7.9 million Toffolis and 893 qubits; an additional
higher-qubit configuration uses 3.3 million Toffolis and 1,191 qubits.

\paragraph{Branch decoders.}
A forward step computes $b$ from the parity of $g$ and uses it to control the
update. After the update, $g$ no longer holds that parity, so the step is left
with $b$ in a qubit of its own. Uncomputing $b$ there would require the
decoder on the forward walk as well. The forward step instead disposes of
$b$ in one of two ways. If the decision is ambiguous at depth
$j\ge j_{\mathrm{store}}$, the step pushes $b$ onto the history stack.
Otherwise, the step removes $b$ by MBU: it measures $b$ in the $X$ basis. The measurement frees the
qubit but leaves a phase that depends on $b$ and on the measurement outcome.

The reverse walk is the uncomputation of the forward walk. Once the circuit
has used the inverse, it undoes the forward steps in the opposite
order, which restores $w$ and returns the walk registers to zero. The matching reverse step obtains $b$
again, either by popping the stack or by
applying the tests and decoders to the updated state. It applies the
phase correction for the measurement outcome and uses $b$ to choose which update
to undo. The restored $g$ then holds $b$ as its parity, so a CNOT from the low bit of $g$ clears it.
Decoders therefore appear only in reverse steps.
A walk that stores more decisions than the stack holds overflows it and
fails.

The decoders evaluate the second form of \cref{eq:app-divstep-ladder},
$b=\operatorname{bit}_{j-1}(p(f'q)^{-1}\bmod2^j)$. Both $f'$ and $q$ are odd, so
the depth-$j$ decoder is a Boolean function of $2(j-1)$ bits. Write
$f_i=f'[i]$ and $q_i=q[i]=f_v'[j+i]$ for $1\le i<j$, where $u[i]$ is bit $i$ of $u$. For the prime $p$ of
\cref{eq:modinv-curve}, the decoders at depths two through five are
\begin{subequations}\label{eq:divstep-decoder-rules}
\begin{align}
  j=2:\quad b&=1\oplus f_1\oplus q_1,\\
  j=3:\quad b&=1\oplus f_2\oplus q_2,\\
  j=4:\quad b&=1\oplus f_1\oplus f_2\oplus f_3\oplus q_1\oplus q_2\oplus q_3
    \oplus(f_1\oplus f_2)(q_1\oplus q_2),\\
  j=5:\quad b&=f_3\oplus f_4\oplus q_3\oplus q_4\oplus M_f\oplus M_q
    \oplus f_1q_1\oplus f_2q_3\oplus f_3q_2\oplus f_3q_3\notag\\
    &\quad\oplus M_f(q_1\oplus q_2)\oplus(f_1\oplus f_2)M_q,
\end{align}
\end{subequations}
where $M_f=f_1f_2\oplus f_1f_3\oplus f_2f_3$ is the majority of
$f_1,f_2,f_3$ and $M_q$ is the majority of $q_1,q_2,q_3$. Written out as
sums of products, these four rules have 3, 3, 11, and 26 terms. The decoders
at depths six, seven, and eight have 79, 268, and 1,092 terms. The growth
reflects \cref{eq:app-divstep-shift-register}: the deciding bit sits above
the bit of every earlier decision in the excursion, and the multiplication by $q^{-1}$
carries through all of them.

The circuit implements equivalent forms of these rules. The determinant
identity of \cref{eq:app-divstep-det} fixes $g_v'\bmod2^j$ once $f'$ is known, so
a decoder may read a bit of $g_v'$ in place of bits of $f'$. The depth-two
decoder is \cref{eq:app-divstep-j2}, which costs two Toffolis. At depths
three and four, the rule is an XOR of bits of $q$, with no products, once
$f'\bmod16$ is fixed. From depth five on, the
rule contains products of bits of $q$.

These rules explain where the cutoff sits. Every reverse step must contain
the decoder for every depth below $j_{\mathrm{store}}$, whether or not that
step needs it. Each decoder is conditioned on a flag marking the
depth, so each term that reads a register bit costs at least one Toffoli. The benefit shrinks as the
cost grows: among ambiguous decisions in sampled walks, $50.1\%$ occur at depth
two, $25.0\%$ at depth three, $12.5\%$ at depth four, and $12.4\%$ at all
greater depths. Beyond depth four, each further depth roughly triples the
number of terms and removes half as much history as the depth before it.

\paragraph{Early stopping.}
A walk terminates when $g$ reaches zero: at that point $f=\pm1$ and the
inverse can be read from $f_v$. Different inputs terminate at different steps,
typically between steps 505 and 530 (\cref{fig:divstep}(b)), but the circuit
runs a fixed schedule of 533 steps. The steps after termination must leave
the state unchanged. Left alone, the recurrence would treat $g=0$ as an even
value, so it would keep doubling $f_v$ and incrementing $z$ until both outgrow
their registers.

The circuit therefore freezes a terminated walk. Each forward step tests
whether $g=0$. On the first step where the test succeeds, the circuit sets a
one-qubit frozen flag, and from then on the flag disables every update. The
reverse walk must clear the flag at that same step. A frozen state looks
identical at every later step, so the state alone cannot reveal when the walk
froze. The circuit writes the step index into a small register
when it sets the flag. Each reverse step compares the register with its own
index, and a match clears both the flag and the register.

This machinery costs Toffolis at every step: a test that every bit of
$g$ is zero, and one more control on every update. The test is most expensive in the
early steps, where $g$ is widest and termination is least likely. The first
498 steps omit the test, the flag, and the flag's controls. A walk
that terminates within those steps fails. The measured probability is
$1.4\times10^{-6}$ per inversion, and \cref{tab:divstep-failures} uses
$3\times10^{-6}$ as an upper bound.

A walk whose counter $z$ outgrows its register also fails, with probability
$9\times10^{-8}$ per inversion; the two events give the early-stopping row of
\cref{tab:divstep-failures}.

\paragraph{Removing the power-of-two scale.}
At termination, $ff_v$ equals the desired inverse times $2^t$, where $t$
counts the active steps, those before
the walk freezes (\cref{prop:app-divstep-scale}):
\begin{equation}
  f f_v\equiv2^t w^{-1}\pmod p.
  \label{eq:divstep-scaled-inverse}
\end{equation}
The stopping time $t$ depends on the input, so different branches of a
superposition carry different scales. Classical implementations avoid this by
running a fixed number of steps, which makes the scale a known constant; the
circuit instead freezes when $g=0$. A direct approach would count $t$ in a register, apply a
multiplication by $2^{-t}$ controlled on that register, and then uncompute the
count.

The circuit avoids this count. It sets $r=aff_v\equiv2^taw^{-1}\pmod p$, where
$a$ is the field value to be multiplied by the inverse. For the
slope of \cref{eq:modinv-slope}, $w=Q_x-P_x$ is the denominator and $a=Q_y-P_y$
is the numerator. The reverse walk then
undoes exactly $t$ active steps, and the frozen flag marks which steps are
active. The circuit attaches one modular halving of $r$ to each active reverse
step: $r\mapsto r/2$ if $r$ is even and $r\mapsto(r+p)/2$ otherwise. After the reverse walk, $r=aw^{-1}$, the
input $w$ is restored, and no record of $t$ remains.

\paragraph{Step-dependent qubit reuse.}
A na\"ive layout would give each of the four signed integers
$f,g,f_v,g_v$ its own permanent 256-bit register, consuming 1,024 qubits before
the counter, controls, or history stack are included. These registers do not
all need their full capacity simultaneously: as the walk progresses, $f$ and
$g$ need progressively fewer bits, while $f_v$ and $g_v$
need more. Each step halves $g$, and swaps alternate which value plays the
role of $g$. Each step also doubles a coefficient. The packed Euclidean construction of
\citet{luo2026space} and its specialization in
SI~\ref*{sec:modinv-eea-specialized} use the same fact: a shrinking remainder and
growing coefficient occupy opposite ends of one register, as visualized in
\cref{fig:modinv-walk}(a). The discovered divstep circuit applies this idea to
$f,g,f_v,g_v$ using two 256-qubit registers rather than four: $g$ and $f_v$ share the
register that initially holds $w$, and $f$ and $g_v$ share the other. Within each register, the growing coefficient claims
positions vacated by its shrinking partner. The width and position of every
value at every step are fixed when the circuit is built. A walk in which
a value outgrows its assigned width fails; this event is the integer-size
limit in \cref{tab:divstep-failures}.

\paragraph{Fused coordinate arithmetic.}
For the remaining coordinate arithmetic, the circuit does not call general
multiplication, squaring, or reduction circuits. It adds each term of the
coordinate formulas straight into the register that needs it. For the square, the
coordinate update needs $\lambda^2$ only while forming $R_x$. Splitting
$\lambda=\lambda_0+2^{128}\lambda_1$ gives
\begin{equation}
  \lambda^2=\lambda_0^2
    +2^{129}\lambda_0\lambda_1
    +2^{256}\lambda_1^2\pmod p.
\end{equation}
This expresses the 256-bit square as three products involving only 128-bit
halves. The circuit adds them one at a time directly into the coordinate
register and clears each temporary product before computing the next. It
therefore avoids both a register for the completed square and the second full
multiplication of the usual approach, which computes the square, copies it,
and uncomputes it.

The remaining coordinate products exploit the special form of the
prime $p$. Since
\begin{equation}
  2^{256}\equiv2^{32}+977\pmod p,
  \label{eq:modinv-sparse-fold}
\end{equation}
the circuit can replace any bits above bit 255 with a small correction to the low
bits instead of performing a general modular reduction. The circuit applies
this correction while it forms each product, so it never holds a complete
product that it must later clear.

We require the exact map in
\cref{eq:modinv-point-add-task}, including its phase and the return of every
ancilla to zero, except on at most a $2\times10^{-4}$ fraction of inputs~\footnote{We
draw $P$ and $Q$ independently and nearly uniformly from the points of the
curve, omitting the cases on which the affine formulas are incomplete.}. The 16-bit-window
algorithm of \citet{babbush2026kickmix} invokes the point adder 28 times; by a union
bound, the failure budget adds less than $0.6\%$ retry overhead.

The four divstep walks use 2.5 million Toffolis, $63\%$ of the circuit. The
field multiplications around the walks use 1.4 million, $35\%$ of the circuit. The fused square and the coordinate additions use the
remaining 0.10 million. The 935-qubit peak occurs during a walk. The algebraic identities behind the
recurrence, branch decoders, fused square, and sparse-prime reductions are
exact. Failures arise only from the finite
stopping-time, integer-size, history, carry, and comparison limits.

Reducing modulo the sparse prime adds a short constant, $2^{32}+977$ or a small
multiple of it, to the low end of a 256-bit register. In principle the carry from
this addition can travel up the whole register. In practice it rarely travels
far, because it keeps moving only while the bits above it are ones. The circuit
therefore propagates each such carry through only the low 58 bits (60 in the
fused square) and drops any carry that would go further. This happens with
probability about $2^{-26}$ per reduction. The comparison limit uses the same
idea: to decide whether a modular addition overflowed, the circuit compares only
the top 27 bits of the two operands, which gives the wrong answer only when those
bits all agree. \Cref{tab:divstep-failures} collects all contributions~\footnote{An exact worst-case analysis guarantees that every nonzero
input modulo $p$ terminates within 590 divsteps, compared with
the circuit's 533-step schedule~\cite{wuille2021safegcd}.
Within 533 scheduled steps, an excursion whose swap fires at counter value $z=D$ consumes $2D+1$ steps
and stores at most $D-3$ decisions, so a walk can store at most
$(533-1)/2-3=263$ decisions. The hatching in
\cref{fig:divstep}(b) extends to this ceiling.}; each entry is the per-use failure
rate times the number of uses in one point addition. Among 100,000 tests of the full circuit, four fail.

\begin{table}[htbp]
\caption{Contributions to the union bound on the point-addition
failure probability. Each entry is per point addition, which contains two
inversions. Probabilities for the walk limits are estimated from
200 million independent random 256-bit walks.}
\label{tab:divstep-failures}
\centering
\begin{tabular}{lr}
\toprule
\textbf{Failure source} & \textbf{Probability (upper bound)} \\
\midrule
Multiplier carry and comparison limits & $7.07\times10^{-5}$ \\
Stopping-time and integer-size limits & $7.56\times10^{-5}$ \\
Early-stopping and $z$-register limits & $6.18\times10^{-6}$ \\
Branch-history storage & $6.8\times10^{-6}$ \\
Fused-square carry limit & $6\times10^{-7}$ \\
\midrule
Union bound & $1.599\times10^{-4}$ \\
\bottomrule
\end{tabular}
\end{table}

\subsection{Specialized Euclidean construction}\label{sec:modinv-eea-specialized}

This subsection reports a second discovered circuit, at the low-qubit end of
the Toffoli--qubit frontier (\cref{tab:modinv-frontier}): 843 qubits against 935 for the
divstep circuit, at about five times the Toffoli count. Here the search does
not find a new inversion algorithm. It starts from the space-efficient
Euclidean point adder of \citet{luo2026space}, which we rebuilt as a
complete circuit, and specializes its schedule for the prime $p$. The
schedule is the fixed list of steps, and of bit positions within each step,
that the circuit runs.
Our reconstruction, the baseline below, uses 69.5 million Toffolis and 838 qubits. The discovered
specialization uses 843 qubits while reducing the Toffoli count
to 20.5 million, a factor of $3.4$. We use the same inputs and outputs, input
distribution, and failure budget as above.

To invert $w$ modulo $p$, the EEA starts from
$r_{-1}=p$ and $r_0=w$ and computes
\begin{equation}
  a_i=\left\lfloor\frac{r_{i-1}}{r_i}\right\rfloor,\qquad
  r_{i+1}=r_{i-1}-a_i r_i.
  \label{eq:modinv-eea-remainder}
\end{equation}
Alongside the remainders, it maintains coefficients
$t_{-1}=0,t_0=1$, with $t_{i+1}=t_{i-1}+a_i t_i$. When a remainder
reaches one, the corresponding coefficient gives the modular inverse up to
sign.

Each iteration of \cref{eq:modinv-eea-remainder} is an integer division. The
quotient $a_i$ of the two current remainders drives both updates: the
remainder update subtracts $a_i$ copies of one remainder from the other, and
the coefficient update adds $a_i$ copies of one coefficient to the other. The
quotients depend on the input and vary in size, so the circuit must compute
each one coherently before it can apply it. It performs each division by
binary long division. We call one elementary move of this division a
\emph{step}; a division whose quotient has $\ell$ bits
takes $4\ell$ steps. The total over an inversion depends on the input and is at most 1,616 for the prime
$p$, so the baseline schedules 1,616 steps.

The baseline circuit exploits the fact that the remainders $r_i$ shrink as the
coefficients $t_i$ grow. It stores a remainder and a coefficient from
opposite ends of the same register, and stores the succeeding pair
in a second such register. We call values that share a register in this
way \emph{packed}. Four short counters record where the
four values begin and end.
The circuit processes these packed values by scanning across their possible
bit positions. At each position, it uses the counters to decide whether that
bit belongs to the value currently being processed; we call this check a
\emph{counter test}. If it does, the circuit
performs the required one-bit arithmetic; otherwise it does nothing. The
fixed circuit must include this test and arithmetic for every position that
any input could require, even though a particular input uses only
a much smaller range of positions.

Each step runs several scans, one each for the remainder update, the coefficient
update, and the rearrangement of the packed values. The path
that the boundary between a remainder and a coefficient traces during an
inversion is the input's \emph{Euclidean walk}.

Each use of the inversion performs an in-place division or multiplication. For the slope, a
forward pass replaces $w$ by its inverse, one multiplication forms the
quotient in an accumulator, an $X$-basis measurement erases the old
numerator, and the reverse pass restores $w$. The second
inversion runs the same sequence in the opposite order to erase the slope.
Every step consists of in-place reversible operations on the packed registers
and counters, so the reverse pass is the inverse circuit and needs no
stored history. By contrast, the divstep circuit must
recover a decision at every step.

The discovered circuit keeps the baseline's packed register layout and reduces
the cost of its schedule mainly in three ways. First, it omits scan positions that
inputs rarely use. Second, it removes selected counter tests. Third, it
recompiles the scans of a step together so that they share intermediate
results. The first two changes fail on rare inputs, which count against the
$2\times10^{-4}$ failure budget per point addition. The third
is exact.

\paragraph{Fitted windows and the schedule cut.}
The baseline circuit includes arithmetic at every bit position required by its
worst-case bounds. Most inputs occupy a much narrower range at
each step. The specialized schedule therefore replaces each worst-case
range with a \emph{fitted window}: a range chosen separately for each step and fitted to the
positions that sampled walks use. The same fit shortens the schedule from
the worst-case 1,616 steps to 1,432 steps; we call this shortening the \emph{schedule cut}.
\Cref{fig:modinv-walk}(b) shows how closely sampled walks concentrate within
the windows, compared with the worst-case bounds.

\begin{figure}[t]
\centering
\includegraphics[width=\linewidth]{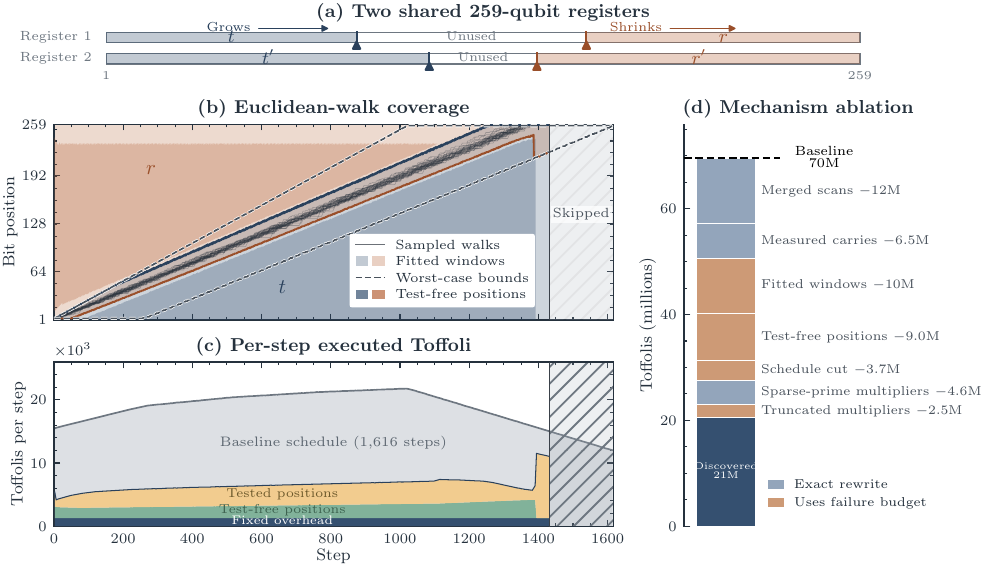}
\caption{Specializing the Euclidean schedule and its
contribution to the Toffoli reduction. (a) Two 259-qubit registers
store coefficients $t,t'$ from their low ends and remainders $r,r'$
from their high ends. Four counters mark the boundaries of these
values as the coefficients grow and the remainders shrink. (b) Each thin curve traces
the moving boundary between the shrinking remainder $r$ and growing
coefficient $t$ for one sampled walk. Light shading shows the fitted
windows retained by the discovered circuit, dashed curves show the exact
worst-case bounds, and dark shading marks the test-free positions, processed
without first testing the counters. The hatched region is removed by the
schedule cut from 1,616 to 1,432 steps. (c) Executed Toffolis per step.
The gray profile is the baseline schedule; the discovered circuit's profile splits into the
fixed overhead of a step, arithmetic at test-free positions, and arithmetic
at tested positions. (d)
Cumulative ablation from 70 million to 21 million Toffolis. Blue changes are exact, whereas brown changes count
against the failure budget.}
\label{fig:modinv-walk}
\end{figure}

\paragraph{Test-free positions.}
A fitted window still leaves a counter test at every position inside it, and
a counter test is expensive. In the baseline, a position whose arithmetic is
conditioned on a counter costs about ten Toffolis, several times the cost of
the arithmetic alone. Where every sampled input uses a position at a given
step, the schedule marks the position \emph{test-free}, and the circuit
performs its arithmetic unconditionally. Counter tests remain at the other
positions inside the window, near the moving boundaries of the values. On an
input that does not use a test-free position, the unconditional arithmetic
produces a wrong result. The two approximations remove different costs. A fitted
window deletes both the test and the arithmetic at the positions outside it.
A test-free position keeps the arithmetic and deletes the test.
\Cref{fig:modinv-walk}(c) shows the resulting per-step cost.

\paragraph{Merged scans and measured carries.}
The baseline circuit compiles the scans within a step separately, so
neighboring scans repeat some of the same comparisons and compute some of
the same control bits. Compiling the complete step as one operation removes
this repetition; we call this change \emph{merged scans}. The clearest
instance is the remainder update. The comparison that decides whether a
subtraction applies and the subtraction itself run over the same two
remainders, and the baseline compiles them as two passes. The discovered
circuit merges them into one compare-and-subtract scan whose intermediate
bits serve both operations. It then erases those bits by measurement, which
needs no Toffolis, instead of uncomputing them. The baseline also shifts a
register's contents, by an amount that depends on the input, as the values
change size. The discovered circuit renumbers the positions instead
of moving the qubits, which costs no gates. In the additions that
remain, a carry that only one position needs is likewise erased by
measurement (\emph{measured carries}). These changes are
exact.

\paragraph{Multiplier arithmetic.}
The coordinate update also requires modular squaring and multiplication. The
discovered circuit specializes these operations to the sparse prime $p$. The
exact part of the rewrite, the \emph{sparse-prime multipliers}, uses
the fold of \cref{eq:modinv-sparse-fold}, as in SI~\ref*{sec:modinv-divsteps}.
The approximate part, the \emph{truncated multipliers}, limits how far a carry
propagates and how many bits a comparison inspects.

\Cref{fig:modinv-walk}(d) shows the cumulative Toffoli reduction. The two
inversions, including their forward and reverse passes, together account for
$98\%$ of the Toffolis.

We estimate the failure probability empirically. Across
10 million sampled walks, 72 use a position outside a fitted window and 258 do
not use a test-free position, a combined rate of about $6.6\times10^{-5}$ per point addition. The schedule cut
adds $1.3\times10^{-5}$, and the truncated multipliers add at most
$2\times10^{-6}$. The total is about $8\times10^{-5}$ per point addition. Among 100,000 tests of the full circuit, four fail.

\section{Numerical function evaluation}\label{sec:function-evaluation}

Elementary-function evaluation has been a central problem in numerical
computing since the earliest mathematical software libraries. The routines
conventionally collected in \texttt{libm}, the mathematical-function library
supplied by C and Unix systems, evaluate roots, exponentials, logarithms, and
trigonometric functions by combining range reduction, table lookup,
polynomial approximation, and iterative refinement~\cite{cody1980software}.
These functions also appear throughout quantum algorithms. They determine
controlled rotation angles in quantum linear-systems
algorithms~\cite{harrow2009linear,haner2018optimizing}, set interaction phases in
real-space Hamiltonian
simulation~\cite{kivlichan2017bounding,su2021fault}, and enter the amplitudes prepared
for Hamiltonian simulation~\cite{babbush2018encoding}.

These classical techniques cannot simply
be ported into coherent quantum circuits. Data-dependent choices must be
compiled into a fixed gate sequence, and all temporary state must ultimately
be returned to $\ket{0}$. The latter reversibility constraint is particularly
costly for numerical routines because many classical methods proceed through
a long sequence of large, high-precision intermediate values. An ordinary
implementation can overwrite each value and reuse the same storage as soon as
the next value is formed. A reversible implementation must instead retain
enough information to reconstruct the overwritten value, or recompute it when
clearing the circuit. Consequently, composing individually efficient routines
can accumulate many large live registers and may require running the complete
computation backward after its result has been consumed. Iterative Newton-style
methods repeatedly replace one full-precision estimate with another, so their
reversible implementations carry large intermediate registers through every
step~\cite{bhaskar2016scientific}. Floating-point representations add exponent bits
to every such register to track a wide range of
magnitudes~\cite{haner2018floating}. The approximation, fixed-point
representation, and uncomputation schedule therefore become inseparable parts
of circuit design~\cite{haner2018optimizing}.

We consider three numerical subroutines:
Harrow--Hassidim--Lloyd (HHL) rotation-angle evaluation ($\arcsin$),
Coulomb-potential evaluation ($1/\sqrt{\|x\|^2}$), and Gaussian
evaluation ($e^{-x^2}$)~\cite{harrow2009linear,kivlichan2017bounding,babbush2016exponentially}. Each
implements an additive fixed-point map
$\ket{x}\ket{k}\mapsto\ket{x}\ket{k+\widetilde f(x)}$, with a uniform error
bound over its promised domain~\footnote{Throughout this section,
\texttt{round} denotes nearest-integer rounding, with half-integer ties away
from zero. Because each result is accumulated modulo $2^n$ in an $n$-bit
register, we measure the error between $a$ and $b$ by the circular distance
$\min_{m\in\mathbb Z}|a-b+m2^n|$, so that values that differ only by
wraparound count as close. Behavior outside the promised domain is
unrestricted.}. The additive form serves both common uses of the result. An
ordinary accumulator $k$ holds the value for later bitwise-controlled
operations, while a phase-gradient accumulator converts the addition directly
into a phase.

Our baselines use the piecewise-polynomial construction of H\"aner,
Roetteler, and Svore (HRS)~\cite{haner2018optimizing}. HRS partitions the
domain, fits a low-degree polynomial on each interval, and coherently loads
the selected polynomial's coefficients. It evaluates the polynomial in general fixed-point arithmetic by Horner's rule, which uses one
multiplication per degree. It adds the
result to the output and
reverses the computation to clear its temporary registers. Its cost has four
sources: locating an interval with arbitrary boundaries; the wide
multiplications of the Horner polynomial-evaluation stage; the registers that hold
intermediate products and the polynomial value; and the reverse pass that
erases the loaded coefficients. For each task we
tune the polynomial degree, fixed-point format, and published recomputation
trade-offs, and widen each interval until it reaches the error bound. We then
use the complete gate-level implementation with the lowest $TQ$-cost.

The three discovered circuits still implement piecewise-polynomial
approximations. They keep the HRS approximation family but reorganize it
for reversible execution, addressing
the four costs in turn. They choose finer
partitions whose intervals can be identified from the input bits. No comparison
is then needed, and a polynomial of lower degree on fewer input bits suffices:
the HHL and Gaussian circuits replace cubics by quadratics, and the Coulomb
circuit replaces quadratics by lines.
Where the intervals differ in width, as in the HHL circuit, a rearrangement of
the input bits maps them into one shared local coordinate, so that one
arithmetic kernel serves every interval. This is the conditional routing of the
main text: a cheap, interval-dependent rearrangement of the data surrounds one
shared kernel. The arithmetic accumulates directly into its destination,
avoiding registers for full products and for the function value.
MBU (Methods) then clears the loaded coefficients
without reversing the lookup.

One supporting technique makes these choices safe. An offline search fits the
coefficients to the exact quantized and truncated integer computation that
the circuit performs, rather than rounding an ideal real-valued approximation
afterward. Coefficient widths, omitted product columns, shifts, truncation,
and rounding therefore share one end-to-end error budget.

All three circuits use these techniques. They differ in how the input reaches
the arithmetic. The HHL circuit, which we describe in full, uses intervals of
unequal width and cyclically shifts the input bits into a shared coordinate. The Coulomb
circuit first builds the squared distance and normalizes it. The Gaussian
circuit folds its input onto the nonnegative axis, prunes its lookup, and
reuses coefficient qubits.

\subsection{HHL rotation angle}\label{sec:case-hhl-angle}

After phase estimation, HHL applies an eigenvalue-dependent rotation whose
angle is $\theta=\arcsin(C/\widetilde\lambda)$; conditioning on the rotated
ancilla reweights each eigencomponent by
$1/\widetilde\lambda$~\cite{harrow2009linear}. We take $C=1/2$ and represent
the eigenvalue estimate by a 21-bit integer $X$, with
$x=X/2^{16}\in[1,32)$. The circuit adds the 33-bit angle
\begin{equation}
  f(X)=\mathtt{round}\!\left(
    2^{32}\arcsin\!\left(\frac{2^{15}}{X}\right)
  \right)
  \label{eq:hhl-fixed-angle}
\end{equation}
to an accumulator $k$. We require error below $10^{-7}$ radians throughout the domain.

The tuned HRS baseline uses cubic fits on 35 intervals and general 44-bit
fixed-point arithmetic. It meets the error bound with 16,000 Toffolis and 300
qubits. Locating the interval requires coherent comparisons against arbitrary
boundaries, and each evaluation performs three multiplications on 44-bit numbers.

The discovered circuit restructures this dataflow with three
transformations, illustrated in \cref{fig:hhl-spine}.

\paragraph{Power-of-two-aligned intervals.}
The function $\arcsin(1/(2x))$ changes rapidly near $x=1$ and progressively
more slowly as $x$ increases. An arbitrary adaptive partition could follow
this variation, but reversibly locating one of its intervals would require
comparisons against general boundaries. The discovered circuit instead uses
narrow intervals near $x=1$ and wider ones where the function is flatter,
while restricting every interval to the power-of-two-aligned, or
\emph{dyadic}, form $[L,L+2^b)$ with $L$ a multiple of $2^b$. We call such an
interval a \emph{leaf}, because these intervals are the leaves of a binary
tree over the input bits. The leading bits of the input then determine the
leaf.

Because selection is cheap, a much finer partition is affordable: the circuit
uses 345 intervals of widths $2^8$ through $2^{15}$, whereas the baseline uses 35
arbitrary intervals. The finer partition simplifies the arithmetic in two ways.
A quadratic suffices instead of the baseline cubic, reducing the number of
multiplications, and each polynomial depends only on the narrow local
displacement $\delta=X-L$, reducing its operand widths. A uniform partition
at the finest width would require 7,936 intervals
(\cref{fig:hhl-spine}(b)), so the dyadic tree obtains these benefits without
paying for uniformly fine resolution.

\paragraph{Shared local coordinate.}
After selecting a leaf, a conventional polynomial-evaluation stage would form $\delta=X-L$ and
rescale it according to the leaf width, requiring arithmetic and temporary
storage. Moreover, leaves of different widths initially present differently
scaled inputs to the polynomial-evaluation stage. The circuit replaces this
coordinate arithmetic with an in-place cyclic shift of the input bits. For
$X=L+\delta$ in a width-$2^b$ leaf, cyclically shifting the low 15 bits
produces
\begin{equation}
  d=2^{15-b}\delta+c_L.
  \label{eq:hhl-shared-coordinate}
\end{equation}
The varying bits $\delta$ move to a common scale, while the fixed leaf bits
wrap around to form $c_L$. Because $c_L$ is constant within the leaf, the
offline fit absorbs it into the coefficients. Thus every leaf feeds the same
15-bit polynomial-evaluation stage without subtraction, arithmetic rescaling, or a coordinate
register (\cref{fig:hhl-spine}(b)). For example, in a width-$2^{12}$ leaf,
the low 15 bits consist of three bits fixed by the leaf followed by the 12
varying bits of $\delta$. Shifting those bits cyclically left by three places produces
$d=2^3\delta+c_L$.

\begin{figure}[t]
\centering
\includegraphics[width=\linewidth]{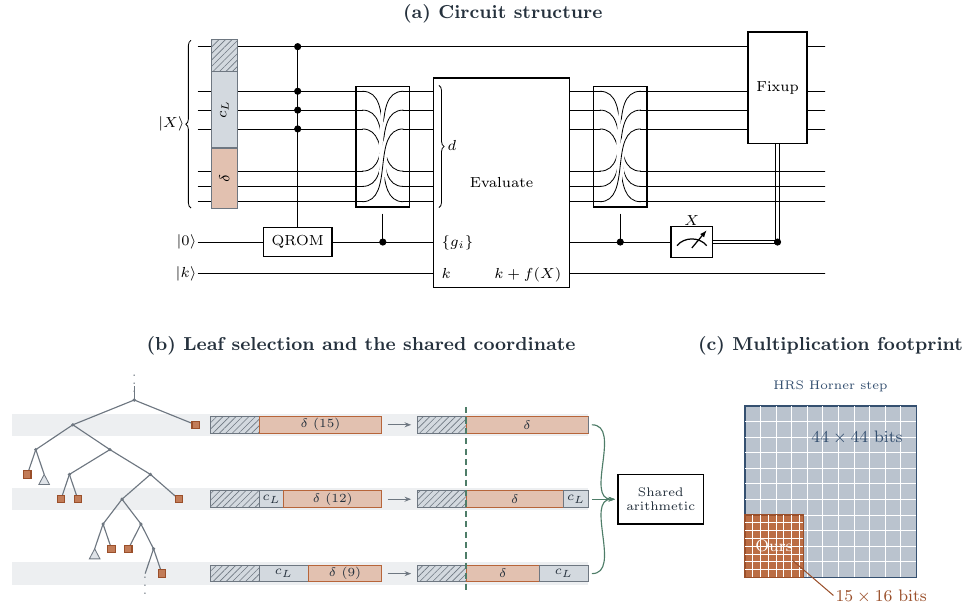}
\caption{Co-designed HHL rotation-angle circuit. (a) Complete reversible
dataflow: a dyadic lookup selects a local quadratic and a cyclic shift, the input is
mapped into a shared local coordinate, and the quadratic is accumulated
directly into $k$. The inverse shift restores $X$; measurement and a shared
phase correction clear the consumed coefficients without reversing the lookup. (b)
Dyadic leaves and the shared coordinate. Power-of-two-aligned
leaves are selected by the leading input bits. Shifting each leaf's varying bits
$\delta$ to a common scale allows leaves of different widths to feed one
15-bit polynomial-evaluation stage without subtraction or arithmetic rescaling. (c) Arithmetic
consequence: the circuit multiplies the 15-bit coordinate by a 16-bit slope,
compared with $44\times44$ bits for each multiplication in the general HRS Horner
polynomial-evaluation stage.}
\label{fig:hhl-spine}
\end{figure}

\paragraph{Narrow direct accumulation.}
For each leaf, the coefficient table defines the integer quadratic
\begin{equation}
  q_L(d)=g_0-d(g_1-g_2d),
  \label{eq:hhl-integer-quadratic}
\end{equation}
with intercept $g_0$, slope $g_1$, and curvature $g_2$.
A general Horner polynomial-evaluation stage multiplies large fixed-point numbers, retains work
registers for those products, and materializes the polynomial value before
adding it to the output. The shared polynomial-evaluation stage instead consumes the selected
coefficients in place. It first overwrites the loaded slope $g_1$ with the
inner Horner value, then accumulates the outer product directly into $k$:
\begin{equation}
  t\gets g_1-g_2d
  \quad\text{(stored over $g_1$)},
  \qquad
  k\gets k+g_0-dt.
  \label{eq:hhl-inplace-update}
\end{equation}
No separate register ever holds the completed polynomial or either full
product.

The products themselves are also truncated: the circuit omits the low-order
partial-product columns whose combined effect stays inside the error budget.
The offline search fits the integer coefficients $g_i$ to the exact updates in
\cref{eq:hhl-inplace-update}, including the truncated columns.

The worst-case error over all inputs is $9.95\times10^{-8}$ radians. The circuit uses
1,348 Toffolis and 115 qubits, for a $TQ$-cost of 160,000. Relative to the tuned HRS
baseline, this reduces the Toffoli count by $12\times$, the qubit count by
$2.6\times$, and the $TQ$-cost by $32\times$.

\subsection{Coulomb potential}\label{sec:case-coulomb}

Real-space simulation repeatedly applies phases proportional to the inverse
distance between particle pairs~\cite{kivlichan2017bounding,su2021fault}. For
distinct points $r_i,r_j\in\{0,\ldots,255\}^3$, our circuit adds
\begin{equation}
  F(r_i,r_j)=\mathtt{round}\!\left(
    \frac{2^{24}}{
      \sqrt{\sum_{c\in\{x,y,z\}}(r_{i,c}-r_{j,c})^2}}
  \right)
\end{equation}
to a 25-bit accumulator $k$. We require error below
$2\times10^{-6}$ in the inverse distance. The tuned HRS
baseline normalizes the squared distance and evaluates quadratics over 16
intervals, using 6,355 Toffolis and 252 qubits. The discovered
circuit instead builds the squared distance in a single accumulation, replaces the quadratics by a
finer table of lines, and erases that larger lookup cheaply.

\paragraph{Direct squared-distance accumulation.}
The inverse-square-root approximation begins from the exact squared distance
$d^2=\sum_c(r_{i,c}-r_{j,c})^2$. A conventional implementation would form
the three coordinate squares in separate product registers and then add them.
The discovered circuit instead adds the partial products of all three
squares directly into one accumulator. This yields $d^2$ exactly, without three product
registers or a separate summation stage.

\paragraph{A finer table of lines.}
The central approximation trade-off is to spend more lookup entries in order
to use a lower-degree polynomial. HRS uses 16 intervals and evaluates a quadratic in each. The discovered
circuit uses 128 intervals and evaluates only a line in each: the larger lookup is
cheaper than the additional reversible multiplication required by a
quadratic.

Both constructions first write $d^2=2^e u$, with mantissa $1\le u<2$, so that
\begin{equation}
  \frac{1}{\sqrt{d^2}}=2^{-e/2}\frac{1}{\sqrt{u}}.
\end{equation}
A few flag qubits record the scale $2^{-e/2}$. The parity of $e$ and the six
leading bits of the mantissa select one of 128 lines, and the next six bits
give a local coordinate. The remaining low bits are irrelevant within the
error bound.
The selected line therefore needs at most six controlled additions rather
than a general multiplication, after which the circuit restores the scale and
adds the result to $k$. As in the HHL circuit, the coefficients are fitted to
the exact integer computation.

The larger table is affordable because MBU of a lookup~\cite{berry2019qubitization}
removes it with 17 Toffolis, instead of the 126 that reversing it
would take.

The worst-case error over all attainable squared distances is
$1.91\times10^{-6}$ in the inverse distance. The
circuit uses 1,140 Toffolis and 142 qubits, for a $TQ$-cost of 160,000.
Relative to the tuned HRS baseline, this reduces the Toffoli count by $5.6\times$,
the qubit count by $1.8\times$, and the $TQ$-cost by $9.9\times$.

\subsection{Fixed-point Gaussian}\label{sec:case-gaussian}

The Gaussian $e^{-x^2}$ enters chemistry simulation through on-the-fly
evaluation of molecular basis functions~\cite{babbush2016exponentially} and
Gaussian state preparation~\cite{mcardle2022state}. The input is a 39-bit
two's-complement integer $X_s$ representing $x=X_s/2^{32}$, and the promised
domain is $x\in[-8,8)$. The
circuit adds
\begin{equation}
  f(X_s)=\mathtt{round}\!\bigl(2^{32}e^{-x^2}\bigr)
  \label{eq:gaussian-fixed-value}
\end{equation}
to a 39-bit accumulator with error below $10^{-7}$. The tuned HRS baseline first
forms $x^2$ and evaluates cubics over 32 intervals using general 40-bit
arithmetic, for 24,000 Toffolis and 362 qubits.

The discovered circuit never forms $x^2$. It folds the input onto the
nonnegative axis and approximates $e^{-y^2}$ directly. For a negative
two's-complement input, forming $|x|$ would complement the data bits and add
one. The circuit omits that addition, so the complement, which needs only
Clifford gates, instead represents $y=|x|-2^{-32}$. This removes an entire carry chain, and the resulting
one-unit offset lies within the error budget.

\paragraph{Pruned coefficient lookup.}
Ten bits of $y$ directly select one of 1,024 equal intervals of width
$2^{-7}$ that cover $[0,8)$, while the lower bits give the local coordinate. The Gaussian's decay divides
these intervals into three regimes: 455 require quadratics, the next
59 need only constants, and the remaining 510 return zero. A conventional
QROM would nevertheless spend gates on every interval. The variable-spaced QROM
of \citet{sanders2020compilation} instead stops whenever an entire subtree
shares one value. Pruning the constant and zero intervals reduces the lookup
from 1,022 Toffolis to 495.

\paragraph{Direct quadratic accumulation.}
On a quadratic interval, the lookup returns an intercept, slope, and curvature. A
truncated product updates the slope register in place, after which the circuit
adds the intercept and accumulates the final product directly into $k$.
Consequently, it never forms $y^2$ or a completed product, and it needs no separate
Gaussian-value register. The coefficients are again fitted to the exact
integer computation.

After MBU, the freed coefficient qubits serve as
carry ancillas for the later arithmetic~\cite{berry2019qubitization}. Together
with direct accumulation, this reuse lowers the peak qubit count from 362 in the
baseline to 139.

The worst-case error over the full domain is $9.99\times10^{-8}$. The circuit uses 1,583 Toffolis and 139 qubits, for a
$TQ$-cost of 220,000. Relative to the tuned HRS baseline, this reduces the Toffoli count
by $15\times$, the qubit count by $2.6\times$, and the $TQ$-cost by $39\times$.

Across the three tasks, these techniques reduce $TQ$-cost by
$9.9\times$--$39\times$ relative to the tuned HRS baselines.

\section{Quantum chemistry simulation}
\label{sec:case-alias-fe2s2}

In Pauli-based electronic-structure simulation, the electronic Hamiltonian is
a weighted sum of Pauli operators:
\begin{equation}
  H=\sum_{\ell=0}^{L-1} w_\ell P_\ell,
  \label{eq:alias-hamiltonian-sum}
\end{equation}
where each $P_\ell$ is a tensor product of $I$, $X$, $Y$, and $Z$. Qubitization encodes
this Hamiltonian in a unitary circuit. Such an encoding is needed because $H$ is not
itself unitary and therefore cannot simply be applied to a quantum state.
Writing $\lambda=\sum_\ell|w_\ell|$, an ideal \textsc{Prepare} creates a superposition of
term labels,
\begin{equation}
  \textsc{Prepare}\ket{0}
  =\sum_\ell\sqrt{|w_\ell|/\lambda}\ket{\ell},
  \label{eq:ideal-prepare-state}
\end{equation}
while \textsc{Select} applies the signed Pauli operator with that label:
\begin{equation}
  \textsc{Select}\ket{\ell}\ket{\psi}
  =\ket{\ell}\operatorname{sgn}(w_\ell)P_\ell\ket{\psi}.
\end{equation}
Preparing this superposition, applying the selected term, and reversing the
preparation embeds the normalized Hamiltonian as a block of a larger unitary:
\begin{equation}
  (\bra{0}\otimes I)
  \textsc{Prepare}^{\dagger}\textsc{Select}\textsc{Prepare}
  (\ket{0}\otimes I)
  =\sum_\ell\frac{w_\ell}{\lambda}P_\ell
  =\frac{H}{\lambda}.
  \label{eq:qubitization-block-encoding}
\end{equation}
Qubitization combines this block encoding with a reflection about the prepared
state to form a quantum walk operator $U$ with eigenvalues $e^{\pm i\arccos(E/\lambda)}$, where $E$ ranges over the eigenenergies of $H$. Phase estimation on $U$ therefore yields the energies; because it applies $U$ many times, the cost of one application sets the cost of the simulation.

We use the $\mathrm{Fe_2S_2}$ integral data published
by \citet{li2017spinprojected}, which models 30 active electrons in 20 spatial
orbitals. Iron--sulfur clusters are important strongly correlated targets, in
which many electronic configurations contribute significantly to the
ground state~\cite{sharma2014iron,low2025spectrum}. Each spatial orbital has a
spin-up and a spin-down mode, so these 20 orbitals become 40 spin-orbitals and,
after the Jordan--Wigner mapping, 40 qubits. Discarding terms with $|w_\ell|<10^{-4}$ Ha, following
the truncation practice of sparse qubitization~\cite{berry2019qubitization}, leaves
$L=41{,}052$ signed Pauli operators.

State-of-the-art resource estimates for this system use double-factorized
tensor hypercontraction (DFTHC), which compresses the Hamiltonian into factors
and orbital rotations before circuit construction~\cite{low2025spectrum}.
However, studying the DFTHC walk operator would
require validating networks of continuously parameterized Givens rotations
acting on superposed states, which goes beyond the \kickmix{} framework. We
therefore focus on the more direct sparse compilation, which stores the
Hamiltonian as an explicit list of its Pauli terms.

In this section we study the data-dependent part of one application of $U$: the alias
selection inside \textsc{Prepare}, the \textsc{Select} application, and the
erasure of the selection. Directly
preparing the unequal amplitudes in \cref{eq:ideal-prepare-state} is expensive,
so a standard construction uses the alias
method~\cite{walker1977alias,vose1991alias,babbush2018encoding,berry2019qubitization}.
Classical preprocessing spreads the probabilities $|w_\ell|/\lambda$ over $L$
equal-sized buckets, one per term. Bucket $i$ is shared between term $i$ and
one other term $\operatorname{alt}(i)$, and a 16-bit threshold
$\operatorname{keep}(i)$ sets the split. The circuit draws a bucket $i$ and a
16-bit value $\tau$ uniformly. It selects term $i$ if
$\tau<\operatorname{keep}(i)$ and term $\operatorname{alt}(i)$ otherwise:
\begin{equation}
  h(i,\tau)=\mathbf{1}[\tau<\operatorname{keep}(i)],
  \qquad
  \operatorname{sel}(i,\tau)=
  \begin{cases}
    i & \text{for } h=1,\\
    \operatorname{alt}(i) & \text{for } h=0.
  \end{cases}
  \label{eq:walk-select-definition}
\end{equation}
The preprocessing chooses the alternate terms and thresholds so that, across
all pairs $(i,\tau)$, each term $\ell$ is selected a number of times
proportional to its weight, rounded to 16 bits. Thus choosing $i$ and $\tau$
uniformly produces the rounded distribution.
Let $U_{\mathrm{alias}}$ denote the reversible
computation of $h$ and $\operatorname{sel}$ from $i$ and $\tau$. The
data-dependent circuit studied here is
$U_{\mathrm{alias}}^\dagger\textsc{Select}U_{\mathrm{alias}}$: it computes the
selected term, applies its signed Pauli operator, and erases the selection,
leaving $i$ and $\tau$ unchanged. The uniform-superposition preparation of
$i$ and $\tau$, its reversal, and the walk reflection lie outside this circuit.

Each table row stores its threshold, its alternate
index, and the signs of
both possible coefficients. We call this 34-bit record the
\emph{alias word}. Write $\theta_\ell=0$ for a positive coefficient and
$\theta_\ell=1$ for a negative one.

To state the required action on the system register $x$, observe that on a
computational-basis state, $X$ flips a bit, $Z$ contributes a sign determined
by that bit, and $Y$ does both. For term $\ell$, let $a_\ell$ mark the qubits
on which $P_\ell$ contains $X$ or $Y$, and let $z_\ell$ mark those on which it
contains $Z$ or $Y$. Write $n_Y(P_\ell)$ for the number of $Y$ factors in
$P_\ell$. Every Pauli operator in the table contains an even number of $Y$
factors, so its phase is real:
\begin{equation}
  P_\ell\ket{x}
  =(-1)^{\eta_\ell\oplus
      \operatorname{parity}(x\mathbin{\wedge}z_\ell)}
    \ket{x\oplus a_\ell},
  \qquad
  \eta_\ell=\frac{n_Y(P_\ell)\bmod4}{2}.
\end{equation}
The required map is therefore
\begin{equation}
  \ket{i}_{16}\ket{\tau}_{16}\ket{x}_{40}
  \longmapsto
  (-1)^\phi
  \ket{i}_{16}\ket{\tau}_{16}
  \ket{x\oplus a_{\operatorname{sel}}}_{40},
  \label{eq:alias-select-oracle}
\end{equation}
where
\begin{equation}
  \phi=\theta_{\operatorname{sel}}\oplus\eta_{\operatorname{sel}}
  \oplus\operatorname{parity}(x\mathbin{\wedge}z_{\operatorname{sel}}).
\end{equation}
The comparison bit $h$ and the selected index do not appear in this map: they
are intermediate values that the circuit must return to zero.
We require this map exactly for every $i<41{,}052$, every $\tau$, and every
40-bit system state.

The baseline combines the circuits used in sparse
qubitization~\cite{babbush2018encoding,berry2019qubitization}. QROM loads the 34-bit alias word, merging repeated alias words~\cite{sanders2020compilation}. A comparator
then calculates $h(i, \tau)$ to choose between $i$ and $\operatorname{alt}(i)$
and apply the selected coefficient sign. A second, independent traversal over
$\operatorname{sel}$ applies the selected Pauli operator to the 40 system
qubits. The baseline then erases the selected index, the comparison bit $h$, and
the alias word.

The lookup and the row-by-row \textsc{Select} traversal cost 13,000 and 41,000
Toffolis, respectively; together they contribute more than $97\%$ of the
total. The remaining comparison, selection, and cleanup bring the
baseline to 55,000 Toffolis and 139 qubits. The main opportunity is
to shorten the two table traversals.

For arbitrary, unstructured table data, QROM has no pattern to exploit: in the
worst case, every address may store an unrelated word. However, this table
is not arbitrary. Its Pauli operators and coefficients come from a physical
system with exact spin symmetries and repeated orbital environments, so many
coefficients are equal and many Pauli operators are related. We therefore
choose the table to suit the circuits. Two degrees of
freedom make this possible. First, permuting terms together with their
coefficients and Pauli operators leaves
\cref{eq:alias-hamiltonian-sum} unchanged. Second, many different alias tables
realize the same rounded distribution. The preprocessing uses the first freedom
to group Pauli operators related by fixed Clifford transformations into
blocks, shortening \textsc{Select}. It uses the second to place equal
alias words into repeated patterns, shortening QROM.

\paragraph{Frame-block \textsc{Select}.}
Many Pauli operators in this Hamiltonian are related by fixed Clifford
transformations, which we call \emph{frames}. Suppose, for example, that two commuting frames
generate four related Pauli operators from one representative. Instead of
selecting four operators independently, the circuit selects the representative
once and uses two bits of the address (the selected row index) to
control the two frames. More
generally, $k$ such bits can identify up to $2^k$ operators in a block. The
circuit pays for one block representative plus its controlled
frames, rather than every Pauli operator separately.

The circuit applies a frame by conjugation. It applies $\sigma^\dagger$ to the system
register under a control, then the representative $P$, then $\sigma$ under the same control.
The result is $\sigma P\sigma^\dagger$ when the control is set and $P$
otherwise. Each frame is applied once for the whole table rather than once per
block. Blocks that use the same frames are contiguous, so a few comparisons of
the address with fixed boundaries determine which frames are active.

These relations have a physical origin: the model contains two
equivalent spin channels, two similar iron sites, and corresponding orbitals
on repeated parts of the cluster. These repetitions create Pauli operators that
are related by a fixed transformation. Each frame is a permutation or change of basis of the
spin-orbitals. The Jordan--Wigner mapping turns it into a fixed Clifford
transformation of the 40 system qubits.

Two of the frames used here come from exact spin symmetries of the Hamiltonian.
Because the model treats the spin-up and spin-down channels
identically, $\sigma_{\alpha\beta}$ exchanges the spin-up and spin-down modes
of every spatial orbital. Spin-rotation symmetry produces a second pairing
after the Jordan--Wigner mapping: $\sigma_{XY}$ exchanges the $X$ and
$Y$ factors and adjusts the accompanying $Z$ strings required by that
mapping. Both transformations preserve the corresponding coefficient
magnitudes.

The other five frames come from parts of the cluster that resemble each
other; unlike the two spin frames, they are not symmetries of the full
Hamiltonian. The frame $\sigma_{\mathrm{Fe}}$ is the approximate left--right
reflection of the cluster, which exchanges the two iron centers and their
surrounding orbitals. The frame $\sigma'_{\mathrm{Fe}}$ combines this
reflection with several local orbital swaps and matches Pauli operators that
the bare reflection misses. The frames $\sigma_{\mathrm{thio}}$,
$\sigma_{3d}$, and $\sigma_{\mu\mathrm{S}}$ exchange corresponding orbitals on
the terminal thiolates, the iron $3d$ shells, and the bridging sulfur atoms.
These sites have similar but not identical electronic environments: the
permutation often maps one term's Pauli operator onto another's, but the
coefficients differ. Only the Pauli relation $P_{\ell'}=\sigma P_\ell\sigma^\dagger$ must hold exactly; $w_{\ell'}$
and $w_\ell$ remain independent entries in the alias table.

The preprocessing thus finds a commuting family of seven frames and orders related
Pauli operators into blocks whose sizes are powers of two, so that address bits
index the operators within a block. The 41,052 terms form 3,221 blocks. Selecting the block representatives costs 3,327
Toffolis. Applying the seven frames and determining which are active costs 534.
\Cref{fig:fe2s2-blocks} shows the resulting structure.

\begin{figure}[t]
\centering
\includegraphics[width=\linewidth]{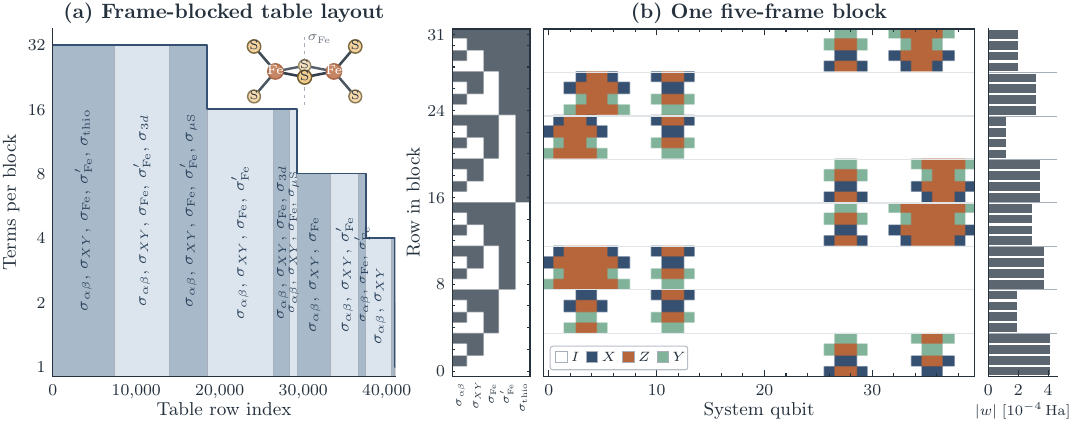}
\caption{Frame structure of the $\mathrm{Fe_2S_2}$ alias table.
(a) The preprocessing places Pauli operators related by fixed Clifford frames into
contiguous blocks. This groups the 41,052 table rows into 3,221 blocks that
\textsc{Select} handles as units. Band height gives the number of terms in
each block, and the labels give the frames used to generate them.
Inset: the $\mathrm{Fe_2S_2}$ cluster core (thiolate substituents
omitted); the dashed line marks the iron-site reflection
$\sigma_{\mathrm{Fe}}$.
(b) One block of 32 Pauli operators generated from a representative operator
using five frames. Each row of the left matrix specifies which frames are
applied; a dark cell means that the frame named below its column is included.
The center shows the resulting Pauli operator on the 40 system qubits, and the
bar on the right gives its coefficient magnitude. The spin frames
$\sigma_{\alpha\beta}$ and $\sigma_{XY}$, shown in the first two columns,
preserve coefficient magnitudes and therefore produce equal-weight groups of
four. The other three frames need not preserve the coefficients.}
\label{fig:fe2s2-blocks}
\end{figure}

\paragraph{Alias table with repeated rows.}
Recall that each alias-table row is a probability bucket split between its own
term and one alternate term. A term whose weight is smaller than one bucket
leaves empty space in its row; a term whose weight is larger than one bucket
fills such space in other rows.

A standard alias construction may pair underfilled rows with heavy terms in
many different ways. Our preprocessing chooses these pairings to make
nearby rows identical. For example, if four light rows of equal weight must draw their
alternate terms from two heavy terms $A$ and $B$, an ordinary schedule might
alternate $A,B,A,B$; ours prefers $A,A,B,B$. Rows 1--2 then share a threshold, an alternate term, and hence an
alias word, as do rows 3--4.

This strategy is particularly effective here because the two exact spin
frames produce many groups of four equal weights (\cref{fig:fe2s2-blocks}(b)). The coefficients also span a broad range, so heavier
terms can supply the missing weight of many lighter rows. Because the pairing only redistributes each term's weight among
buckets, it preserves the rounded distribution and creates long
runs of identical alias words without breaking the frame blocks. Of the 20,374
$\sigma_{\alpha\beta}$-related row pairs it places next to each other, 18,463 store
identical alias words.

The alias-table QROM realizes this saving by pruning the binary tree over its address $i$.
Whenever the two halves of a subtree contain the same sequence of alias words,
the QROM visits one copy and skips the address bit that distinguishes them. For example, the
eight-row pattern $A,A,B,B,A,A,B,B$ consists of two identical halves, so the
top address bit is irrelevant and the circuit looks up only $A,A,B,B$.
Together, this scheduling and repeated-subtree pruning reduce the cost of the alias-table QROM to
10,000 Toffolis, compared with 41,000 for a row-by-row traversal and 13,000 for the
baseline lookup with merged repeated alias words.

The rest of the circuit follows the baseline. A comparator
computes $h$, a controlled swap moves the selected index into place, and both
are reversed after \textsc{Select}, for 64 Toffolis. The circuit then clears
the $\operatorname{keep}$ and $\operatorname{alt}$ fields by MBU and corrects
the resulting phase in one pass over the address $i$, for 782 Toffolis.

We test every table row and a further $2^{24}$ uniformly random inputs. The discovered circuit uses
15,000 Toffolis and 126 qubits, for a $TQ$-cost of $1.9\times10^6$. Relative to the
baseline, this reduces the Toffoli count
by $3.7\times$, the qubit count by $9.4\%$, and the $TQ$-cost by
$4.1\times$.

The amount of compression is specific to this Hamiltonian, but the preprocessing
can search for such relations in other sparse Pauli tables, and it falls
back to conventional \textsc{Select} and QROM costs when it finds none. Exact spin
symmetry is common in spin-independent molecular Hamiltonians, so the scheduling
strategy should carry over. The opportunity is smaller when a Zeeman term or
spin--orbit coupling breaks the spin symmetries. In DFTHC, analogous choices
arise in the factorization
itself~\cite{lee2021hypercontraction,low2025spectrum}: the factorization is
not unique, and a circuit-aware choice could make its coefficient data
repeated or sparse and therefore cheaper to load.

\section{Nonabelian lattice-gauge simulation}\label{sec:case-nonabelian}

In Hamiltonian lattice gauge theory, space is replaced by a grid. A gauge-field
configuration assigns a group element $U_i(\mathbf x)$ to each directed edge,
or link, from site $\mathbf x$ in direction $i$. This element describes how the
gauge field transports a state between neighboring sites; traversing the link
backward contributes its inverse $U_i(\mathbf x)^{-1}$. Multiplying the links
around a closed path therefore describes the net transport along
that loop. Although this product depends on an arbitrary convention at each site (the
gauge), the real part of its matrix trace does not, and these loop traces
form the magnetic energy~\cite{lamm2019gauge}. Because the gauge group is
nonabelian, its elements need not commute, so the links must be multiplied in
the order in which the path traverses them.

We consider the group-register formulation of Hamiltonian lattice gauge
theory~\cite{lamm2019gauge}, in which a register holds the group element of
each link. The Hamiltonian separates into an
electric (kinetic) term $H_E$ and a magnetic (potential) term $H_B$. A Trotter
step evolves these terms separately. Here we implement the magnetic evolution
$e^{-i\Delta t H_B}$ for the improved Hamiltonian of
\citet{carena2022improved}. The smallest loop contributing to $H_B$ goes around
one square face of the lattice and is called a \emph{plaquette}; the standard
magnetic term sums a trace over these four-link loops. The improved
Hamiltonian adds weighted $2\times1$ rectangles, which cancel the leading
discretization error. Thus the
magnetic evolution requires the same basic operation at many locations:
follow a short closed path, multiply its links in order, extract a
trace, and add the weighted result to the phase. The
compilation task studied here is one application of the magnetic evolution to
the whole lattice, which we call the \emph{magnetic layer}.

We use a periodic
$4\times4\times4$ lattice~\footnote{The resource estimate of
\citet{osorio2025sigma72} uses a $10^3$ lattice. In this work, we use $4^3$ for ease of evaluation.}. The lattice has 192 links, and the layer contains
192 four-link plaquettes and 384 six-link rectangles. For directions $i$ and
$j$, representative loop products are
\begin{equation}
\begin{gathered}
  P_{ij}(\mathbf x)
  =U_i(\mathbf x)U_j(\mathbf x+\hat i)
    U_i(\mathbf x+\hat j)^{-1}U_j(\mathbf x)^{-1},\\
  R_{ij}(\mathbf x)
  =U_i(\mathbf x)U_i(\mathbf x+\hat i)
    U_j(\mathbf x+2\hat i)
    U_i(\mathbf x+\hat i+\hat j)^{-1}
    U_i(\mathbf x+\hat j)^{-1}U_j(\mathbf x)^{-1}.
\end{gathered}
\label{eq:gauge-loops}
\end{equation}
Here $\hat i$ is one lattice step in direction $i$, with coordinates wrapping
modulo four. Plaquettes use the pairs $i<j$, and rectangles use both
orderings $i\ne j$.

To store each link in a finite register, we restrict its value from the
continuous group $SU(3)$ to the 216-element subgroup
$G=\Sigma(72\times3)\subset SU(3)$ used by
\citet{osorio2025sigma72}. We encode these elements using the unique normal
form
\begin{equation}
  g=\omega^p C^qE^rV^{2s+t}X^u,
  \qquad
  p,q,r\in\{0,1,2\},\qquad s,t,u\in\{0,1\},
  \label{eq:nonabelian-normal-form}
\end{equation}
where $\omega=e^{2\pi i/3}$ and $C,E,V,X$ are fixed matrices that generate the
group. These six exponents are the coordinates of a group element. Each three-valued
exponent uses two qubits and each binary exponent one, for nine qubits per
link and 1,728 in total. We write $g$ for the value of one link and
$g_0,\ldots,g_{191}$ for all of them. Setting $u=0$ in
\cref{eq:nonabelian-normal-form} gives the 108-element subgroup
$\Sigma(36\times3)$, while $u=1$ gives its other coset.

The main difficulty in multiplying two group elements is that the generators do not commute. For instance, within the $\omega,C,E$ part, the relation
$EC=\omega CE$ gives
\begin{equation}
  (\omega^PC^QE^R)(\omega^pC^qE^r)
  =\omega^{P+p+Rq}C^{Q+q}E^{R+r},
  \label{eq:gauge-heisenberg-product}
\end{equation}
with exponents modulo three. Moving $V$ through the $C,E$ part
similarly transforms several coordinates. The action of $X$ is more involved:
\begin{equation}
  XC=C^2EX,\qquad XE=\omega^2CEX,\qquad
  XV=\omega^2EV^3X,\qquad X^2=C^2V^2.
  \label{eq:gauge-x-relations}
\end{equation}
Consequently, when the left factor contains $X$, every coordinate of the right
factor must be rewritten before the two can be combined.
\Cref{alg:gauge-multiplication-comparison}(a) gives the general rule.

Each $g\in G$ can be written as a $3 \times 3$ matrix $M_g$, and each completed
loop product contributes to $H_B$ only through the seven-valued function
\begin{equation}
  f(g)=2\operatorname{Re}\Tr(M_g)
  \in\{-3,-2,-1,0,1,2,6\}.
  \label{eq:gauge-trace-class}
\end{equation}
Let $S_P$ and $S_R$ be the sums of $f$ over all plaquettes and rectangles, respectively.
For the improved discretization, $H_B$ is proportional to
$S_P-\frac{1}{20}S_R$~\cite{carena2022improved}. We represent these phase
weights as 51-bit integers $a_P$ and $a_R$. The weight $a_P$ depends on the time step and the
coupling, so we use a generic 51-bit value, and we choose $a_R$ so that
$a_R/a_P\approx-1/20$. Because all of the loop traces are diagonal in the
group-element basis, the magnetic evolution leaves each link configuration
unchanged. In our 51-bit representation, the target operation multiplies each
link configuration by the phase
$e^{2\pi i(a_P S_P+a_R S_R)/2^{51}}$. We implement this operation using a
reusable phase-gradient state, which satisfies
\begin{equation}
  \ket{\mathrm{PG}_{51}}
  =\frac{1}{\sqrt{2^{51}}}\sum_{k=0}^{2^{51}-1}
    e^{-2\pi i k/2^{51}}\ket{k},
  \qquad
  \operatorname{ADD}_y\ket{\mathrm{PG}_{51}}
  =e^{2\pi i y/2^{51}}\ket{\mathrm{PG}_{51}},
  \label{eq:gauge-phase-gradient}
\end{equation}
where $\operatorname{ADD}_y$ denotes modular addition by $y$. Taking
$y=a_P S_P+a_R S_R$ therefore applies the desired phase to the link registers
while leaving the phase-gradient state unchanged. This reduces the
compilation task to constructing the underlying reversible addition
\begin{equation}
  \ket{g_0\cdots g_{191}}\ket{k}
  \longmapsto
  \ket{g_0\cdots g_{191}}
  \ket{k+a_P S_P+a_R S_R\bmod 2^{51}},
  \label{eq:gauge-layer-map}
\end{equation}
for arbitrary computational-basis values $k$.

The baseline treats all 576 loops independently, combining the inversion
and multiplication circuits of \citet{osorio2025sigma72}, whose inner
$\Sigma(36\times3)$ circuits come from \citet{gustafson2024sigma36}; we
implement these circuits from the published descriptions. For a loop $L$, let $a(L)=a_P$ for a plaquette and $a(L)=a_R$ for a
rectangle. For each loop, the baseline
starts with $W=1$ and successively multiplies the four or six directed links
into it, so $W$ is the running product of the links traversed so far. After
completing the loop, the baseline looks up the 51-bit addend $a(L)f(W)$ by
QROM~\cite{babbush2018encoding}, adds it to $k$, and reverses
the lookup and loop product. Each multiplication costs 146 Toffolis and each link inversion costs 87.
A plaquette costs 1,702 Toffolis and a rectangle costs 2,460, so the 192
plaquettes and 384 rectangles cost 1.3 million Toffolis. The baseline peaks at
1,927 qubits.

The discovered circuit combines three optimizations. First, it uses different
factor orders for the links and the running product, which makes updates of $W$
cheap. Second, it
reuses partial loop products across the layer instead of computing every loop
from the identity. Third, it avoids QROM by reading the
trace value directly from the coordinates of $W$ and using the result to
control the addition of a known constant.

\begin{algorithm}[t]
\caption{General normal-form multiplication and the specialized running-product update.
The former supports two arbitrary operands in the original normal form; the
latter uses different factor orders for the running product and the incoming link,
stores $\widetilde P=P+QR$, and writes $m=2s+t$.}
\label{alg:gauge-multiplication-comparison}
\footnotesize
\begin{minipage}[t]{0.485\linewidth}
\centering
\textbf{(a) General normal-form product}\par
\vspace{4pt}
\makeatletter\def\theHALG@line{gaugeprod.a.\arabic{ALG@line}}\makeatother
\begin{algorithmic}[1]
\Statex \textbf{Input:} running product $W=\omega^PC^QE^RV^MX^\upsilon$
\Statex \hspace{\algorithmicindent} link
  $g=\omega^pC^qE^rV^mX^u$
\Statex \textbf{Output:} $W\gets Wg$ in the original normal form
\If{$\upsilon=1$} \Comment{apply \cref{eq:gauge-x-relations}}
  \State Write $m=2s+t$, with $s,t\in\{0,1\}$
  \State $p_0\gets p+qr+2\cdot\mathbf{1}[q=2]
    +2\cdot\mathbf{1}[r\ne0]\pmod3$
  \State $(q_0,r_0)\gets(2q+r,q+r)\pmod3$
  \State $p\gets p_0+2s+sr_0+2t\pmod3$
  \State $(q,r)\gets(q_0+s,r_0+s+t(1+s))\pmod3$
  \State $s\gets s\mathbin\oplus t$
  \State $m\gets2s+t$
  \If{$u=1$} \Comment{apply $X^2=C^2V^2$}
    \State $p\gets p+(1-t)(2-s)r\pmod3$
    \State $q\gets q+(1-t)(2-s)\pmod3$
    \State $r\gets r+t(2-s)\pmod3$
    \State $s\gets s\mathbin\oplus1$
    \State $m\gets2s+t$
  \EndIf
\EndIf
\If{$M=0$} \Comment{commute $C^qE^r$ through $V^M$}
  \State $(\delta p,q',r')\gets(0,q,r)$
\ElsIf{$M=1$}
  \State $(\delta p,q',r')\gets(2qr,-r,q)$
\ElsIf{$M=2$}
  \State $(\delta p,q',r')\gets(0,-q,-r)$
\Else
  \State $(\delta p,q',r')\gets(2qr,r,-q)$
\EndIf
\State $P\gets P+p+\delta p+Rq'\pmod3$
\State $(Q,R)\gets(Q+q',R+r')\pmod3$
\State $M\gets M+m\pmod4$
\State $\upsilon\gets\upsilon\mathbin\oplus u$
\end{algorithmic}
\end{minipage}
\hfill
\begin{minipage}[t]{0.485\linewidth}
\centering
\textbf{(b) Specialized running-product update}\par
\vspace{4pt}
\makeatletter\def\theHALG@line{gaugeprod.b.\arabic{ALG@line}}\makeatother
\begin{algorithmic}[1]
\Statex \textbf{Input:} running product
  $W=\omega^PV^MX_o^\upsilon C^QE^R$
\Statex \hspace{\algorithmicindent} link
  $g=\omega^pC^qE^rV^{2s+t}X_o^u$
\Statex \textbf{Output:} $W\gets Wg$ in the specialized representation
\State $\widetilde P\gets
  \widetilde P+p+qr+Qr+2Rq\pmod3$
\State $(Q,R)\gets(Q+q,R+r)\pmod3$
\If{$t=1$}
  \State $(Q,R)\gets(R,-Q)\pmod3$
\EndIf
\If{$s=1$}
  \State $(Q,R)\gets(-Q,-R)\pmod3$
\EndIf
\If{$u=1$}
  \State $(Q,R)\gets(Q-R,-Q-R)\pmod3$
\EndIf
\State $M\gets M+(-1)^\upsilon(2s+t)+2\upsilon u\pmod4$
\State $\upsilon\gets\upsilon\mathbin\oplus u$
\end{algorithmic}
\end{minipage}
\end{algorithm}

\paragraph{Specialized running-product representation.}
Repeatedly multiplying the running product $W$ by links dominates the cost,
so the circuit uses a representation tailored to the restricted update
$W\gets Wg$ rather than to arbitrary group multiplication. Multiplication in
the original normal form is costly in part because $X$ has a complicated
conjugation action. However, the normal form can use any element of the
$u=1$ coset in place of $X$. We replace $X$ by $X_o=CEX$. Because $X_o$ differs from $X$ only by
factors of $C$ and $E$, rewriting a link in terms of $X_o$ changes only
$p,q,r$. $X_o$ is the unique such element that satisfies
\begin{equation}
  X_oVX_o^{-1}=V^{-1},
  \qquad X_o^2=V^2.
\end{equation}
Compare these relations with \cref{eq:gauge-x-relations}: the $V,X_o$ part now
closes on itself instead of spilling corrections into the other coordinates.

After making this replacement, the circuit uses different factor orders for
the links and for $W$. Each link is recoded as
\begin{equation}
  g=\omega^p C^qE^rV^mX_o^u,
\end{equation}
whereas the frequently updated product is stored in the opposite order,
\begin{equation}
  W=\omega^P V^M X_o^\upsilon C^Q E^R,
  \label{eq:gauge-internal-normal-form}
\end{equation}
with $p,q,r,P,Q,R\in\mathbb Z_3$, $m,M\in\mathbb Z_4$, and
$u,\upsilon\in\mathbb Z_2$. This places the incoming and accumulated $C,E$
factors next to one another in $Wg$, so they can be combined directly before
a few fixed permutations account for the remaining factors. One complication
remains: reordering the factors would ordinarily introduce recurring
corrections to the coordinate $P$. The circuit stores
$\widetilde P=P+QR\pmod3$, folding those corrections into the coordinate
itself rather than recomputing them during every multiplication.

The 192 links are recoded once
before the first loop and restored after the last. Each recoded link
remains in a compact nine-qubit encoding; storing every link one-hot would cost
three additional qubits per link. In contrast, there is only one running
product, so we can afford to store its three-valued coordinates $\widetilde P$, $Q$,
and $R$ one-hot (three qubits, exactly one set), increasing the size of $W$ from nine to twelve qubits. Adding one
modulo three then becomes a cyclic permutation of three qubits, while negation
becomes a single swap. Controlled additions and the fixed transformations
above can therefore use controlled qubit permutations rather than arithmetic
on the two-qubit encoding. Together, these
representation choices reduce each update of $W$ from 146
Toffolis in the baseline's general multiplication circuit to 35 Toffolis.
Multiplying by an inverse link $g^{-1}$ costs the same, so the circuit needs
no separate link inversion.
\Cref{alg:gauge-multiplication-comparison} gives the complete coordinate
rules for both representations.

\paragraph{Reusing loop products.}
The baseline constructs each loop product from the identity and then reverses
the construction. Neighboring loops, however, often share a sequence of
factors. The discovered circuit therefore retains $W$ after one trace
evaluation and edits it into the product needed for the next, reusing the
shared factors instead of rebuilding them.

More loops share factors than first appears, because the
real trace is invariant under cyclic rotation and inversion:
\begin{equation}
  f(h_1h_2\cdots h_n)
  =f(h_2\cdots h_nh_1)
  =f((h_1h_2\cdots h_n)^{-1}).
  \label{eq:gauge-trace-invariance}
\end{equation}
The first equality is cyclicity of the matrix trace; the second follows
because these matrices are unitary and $f$ takes the real part.
For example, to change $abcde$ into $abcxy$, the circuit right-multiplies by
$e^{-1}$ and $d^{-1}$ to cancel the old suffix, then appends $x$ and $y$,
retaining the common prefix $abc$. Inverting $W$, which costs 21 Toffolis, makes the other end of the
product available for the same operation. Cyclic rotations and inversion also
give several equivalent ways to write each loop, which we call its
\emph{representatives}, all with the same value of $f$. As \cref{fig:plaquette-layer}(a) illustrates, a good
choice can expose a longer segment shared with the next loop. A classical dynamic program chooses the representatives and the order in
which the loops are evaluated, minimizing the
cost of transforming $W$ between consecutive representatives. The chosen schedule uses 3,136 link
multiplications and 480 inversions of $W$ for the 576 loops, about 5.4
multiplications per loop. By comparison, building and reversing each loop separately takes 6 multiplications
for a plaquette and 10 for a rectangle.

\begin{figure}[t]
\centering
\includegraphics[width=\linewidth]{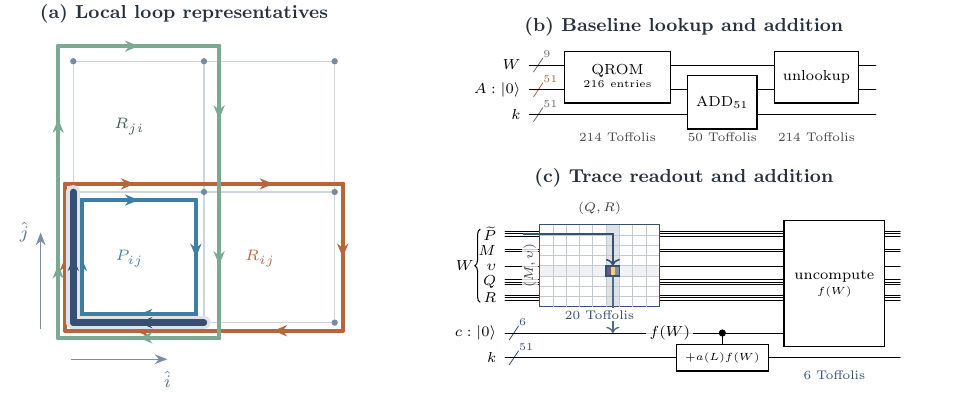}
\caption{Loop geometry and trace accumulation in the nonabelian lattice-gauge
circuit. (a) The blue, orange, and green paths are cyclic or inverted
representatives of $P_{ij}$, $R_{ij}$, and $R_{ji}$ whose trace values equal
those of \cref{eq:gauge-loops} by \cref{eq:gauge-trace-invariance}; the dark
segment is shared after these symmetry choices. (b) The baseline lookup,
modular addition, and unlookup for a completed loop product.
(c) The trace readout, drawn as a truth table over the coordinates
$(\widetilde P,M,\upsilon,Q,R)$ of the specialized representation, followed by
the controlled addition. Rows encode $(M,\upsilon)$, columns encode $(Q,R)$,
and the highlighted cell and stripes show how $(M,\upsilon,Q,R)$ and then
$\widetilde P$ determine the trace value.}
\label{fig:plaquette-layer}
\end{figure}

\paragraph{Reading the trace directly.}
The baseline in \cref{fig:plaquette-layer}(b) uses $W$ as the address of a
216-entry QROM and loads a full 51-bit addend, although $f(W)$ has only seven
possible values. The discovered circuit instead computes which of the seven
values occurs. In the specialized representation, a small Boolean circuit acting directly on
$(\widetilde P,M,\upsilon,Q,R)$ computes six flag qubits, one per nonzero trace
value, of which at most one is set (\cref{fig:plaquette-layer}(c)).
Each flag controls the addition of the corresponding constant $a(L)v$ to
$k$, where $v$ is the flagged trace value, so the circuit never stores a 51-bit lookup result. Computing the six flags costs 20 Toffolis.
Readout, addition, and cleanup cost 76 Toffolis for a
plaquette and 75 for a rectangle, compared with 478 for the baseline QROM lookup,
addition, and unlookup.

We test the complete circuit on 100,000 random inputs and each group
operation on all of its inputs. The discovered
circuit uses 170,000 Toffolis and 1,859 qubits, for a $TQ$-cost of $3.2\times10^8$.
Relative to the baseline, this reduces the Toffoli count by
$7.4\times$, the qubit count by $3.5\%$, and the $TQ$-cost by $7.7\times$.
Updates of $W$ account for 110,000 of the Toffolis, trace
readout and addition for 43,000, inversions of $W$ for 10,000, and
link recoding for 8,832.

\clearpage

\appendix
\phantomsection
\addcontentsline{toc}{section}{Appendix}
\addtocontents{toc}{\protect\maintocdepth{0}}
\begin{center}{\LARGE\bfseries Appendix}\end{center}
\medskip
\startcontents[appendices]
\bigskip
\crefalias{section}{appendix}
\crefalias{subsection}{appendix}
\numberwithin{equation}{section}
\renewcommand{\theequation}{\thesection\arabic{equation}}

\section{Exact branch recovery for divstep inversion}
\label{app:divstep-recovery}

This appendix proves the decoding rule of
\cref{eq:divstep-decoding-rule}~\cite{gu2026lean}. The argument first establishes the invariants carried by
the coefficients $f_v$ and $g_v$, then identifies the decisions already visible from the
updated state, and finally shows how the remaining bit $b$ is encoded in
a window of low-order bits. Let $p$ be an odd positive integer and let
$0\le w<p$.
The walk begins at
\begin{equation}
  (z,f,g,f_v,g_v)=(0,p,w,0,1).
  \label{eq:app-divstep-seed}
\end{equation}
Every step follows the recurrence
\begin{equation}
  (z',f',g',f_v',g_v')=
  \begin{cases}
    (-z,\ g,\ (g-f)/2,\ 2g_v,\ g_v-f_v) & \text{for } z\ge0,\ g\text{ odd},\\
    (z+1,\ f,\ (g+f)/2,\ 2f_v,\ g_v+f_v) & \text{for } z<0,\ g\text{ odd},\\
    (z+1,\ f,\ g/2,\ 2f_v,\ g_v) & \text{for } g\text{ even}.
  \end{cases}
  \label{eq:app-divstep-recurrence}
\end{equation}
Let $s$ indicate whether the step swaps, and let $b=1$ exactly when $g$
is odd, so $b$ records whether the two values are combined before halving. In this notation,
$s=1$ selects the first case of \cref{eq:app-divstep-recurrence}; when
$s=0$, the value of $b$ selects the second or third case. The walk
freezes when $g$ first becomes zero; the steps before that are \emph{active}.

For an integer $u$ and $k\ge0$, write $u[k]\in\{0,1\}$ for bit $k$ of
$u\bmod 2^{k+1}$, counted from zero. This notation agrees with the usual
binary digit for $u\ge0$ and with the two's-complement digit for $u<0$.
We also write $\operatorname{bit}_k(u)=u[k]$. Let $v_2(u)$ denote the
number of trailing zero bits of $u$ (its $2$-adic valuation), with $v_2(0)=+\infty$.

We first show, via the determinant and scaling relations, that the integer recurrence
computes a modular inverse.

\begin{lemma}[Basic invariants]
\label[lemma]{lem:app-divstep-basic}
Every reachable state satisfies
\begin{equation}
  f g_v-g f_v=p.
  \label{eq:app-divstep-det}
\end{equation}
Moreover, $f$ and $g_v$ are odd and $f_v$ is even.
\end{lemma}

\begin{proof}
The initial state gives $p\cdot1-w\cdot0=p$. Direct substitution in the nonswap and
swap cases gives
\begin{equation}
\begin{aligned}
  f(g_v+b f_v)-\frac{g+b f}{2}(2f_v)&=f g_v-g f_v,
    &&\text{(nonswap)},\\
  g(g_v-f_v)-\frac{g-f}{2}(2g_v)&=f g_v-g f_v,
    &&\text{(swap)}.
\end{aligned}
\end{equation}
Hence every active step preserves the determinant, while a frozen step leaves
the state unchanged. The parity claims also hold at the initial state. On a nonswap,
$f$ remains odd, $f_v'=2f_v$ is even, and
$g_v'=g_v+b f_v$ is odd. On a swap, the swap condition ($g$ odd) makes
$f'=g$ odd, while $f_v'=2g_v$ is even and
$g_v'=g_v-f_v$ is odd. Induction proves all four claims.
\end{proof}

\begin{proposition}[Scaling and termination]
\label[proposition]{prop:app-divstep-scale}
After $t$ active steps,
\begin{equation}
  f_vw\equiv2^t f\pmod p,
  \qquad
  g_vw\equiv2^t g\pmod p.
  \label{eq:app-divstep-scale}
\end{equation}
If $\gcd(p,w)=1$, then at termination $f=\pm1$, $g_v=fp$, and
\begin{equation}
  w^{-1}\equiv f f_v2^{-t}\pmod p.
  \label{eq:app-divstep-inverse}
\end{equation}
\end{proposition}

\begin{proof}
The congruences hold at the initial state. Substitution in
\cref{eq:app-divstep-recurrence} shows that if they hold after $t$ active
steps, the updated coefficients satisfy them with $t+1$ in place of $t$. They
therefore hold for every $t$. Every step also preserves $\gcd(f,g)$, so termination with
$g=0$ gives $|f|=\gcd(p,w)=1$. The determinant invariant then gives
$g_v=fp$, and multiplying the first congruence by $f$ yields
\cref{eq:app-divstep-inverse}.
\end{proof}

Before turning to the decoding rule, we record the independent worst-case
termination guarantee used by the circuit analysis.

\begin{proposition}[Worst-case step bound]
\label[proposition]{prop:app-divstep-bound}
Let $p=2^{256}-2^{32}-977$. For every
$0\le w<p$, the walk in
\cref{eq:app-divstep-seed} reaches $g=0$ within $590$ active steps.
If $w>0$, its final coefficients satisfy \cref{eq:app-divstep-inverse}; if $w=0$,
the walk starts frozen.
\end{proposition}

\begin{proof}
The worst-case analysis of \citet{wuille2021safegcd} guarantees that every
$0<w<p$ reaches $g=0$ within $590$ divsteps.
The case $w=0$ starts with $g=0$. For $0<w<p$, $\gcd(p,w)=1$ because $p$ is prime,
so once the walk terminates, the remaining claims follow from \cref{prop:app-divstep-scale}.
\end{proof}

We now recover the bits $s$ and $b$ of an active step from its updated state.
The bit $s$, and $b$ in two special cases, can be read from a few bits of that state.

\begin{lemma}[Locally visible decisions]
\label[lemma]{lem:app-divstep-local}
Every active step has the following properties:
\begin{enumerate}
\item\label{local:app-divstep-swap}
  The swap is visible from the updated coefficient:
  \begin{equation}
    s=1
    \quad\Longleftrightarrow\quad
    f_v'\equiv2\pmod4
    \quad\Longleftrightarrow\quad
    v_2(f_v')=1.
    \label{eq:app-divstep-swap}
  \end{equation}
  If $s=0$, then $4\mid f_v'$; moreover, $s=1$ implies $b=1$.
\item\label{local:app-divstep-counter}
  If $s=0$ and $z'\ge1$, then $b=0$.
\item\label{local:app-divstep-freeze}
  If $g'=0$, then $b=1$.
\end{enumerate}
\end{lemma}

\begin{proof}
On a swap step, the incoming $g$ is odd, so
$b=1$; the update also gives $f'=g$, preserving the oddness of $f$. Moreover,
$f_v'=2g_v$, which is twice an odd integer by
\cref{lem:app-divstep-basic}. On a nonswap step,
$f_v'=2f_v$, which is divisible by four. This proves
part~\labelcref{local:app-divstep-swap}. For
part~\labelcref{local:app-divstep-counter}, a nonswap has
$z=z'-1\ge0$. Since the swap condition nevertheless failed, the incoming
value $g$ was even, and hence $b=0$. Finally, on a nonswap,
$g'=(g+b f)/2$. If $b=0$ and $g'=0$, then $g=0$, contrary to the step being
active; a swap already has $b=1$ by
part~\labelcref{local:app-divstep-swap}. This proves
part~\labelcref{local:app-divstep-freeze}.
\end{proof}

The steps these tests leave undecided, apart from a terminating step, are the nonswaps entered with $z<0$;
we call them \emph{negative-counter nonswaps}. Their decisions are the \emph{ambiguous}
decisions of SI~\ref*{sec:modinv-divsteps}: the preceding tests do not determine $b$, but the
coefficients still retain it. Whenever $f_v\ne0$, put
$j=v_2(f_v)$ and write $f_v=2^j q$ with $q$ odd. On an updated state, $j$ is the
excursion depth defined in the main text, and a swap followed by a run of ambiguous
decisions is an \emph{excursion}. The $j$ low-order bits of $g_vq^{-1}$, which we call the
\emph{window}, accumulate successive decisions.

\begin{lemma}[Window invariant]
\label[lemma]{lem:app-divstep-window}
Every reachable state satisfies the following conditions:
\begin{enumerate}
\item\label{inv:app-divstep-high} If $z\ge0$, then $2^{z+1}\mid f_v$.
\item\label{inv:app-divstep-low} If $z<0$, then $f_v\ne0$ and, with
  $j=v_2(f_v)$ and $q=f_v/2^j$, there is an integer $c$ such that
  $0\le c<2^j$ and
  \begin{equation}
    g_v\equiv q c\pmod{2^{j-z}}.
    \label{eq:app-divstep-window}
  \end{equation}
\end{enumerate}
\end{lemma}

Because $q$ is odd, part~\labelcref{inv:app-divstep-low} says that
$g_vq^{-1}\bmod2^{j-z}$ fits in $j$ bits: the $-z$ bits above position
$j-1$ vanish.

\begin{proof}
Proceed by induction along the walk. A frozen step changes nothing. At the
initial state, $z=0$, $f_v=0$, and $2\mid f_v$.

Suppose first that a nonswap begins with $z\ge0$. Then
$z'=z+1$ and $f_v'=2f_v$, so
$2^{z'+1}\mid f_v'$.

Suppose next that a swap occurs at $z=D\ge0$. Then $z'=-D$ and
$f_v'=2g_v$. By \cref{lem:app-divstep-basic}, $g_v$ is odd, so
$v_2(f_v')=1$ and the odd part of $f_v'$ is $g_v$. Moreover,
$g_v'=g_v-f_v$, while the induction hypothesis gives
$2^{D+1}\mid f_v$. Hence
\begin{equation}
  g_v'\equiv g_v\pmod{2^{D+1}}.
\end{equation}
If $D>0$, part~\labelcref{inv:app-divstep-low} holds with $j=1$ and
$c=1$; if $D=0$, part~\labelcref{inv:app-divstep-high} asks only that
$2\mid f_v'$.

Finally, suppose that a nonswap begins with $z<0$. Let $j=v_2(f_v)$ and
let $c$ be the witness from part~\labelcref{inv:app-divstep-low}. Then
$z'=z+1$, $f_v'=2f_v$, and the odd part $q$ is unchanged. Since
$g_v'=g_v+b f_v$,
\begin{equation}
  g_v'\equiv q(c+b2^j)\pmod{2^{j-z}},
  \qquad 0\le c+b2^j<2^{j+1}.
\end{equation}
The required modulus after the step is unchanged because
$2^{(j+1)-z'}=2^{j-z}$. If $z'<0$, this proves
part~\labelcref{inv:app-divstep-low}; if $z'=0$, then
part~\labelcref{inv:app-divstep-high} asks only that $2\mid f_v'$.
\end{proof}

The invariant makes the bit $b$ of every negative-counter nonswap an exact
function of the updated state.

\begin{theorem}[Decoding rule for an ambiguous decision]
\label[theorem]{thm:app-divstep-ladder}
Let $(z',f',g',f_v',g_v')$ be the updated state of a
negative-counter nonswap. Put $j=v_2(f_v')$ and $q=f_v'/2^j$. Then
$j\ge2$, $q$ is odd, and
\begin{equation}
  b=\operatorname{bit}_{j-1}\!\left(g_v'q^{-1}\bmod2^j\right)
   =\operatorname{bit}_{j-1}\!\left(
       p(f'q)^{-1}\bmod2^j\right).
  \label{eq:app-divstep-ladder}
\end{equation}
In particular, for fixed odd $p$, the bit $b$ is determined by
$(j,f'\bmod2^j,q\bmod2^j)$.
\end{theorem}

\begin{proof}
The step doubles $f_v$, so its incoming valuation is $j-1$. The window
invariant makes the incoming $f_v$ nonzero, while
\cref{lem:app-divstep-basic} makes it even; hence $j\ge2$. By definition,
$q$ is odd. The modulus provided by part~\labelcref{inv:app-divstep-low} of
\cref{lem:app-divstep-window} is $2^{(j-1)-z}$, which is divisible by
$2^j$. The invariant therefore supplies a $c$ with $0\le c<2^{j-1}$ and
$g_v\equiv q c\pmod{2^j}$. The odd part $q$ is unchanged, and the update
gives
\begin{equation}
  g_v'\equiv q(c+b2^{j-1})\pmod{2^j}.
\end{equation}
Since $q$ is odd and $0\le c+b2^{j-1}<2^j$, bit $j-1$ of
$g_v'q^{-1}\bmod2^j$ is exactly $b$. For the second form, reduce the
determinant identity
$f'g_v'-g'f_v'=p$ modulo $2^j$. Because $2^j\mid f_v'$ and $f'$
is odd,
$g_v'\equiv p(f')^{-1}\pmod{2^j}$. Substitution into the first form
gives the second.
\end{proof}

The proof also shows how the window records the decisions.
Immediately after the swap that begins an excursion, the window
contains $1$; each subsequent negative-counter nonswap places its bit $b$
in the next higher position. Thus, at depth $j$, the window contains
\begin{equation}
  1+\sum_{k=1}^{j-1}b_k2^k,
  \label{eq:app-divstep-shift-register}
\end{equation}
where $b_k$ is the bit $b$ of the $k$th negative-counter nonswap after the
swap.

\begin{corollary}[Finite dependence]
\label[corollary]{cor:app-divstep-finite}
For each fixed odd $p$, the bit $b$ at depth $j$ is determined by at most
$2j$ bits of the updated state.
\end{corollary}

\begin{proof}
By \cref{thm:app-divstep-ladder}, the residues of $f'$ and $q$ modulo
$2^j$ determine $b$, and each residue is represented by $j$ bits.
\end{proof}

\begin{corollary}[Depth-two rule]
\label[corollary]{cor:app-divstep-j2}
At every negative-counter nonswap with $j=2$,
\begin{equation}
  b=g_v'[1]\mathbin{\oplus}f_v'[3].
  \label{eq:app-divstep-j2}
\end{equation}
\end{corollary}

\begin{proof}
By \cref{thm:app-divstep-ladder},
$b=\operatorname{bit}_1(g_v'q^{-1}\bmod4)$ with $q=f_v'/4$.
Every odd $q$ satisfies $q^{-1}\equiv q\pmod4$. For odd $u$ and $v$,
$uv\equiv1+2(u[1]\mathbin{\oplus}v[1])\pmod4$, so
$\operatorname{bit}_1(uv\bmod4)=u[1]\mathbin{\oplus}v[1]$. Taking
$u=g_v'$ and $v=q$ gives the result because $q[1]=f_v'[3]$.
\end{proof}

\end{document}